\documentclass[twocolumn,prx,longbibliography,nofootinbib, superscriptaddress]{revtex4-2}
\usepackage[cal=cm]{mathalfa}
\usepackage[dvips]{graphicx} 
\usepackage{amsfonts}
\usepackage{amssymb}
\usepackage{amscd}
\usepackage{amsmath} 
\usepackage{amsthm}
\usepackage{appendix}
\usepackage{bbm}
\usepackage{dsfont}
\usepackage{mathptmx}
\usepackage{enumerate}
\usepackage{enumitem}
\usepackage{upgreek}
\usepackage{epsfig}
\usepackage{float}
\usepackage[colorlinks=true,linkcolor=teal,citecolor=teal,urlcolor=teal]{hyperref}
\usepackage{cleveref}

\crefname{theorem}{theorem}{theorems}
\Crefname{theorem}{Theorem}{Theorems}

\crefmultiformat{theorem}
 {theorems~#2#1#3}
 {,~#2#1#3}
 {,~#2#1#3}
 {,~#2#1#3}

\Crefmultiformat{theorem}
 {Theorems~#2#1#3}
 {,~#2#1#3}
 {,~#2#1#3}
 {,~#2#1#3}

\usepackage{makecell}
\usepackage{times}
\usepackage{pifont}
\usepackage{subfigure}
\usepackage{subfloat}
\usepackage{xcolor}		
\usepackage{physics}
\usepackage[most]{tcolorbox}
\usepackage{tabularx}
\usepackage{MnSymbol}
\usepackage{soul}
\usepackage{booktabs}
\usepackage{amssymb}
\usepackage{wasysym}
\usepackage{tikz}
\usepackage{tikzpeople}
\usepackage{pgfplots}
\usepackage{pifont}
\usepackage[normalem]{ulem}
\usetikzlibrary{quantikz}
\usetikzlibrary{arrows}
\usetikzlibrary{shapes,fadings,snakes}
\usetikzlibrary{decorations.pathmorphing,patterns}
\usetikzlibrary{calc}
\usetikzlibrary{positioning}
\usepackage[linesnumbered,ruled,vlined]{algorithm2e}

\newtheorem{theorem}{Theorem}

\newtheorem{lemma}{Lemma}
\newtheorem{corollary}{Corollary}

\newtheorem{definition}{Definition}
\newtheorem{proposition}{Proposition}

\newcommand{\bI}{\mathbb{I}}

\newcommand{\cO}{\mathcal{O}}
\newcommand{\cA}{\mathcal{A}}
\newcommand{\cB}{\mathcal{B}}
\newcommand{\cC}{\mathcal{C}}
\newcommand{\cD}{\mathcal{D}}
\newcommand{\cE}{\mathcal{E}}

\newcommand{\cG}{\mathcal{G}}
\newcommand{\cH}{\mathcal{H}}
\newcommand{\cI}{\mathcal{I}}
\newcommand{\cJ}{\mathcal{J}}
\newcommand{\cL}{\mathcal{L}}
\newcommand{\cM}{\mathcal{M}}
\newcommand{\cN}{\mathcal{N}}

\newcommand{\cR}{\mathcal{R}}

\newcommand{\dF}{d_{\mathrm F}}

\DeclareDocumentCommand\dket{ s m }{%
 \IfBooleanTF{#1}
 {\lvert #2\rangle\!\rangle}
 {\left\lvert #2\middle\rangle\!\middle\rangle\right.}
}

\DeclareDocumentCommand\dbra{ s m }{%
 \IfBooleanTF{#1}
 {\langle\!\langle #2\rvert}
 {\left.\middle\langle\!\middle\langle #2\right\rvert}
}

\DeclareDocumentCommand\dketbra{ s m g }{%
 \IfBooleanTF{#1}
 {%
  \IfNoValueTF{#3}
  {\lvert #2\rangle\!\rangle\langle\!\langle #2\rvert}
  {\lvert #2\rangle\!\rangle\langle\!\langle #3\rvert}
 }
 {%
  \IfNoValueTF{#3}
  {\left\lvert #2
   \middle\rangle\!\middle\rangle
   \middle\langle\!\middle\langle #2
   \right\rvert}
  {\left\lvert #2
   \middle\rangle\!\middle\rangle
   \middle\langle\!\middle\langle #3
   \right\rvert}
 }
}

\DeclareDocumentCommand\dbraket{ s m m }{%
 \IfBooleanTF{#1}
 {\langle\!\langle #2\mid #3\rangle\!\rangle}
 {\left.
  \middle\langle\!\middle\langle #2
  \middle| #3
  \middle\rangle\!\middle\rangle
  \right.}
}

\usepackage{phfcc}
\phfMakeCommentingCommand[initials={PhF}]{phf}

\usepackage[most]{tcolorbox}
\newtcolorbox[auto counter]{mybox}[2][]{
	enhanced,
	breakable,
	colback=blue!5!white,
	colframe=blue!75!black,
	fonttitle=\bfseries,
	title=Box \thetcbcounter: #2,#1
}

\makeatletter

\newif\ifappendixtocrecord
\appendixtocrecordfalse

\newcommand{\appendixtableofcontents}{%
 \section*{Contents of Appendices}%
 \setcounter{tocdepth}{3}%
 \@starttoc{atoc}%
 \appendixtocrecordtrue
}

\let\oldsection\section
\let\oldsubsection\subsection
\let\oldsubsubsection\subsubsection

\renewcommand{\section}{%
 \@ifstar{\app@sectionstar}{\@ifnextchar[{\app@sectionopt}{\app@sectionnoopt}}%
}

\newcommand{\app@sectionstar}[1]{%
 \oldsection*{#1}%
}

\newcommand{\app@sectionnoopt}[1]{%
 \oldsection{#1}%
 \ifappendixtocrecord
 \addcontentsline{atoc}{section}{\protect\numberline{\thesection}#1}%
 \fi
}

\def\app@sectionopt[#1]#2{%
 \oldsection[#1]{#2}%
 \ifappendixtocrecord
 \addcontentsline{atoc}{section}{\protect\numberline{\thesection}#1}%
 \fi
}

\renewcommand{\subsection}{%
 \@ifstar{\app@subsectionstar}{\@ifnextchar[{\app@subsectionopt}{\app@subsectionnoopt}}%
}

\newcommand{\app@subsectionstar}[1]{%
 \oldsubsection*{#1}%
}

\newcommand{\app@subsectionnoopt}[1]{%
 \oldsubsection{#1}%
 \ifappendixtocrecord
 \addcontentsline{atoc}{subsection}{\protect\numberline{\thesubsection}#1}%
 \fi
}

\def\app@subsectionopt[#1]#2{%
 \oldsubsection[#1]{#2}%
 \ifappendixtocrecord
 \addcontentsline{atoc}{subsection}{\protect\numberline{\thesubsection}#1}%
 \fi
}

\renewcommand{\subsubsection}{%
 \@ifstar{\app@subsubsectionstar}{\@ifnextchar[{\app@subsubsectionopt}{\app@subsubsectionnoopt}}%
}

\newcommand{\app@subsubsectionstar}[1]{%
 \oldsubsubsection*{#1}%
}

\newcommand{\app@subsubsectionnoopt}[1]{%
 \oldsubsubsection{#1}%
 \ifappendixtocrecord
 \addcontentsline{atoc}{subsubsection}{\protect\numberline{\thesubsubsection}#1}%
 \fi
}

\def\app@subsubsectionopt[#1]#2{%
 \oldsubsubsection[#1]{#2}%
 \ifappendixtocrecord
 \addcontentsline{atoc}{subsubsection}{\protect\numberline{\thesubsubsection}#1}%
 \fi
}

\makeatother

\begin{document}
\title{Fermionic quantum error correction is never free}

\author{Yifan Tang}
\thanks{\href{mailto:yifta@zedat.fu-berlin.de}{yifta@zedat.fu-berlin.de}}
\affiliation{Dahlem Center for Complex Quantum Systems, Freie Universität Berlin, 14195 Berlin, Germany}

\author{Ingo Roth}
\affiliation{Quantum Research Center, Technology Innovation Institute (TII), Abu Dhabi, United Arab Emirates}

\author{Philippe Faist}
\affiliation{Dahlem Center for Complex Quantum Systems, Freie Universität Berlin, 14195 Berlin, Germany}

\author{Zi-Wen Liu}
\thanks{\href{mailto:zwliu0@tsinghua.edu.cn}{zwliu0@tsinghua.edu.cn}}
\affiliation{Yau Mathematical Sciences Center, Tsinghua University, Beijing 100084, China} 

\author{Jens Eisert}
\thanks{\href{mailto:jense@zedat.fu-berlin.de}{jense@zedat.fu-berlin.de}}
\affiliation{Dahlem Center for Complex Quantum Systems, Freie Universität Berlin, 14195 Berlin, Germany}
\affiliation{Helmholtz-Zentrum Berlin für Materialien und Energie, 14109 Berlin, Germany}

\author{Zhenhuan Liu}
\thanks{\href{mailto:qubithuan@gmail.com}{qubithuan@gmail.com}}
\affiliation{Quantum Research Center, Technology Innovation Institute (TII), Abu Dhabi, United Arab Emirates}


\begin{abstract}
Fermionic platforms offer compelling architectures for quantum computing, ranging from topologically protected Majorana-based qubits to fermionic cold atoms. To achieve scalability, however, they require quantum error correction. 
In this work, we prove that any exact and sufficiently accurate approximate fermionic quantum error correction necessarily requires non-Gaussian operations, beyond the free-fermion regime of quadratic dynamics. This is in sharp contrast to the qubit setting, where the efficiently classically simulable stabilizer operations form the standard framework for quantum error correction. 
Specifically, we show that the logical space of any non-trivial fermionic error-correcting code contains no pure fermionic Gaussian state, utilizing a fundamental incompatibility between fermionic error correction and Wick's theorem.
We further show that the required number of bounded-weight non-Gaussian gates for unitary codeword preparation grows at least linearly with both the code distance and the number of encoded qubits, revealing an intrinsic resource overhead that increases simultaneously with error-protection strength and logical capacity.
Furthermore, we analyze the performance of fermionic Gaussian operations in entanglement distillation, revealing a distinction from their bosonic counterparts.
Our results reveal fundamental difficulties for fermionic error correction from the perspectives of both physical implementation and classical simulation, suggesting connections to fermionic phases of matter and complexity theory.
\end{abstract}

\maketitle

\section{Introduction}
Quantum mechanics enables information processing capabilities beyond the reach of classical physics, motivating sustained efforts to realize universal quantum computation across a broad range of physical platforms. 
Among these, fermionic platforms, particularly those based on Majorana zero modes and fermionic cold atoms, offer a promising route toward intrinsically protected quantum information processing. 
In topological approaches to quantum computing, the nonlocal encoding of information in fermion-parity degrees of freedom provides hardware-level protection against important classes of local errors, making these systems attractive candidates for fault-tolerant quantum computation provided the relevant engineering challenges can be overcome~\cite{kitaev2001UnpairedMajoranaFermions,aasen2016MilestonesMajoranaBasedQuantum,karzig2017ScalableDesignsQuasiparticlepoisoningprotected,litinski2018QuantumComputingMajorana}. Cold-atom platforms offer high levels of control~\cite{DopedFermiHubbard,jordens_mott_2008,GreinerHubbard,CollisionalQuantumGates,bakr2025microscopy,BlochSimulation}, and there is growing evidence that fermionic systems can be simulated more efficiently on fermionic hardware than on qubit-based architectures, thereby 
avoiding additional encoding overheads~\cite{BravyiKitaev,LowDepth,PhysRevA.95.032332,JaschkeEncoding}, providing a second route to fermionic quantum computing.

Achieving scalable quantum computation requires some form of \emph{quantum error correction} (QEC)~\cite{Roads,terhal2015qec,MindTheGaps}. This remains true even for topologically protected fermionic platforms. While topological encoding provides hardware-level protection against important classes of local errors, it cannot completely eliminate the effects of disorder, finite-size corrections, quasiparticle poisoning, thermal excitations, or imperfections in measurement and control. As a result, residual errors inevitably accumulate during computation. To enable arbitrarily long and reliable quantum computations, this intrinsic hardware protection must therefore be complemented by fault-tolerant quantum error correction at the software level~\cite{qx36-4rv1,Litinski2017,bravyi2010MajoranaFermionCodes,vijay2017QuantumErrorCorrection,litinski2018QuantumComputingMajorana}. The same conclusion applies to fermionic cold-atom platforms, where quantum error correction is likewise indispensable for scalable quantum computation.

This raises a basic conceptual and practical question: What physical resources are required to implement fermionic quantum error correction?
Fermionic Gaussian operations arise from quadratic, effectively free-fermion dynamics, encompassing basic processes such as hopping, pairing, and pairwise Majorana braiding. 
They therefore constitute particularly natural 
control primitives in fermionic platforms~\cite{Sarma2015zero,aasen2016MilestonesMajoranaBasedQuantum}; moreover, fermionic 
Gaussian dynamics admit efficient classical simulation~\cite{Valiant2002QuantumCircuits,terhal2002classical, jozsa2008MatchgatesClassicalSimulation}, making Gaussian QEC architectures particularly amenable to large-scale classical benchmarking, optimization, and verification. 
By contrast, non-Gaussian operations require resources beyond free-fermion dynamics, such as controlled many-fermion interactions or multi-Majorana measurements, and generally demand 
an additional layer of experimental control. 
A fully Gaussian scheme for 
fermionic QEC would therefore be especially attractive, combining experimentally natural operations with 
efficient classical simulability.
 
Several previous observations, 
however, suggest that this
may be challenging. Some known encoding constructions for
general Majorana stabilizer codes 
rely on quartic and hence
non-Gaussian operations~\cite{mudassar2024EncodingMajoranaCodes}. Moreover, the canonical quadratic Majorana encoding based on unpaired zero modes has Majorana distance one, despite its nonlocal protection against parity-preserving errors~\cite{bravyi2010MajoranaFermionCodes}.
Relatedly, in bosonic systems, Gaussian
operations are known to be insufficient not only for QEC~\cite{niset2009NoGoTheoremGaussian}, but also for entanglement distillation~\cite{eisert2002distilling,Giedke2002distillation,Fiurasek2002Gaussian,GiedkeSchmidt}. 
These observations point to a broader tension between Gaussianity and QEC, yet they do
not fully exclude the possibility of realizing useful fermionic QEC using Gaussian operations.

In this work, we establish the fundamental incompatibility between fermionic Gaussianity and QEC, proving that the logical space of any fermionic code with Majorana distance $d_{\mathrm F}\geq 3$, namely nontrivial error-correcting capability, contains no pure fermionic Gaussian state. A
simple proof reveals the origin of this general obstruction: a pure fermionic
Gaussian state is completely determined by its covariance matrix, which is also commonly referred to as the fermionic correlation matrix and is specified by weight-two Majorana correlations, whereas a
non-trivial error-correcting code distance requires all such low-weight
observables to be unable to distinguish logical states. Even at distance two, where we construct a code whose entire logical space consists of fermionic Gaussian states, we prove that neither its universal encoding map nor the non-destructive syndrome measurement of the code space can be implemented using Gaussian operations alone.
Together, these results reveal a fundamental limitation of fermionic Gaussian operations for both quantum error correction and error detection.

In order to flesh out this obstruction more substantially, 
we show that the unitary codeword preparation of an $[\![n,k,\dF\geq3]\!]_{\mathrm F}$ code requires $\Omega(d_{\mathrm F}+k)$ bounded-weight non-Gaussian unitaries. 
This establishes a direct resource tradeoff with two central code parameters: the distance $d_{\mathrm F}$, which characterizes robustness against noise, and the number of encoded qubits $k$, which characterizes logical storage capacity. 
We further analyze fermionic Gaussian operations in entanglement distillation, showing that products of local Gaussian channels cannot increase the Einstein–Podolsky–Rosen (EPR) pair fidelity of the states considered, even with classical randomness, whereas local Gaussian measurements with postselection can.
Taken together, our QEC and distillation results reveal fundamental differences in noise processing among fermionic Gaussianity, bosonic Gaussianity, and stabilizerness, despite their close parallels from the perspective of efficient classical simulation~\cite{gottesman1998heisenberg,dehaene2003clifford,aaronson2004improved,bartlett2002efficient,mari2012positive,Valiant2002QuantumCircuits,terhal2002classical,jozsa2008MatchgatesClassicalSimulation,StabilizerPolytope}.

\section{Preliminaries}
We first briefly review the notations for fermionic systems, fermionic
Gaussianity, and fermionic quantum codes used throughout this work
(see Appendix~\ref{app:preliminaries} for further details). For a
system of $n$ fermionic modes, we denote the associated $2n$
\emph{Majorana operators} by $\gamma_j$, $j\in[2n]$. They satisfy
$\gamma_j=\gamma_j^\dagger$, $\gamma_j^2=\mathbb I$, and
$\{\gamma_j,\gamma_k\}=2\delta_{j,k}\mathbb I$. For an ordered subset
$S=\{j_1,\cdots,j_{\abs{S}}\}\subseteq[2n]$, the corresponding
\emph{Majorana product} is
$\gamma_S=\gamma_{j_1}\cdots\gamma_{j_{\abs{S}}}$, with
$\gamma_\emptyset=\mathbb I$, and its \emph{Majorana weight} is
$\operatorname{wt}_{\mathrm F}(\gamma_S)=\abs{S}$. When Hermiticity is
needed, we use
$\hat{\gamma}_S=(-i)^{\abs{S}(\abs{S}-1)/2}\gamma_S$, which has the
same Majorana weight. The operators
$\{\hat{\gamma}_S\}_{S\subseteq[2n]}$ form an orthogonal operator basis
with respect to the Hilbert--Schmidt inner product.

Both pure and mixed fermionic Gaussian states admit an equivalent
characterization through Wick's theorem~\cite{wick1950EvaluationCollisionMatrix}. For
a fermionic state $\rho$, its Majorana covariance matrix $\Gamma(\rho)$ is defined by its entries
\begin{equation}\label{eq:correlation_matrix_main}
\Gamma_{j,k}(\rho)
=
-\frac{i}{2}\Tr([\gamma_j,\gamma_k]\rho)
\end{equation}
for $j,k\in[2n]$. 
A fermionic state is Gaussian precisely when its higher-order Majorana
correlations are completely determined by this two-point correlation
matrix according to Wick's theorem. Importantly, the set of fermionic
Gaussian states is not convex: a classical probabilistic mixture of
Gaussian states need not itself satisfy Wick's theorem and therefore
need not be Gaussian. The Gaussian structure is preserved under the
standard fermionic Gaussian operations, including Gaussian unitary
evolution, tensoring with another Gaussian state, conditioning on a
fine-grained outcome of a Gaussian measurement, and discarding
fermionic modes, but not under arbitrary classical probabilistic
mixing.

A parity-preserving fermionic Gaussian unitary is generated by a
quadratic Majorana Hamiltonian and induces a linear rotation of the
Majorana operators,
\begin{equation}\label{eq:Gaussian_unitary_action_main}
U_R^\dagger\gamma_jU_R
=
\sum_k R_{j,k}\gamma_k ,
\end{equation}
with $R\in\operatorname{SO}(2n)$. We denote by $\mathrm M_n$ the
(generalized) fermionic Gaussian unitaries corresponding to
$R\in\operatorname{O}(2n)$.  
The set of $n$-mode pure fermionic Gaussian states can be written as
\begin{equation}\label{eq:def_pure_Gaussian_state_main}
\mathcal G_{\mathrm{pure}}^{(n)}
=
\left\{
U\ket{0^n}
\middle\mid
U\in\mathrm M_n
\right\},
\end{equation}
each of which is generated by applying a Gaussian unitary on the $n$-mode vacuum state $\ket{0^n}$.
A fermionic Gaussian measurement is one whose fine-grained outcomes can be implemented by adjoining Gaussian auxiliary modes, applying Gaussian
unitaries, and measuring occupations of canonical fermionic modes,
possibly followed by discarding modes
\cite{bravyi2005LagrangianRepresentationFermionica}. 
In particular,
each fine-grained measurement branch maps a Gaussian input to a subnormalized Gaussian state.
Due to the concise mathematical description of fermionic Gaussianity, fermionic non-Gaussianity can therefore serve as a computational resource~\cite{hebenstreit2019AllPureFermionic,oszmaniec2022FermionSamplingRobust},
while circuits composed only of Gaussian states and Gaussian
operations admit efficient classical simulation
\cite{jozsa2008MatchgatesClassicalSimulation,Valiant2002QuantumCircuits,terhal2002classical}.

Finally, we characterize an $[\![n,k,d_{\mathrm F}]\!]_{\mathrm F}$ fermionic quantum code that encodes $k$ logical qubits into $n$ fermionic modes by its
code space $\mathcal C$ and the orthogonal projector $P$ onto it,
$\mathcal C=\operatorname{im}(P)$. A non-trivial code has
$\operatorname{rank}(P)=2^k>1$ and its \emph{Majorana distance} is defined as~\cite{bravyi2010MajoranaFermionCodes}
\begin{equation}\label{eq:def_majorana_distance_main}
d_{\mathrm F}(P)
=
\min
\left\{
\abs{S}\ 
\middle|\ 
\emptyset\neq S\subseteq[2n],\
P\gamma_S P\notin\mathbb CP
\right\}.
\end{equation}
According to the Knill-Laflamme condition, a code with Majorana
distance $d_{\mathrm F}(P)$ detects every Majorana product of weight
smaller than $d_{\mathrm F}(P)$, while the error space spanned by all Majorana
products of weight at most $t$ is correctable whenever
$2t<d_{\mathrm F}(P)$~\cite{knill1997TheoryQuantumErrorcorrecting}.

\section{No fermionic Gaussian QEC}\label{app:no_Gaussian_qec}
We begin with a central observation: the code space of any non-trivial fermionic quantum error-correcting code contains no pure Gaussian state.
In this sense, fermionic non-Gaussianity not only provides computational ``magic'' but is also necessary for encoding any non-trivial fermionic
quantum error-correcting code.

\begin{theorem}[No Gaussian codewords]\label{thm:no-Gaussian-codeword}
Let $P$ be the orthogonal projector onto a fermionic quantum code
$\mathcal{C}=\operatorname{im}(P)$ of $n$ modes with $\operatorname{rank}(P)>1$.
If
$d_{\mathrm F}(P)\geq 3$,
then the code space contains no pure fermionic 
Gaussian state, in that
\begin{equation}
\mathcal{C}\cap\mathcal G_{\mathrm{pure}}^{(n)}=\emptyset.
\end{equation}
\end{theorem}

The full proof of Theorem~\ref{thm:no-Gaussian-codeword} is provided in Appendix~\ref{app:proof_no-Gaussian-qec}, and it can be sketched by using the following two observations.
First, as discussed in the last section, a pure fermionic Gaussian state is completely specified by its two-point Majorana correlations, collected in the covariance matrix in Eq.~\eqref{eq:correlation_matrix_main}.
All higher-order correlations are then fixed by Wick's theorem~\cite{wick1950EvaluationCollisionMatrix}.
Moreover, this characterization is unique: if an arbitrary pure state has the same covariance matrix as a pure fermionic Gaussian state, then the two states coincide up to a global phase~\cite{bittel2025OptimalTraceDistanceBoundsa}.
Second, the Knill--Laflamme condition implies that sufficiently low-weight Majorana observables cannot distinguish states within a quantum code.
More precisely, for a code projector $P$ with Majorana distance $d_{\mathrm F}(P)$, every Majorana product $\gamma_S$ of weight $\abs{S}<d_{\mathrm F}(P)$ satisfies
\begin{equation}
\begin{gathered}
P\gamma_S P=\lambda_S P\\
\Downarrow\\
\bra{\psi}\gamma_S\ket{\psi}=\lambda_S,
\forall\ket{\psi}\in\operatorname{im}(P),\text{ with }\bra{\psi}\ket{\psi}=1.
\label{eq:low_weight_indistinguishability}
\end{gathered}
\end{equation}
Thus, when $d_{\mathrm F}\geq3$, all logical states have identical two-point Majorana correlations and hence the same covariance matrix.
Combining this indistinguishability with the uniqueness of pure Gaussian states leads to the theorem.
At a fundamental level, this obstruction originates from the intrinsically low-weight characterization of fermionic Gaussianity, in sharp contrast to stabilizerness, whose defining stabilizer operators may have arbitrarily high weight.

Since $d_{\mathrm F}=3$ is the minimal distance required to correct arbitrary single-Majorana errors, Theorem~\ref{thm:no-Gaussian-codeword} already implies that exact fermionic quantum error correction cannot be realized entirely within the Gaussian framework.
More generally, following the framework of approximate quantum error correction~\cite{beny2010AQEC,faist2020ContinuousSymmetriesApproximatea}, we show in Appendix~\ref{app:proof_approximate_qec} that the Gaussianity of the code space and its approximate error-correcting capability obey a general quantitative tradeoff, with some specific noise models as examples~\cite{knapp2018ModelingNoiseError,wu2024ErrormitigatedFermionicClassical,greplova2018DegradabilityFermionicGaussian}. Intuitively, a code space with sufficiently large Gaussian fidelity necessarily contains orthogonal codewords whose expectation values differ for a quadratic observable accessible from the complementary environment. This is in tension with approximate correctability, which requires logical information to remain nearly decoupled from the environment. 

In Appendix~\ref{app:proof_distance_two_Gaussian_code} we show that the threshold Majorana distance shown in Theorem~\ref{thm:no-Gaussian-codeword} is tight: there exists a non-trivial code with $d_{\mathrm F}=2$ whose entire pure logical space consists of fermionic Gaussian states.
However, a perhaps counterintuitive point is that the existence of such a distance-two code does not imply that either its encoding or its syndrome measurement can be implemented using Gaussian operations. 
Encoding requires a single universal channel that coherently maps every logical input state to its corresponding physical codeword, rather than a state-dependent preparation procedure. 
Likewise, syndrome extraction requires a non-destructive measurement of the projector onto the entire logical space, rather than separate measurements that distinguish and destroy coherence between individual logical states. We establish the following no-go theorem:

\begin{proposition}[Gaussian obstruction to fermionic encoding and detection]
\label{prop:Gaussian_encoding_measurement_informal}
For any non-trivial fermionic code with $d_{\mathrm F}\geq2$, neither an
exact universal encoding channel nor an exact non-destructive measurement
of the code-space projector can be implemented using fermionic Gaussian
operations, even allowing arbitrary classical probabilistic mixtures
of Gaussian operations.
\end{proposition}

The detailed setting and proof are provided in Appendix~\ref{app:proof_Gaussian_encoding_measurement_informal}. 
Since $d_{\mathrm F}\geq2$ is the minimal requirement for detecting all single-Majorana errors, Proposition~\ref{prop:Gaussian_encoding_measurement_informal} together with Theorem~\ref{thm:no-Gaussian-codeword} shows that fermionic Gaussian operations, even supplemented by classical randomness, cannot realize any non-trivial code for either arbitrary single-Majorana error detection or quantum error correction.

\section{Non-Gaussianity Cost}
\label{sec:non-Gaussianity_cost}

Given the no-go results, a natural quantitative question is then how much non-Gaussian resource is required. 
To address this question, we consider the doped Gaussian model, in which Gaussian states and Gaussian unitaries are free, while bounded-weight non-Gaussian unitaries are counted as resource primitives~\cite{mele2025EfficientLearningQuantum,tarabunga2026Fermionicnon-Gaussianity}.
This model also captures existing constructive encoding schemes for Majorana stabilizer codes, which supplement Gaussian operations with bounded-weight non-Gaussian unitaries~\cite{mudassar2024EncodingMajoranaCodes}. Specifically, an $n$-mode $(g,q)$-doped Gaussian circuit for $n,g,q\in\mathbb{N}$ and $q\geq4$ takes the form
\begin{equation}
U
=
G_gV_gG_{g-1}\cdots V_2G_1V_1G_0,
\label{eq:g_doped_circuit}
\end{equation}
where each $G_j$ is a fermionic Gaussian unitary, and each $V_j$ is a non-Gaussian unitary generated by a Hermitian Majorana product of even weight $w_j\leq q$.
We denote by $\mathcal G_{g,q}^{(n)}$ the set of pure $n$-mode states
that can be prepared from a vacuum state using a circuit of the
form Eq.~\eqref{eq:g_doped_circuit} and call a state in $\mathcal G_{g,q}^{(n)}$ a $(g,q)$-doped state. For a
fermionic code $\mathcal C$, we define its $(g,q)$-doped Gaussian fidelity
as
\begin{equation}
F_{g,q}^{(n)}(\mathcal C)
\coloneqq
\sup_{\substack{\ket{\psi}\in\mathcal C\\
\braket{\psi}{\psi}=1}}
\sup_{\ket{\phi}\in\mathcal G_{g,q}^{(n)}}
\abs{\braket{\phi}{\psi}}^2.
\label{eq:def_code_g_doped_fidelity}
\end{equation}
This quantity characterizes the largest fidelity attainable with a codeword using $g$ bounded-weight non-Gaussian unitaries.
The following theorem gives a universal bound on this fidelity and, as
a direct consequence, a lower bound on the non-Gaussian cost of
codeword preparation.

\begin{theorem}[Universal non-Gaussian cost bound]
\label{thm:universal_non_Gaussian_cost}
Let $P$ be the orthogonal projector onto an
$[\![n,k,d_{\mathrm F}]\!]_{\mathrm F}$ fermionic code
$\mathcal C=\operatorname{im}(P)$ with
$\operatorname{rank}(P)=2^k>1$, and $d_{\mathrm F}\geq3$.
Define $r_{\dF,q}\coloneqq\left\lfloor(d_{\mathrm F}-3)/(q-2)\right\rfloor$, $s_{g,q}\coloneqq q\max\{g-r_{\dF,q},0\}$
and
\begin{equation}
t_{g,q}
\coloneqq
\left\lfloor
\frac{
d_{\mathrm F}
-(q-2)\min\{g,r_{\dF,q}\}
-1
}{2}
\right\rfloor .
\end{equation}
Then
\begin{equation}\label{eq:unified_gq_fidelity_bound}
F_{g,q}^{(n)}(\mathcal C)
\leq
\left[
1-H_2^{-1}\left(
\frac{ \max\{0,k-s_{g,q}\} }{\max\{n-s_{g,q},n-k\}}
\right)
\right]^{t_{g,q}}.
\end{equation}
Here $H_2^{-1}$ denotes the inverse of the binary entropy function
restricted to $[0,1/2]$.
Consequently, if a $(g,q)$-doped Gaussian circuit exactly prepares
a codeword in $\mathcal C$ from a vacuum state, then
\begin{equation}\label{eq:exact_non_Gaussian_cost_from_fidelity}
g
\geq
\left\lfloor
\frac{d_{\mathrm F}-3}{q-2}
\right\rfloor
+
\left\lceil
\frac{k}{q}
\right\rceil.
\end{equation}
\end{theorem}

The proof is given in
Appendix~\ref{app:proof_universal_non-Gaussian_cost}.
The main idea is to split the circuit into two parts. We absorb the
$\min\{g,r_{\dF,q}\}$ output-nearest non-Gaussian gates into the code projector,
leaving a transformed code with Majorana distance at least three.
After collecting the remaining Gaussian layers at the output, the
residual non-Gaussian gates can affect only a core of at most
$s_{g,q}$ fermionic modes via a compression trick~\cite{mele2025EfficientLearningQuantum}, while the complementary subsystem remains Gaussian. 
If $s_{g,q}<k$, this core is too small to accommodate the full logical dimension of the code. The Gaussian complement then necessarily produces a nonzero loss in the overlap with every codeword. 
Iterating this argument for the number of steps permitted by the remaining Majorana distance yields Eq.~\eqref{eq:unified_gq_fidelity_bound}. 
Exact preparation therefore requires $s_{g,q}\geq k$, which gives Eq.~\eqref{eq:exact_non_Gaussian_cost_from_fidelity}.

Thus, for fixed $q=\cO(1)$, exact codeword preparation requires
\begin{equation}\label{eq:distance_logical_bound}
g=\Omega(d_{\mathrm F}+k)
\end{equation}
non-Gaussian gates. 
This bound quantitatively captures the intrinsic difficulty of fermionic quantum error correction. The code distance $d_{\mathrm F}$ and the number of encoded qubits $k$ are two of its most fundamental parameters, characterizing robustness against Majorana noise and logical information storage capacity, respectively.
Our result shows that improving either capability necessarily requires a corresponding increase in non-Gaussian resources: neither stronger error protection nor larger logical capacity can be achieved at low non-Gaussian cost.

\section{Entanglement Distillation}\label{sec:entanglement_distillation}

The preceding results naturally raise the question of whether the limitations of fermionic Gaussian operations extend beyond quantum error correction to other closely related error-suppression tasks.
A particularly pertinent test case is entanglement distillation, which is intimately connected to quantum error correction and can enable reliable quantum communication through the purification of noisy entangled pairs~\cite{bennett1996mixed}.
Remarkably, the answer reveals a sharp boundary of our obstruction.
While non-trivial fermionic quantum error correction necessarily requires non-Gaussian resources, local Gaussian measurements with postselection can distill fermionic Gaussian entanglement.
For the state family considered below, however, products of local Gaussian channels cannot increase the EPR-pair fidelity, even when supplemented with classical randomness.
Thus, fermionic Gaussianity does not preclude error suppression in general: the obstruction to encoding and protecting arbitrary logical quantum information does not prevent the selective purification of a fixed entanglement resource.

To show this, we 
consider the parametrized fermionic Gaussian pair
\begin{equation}
\ket{\psi_\theta}
=
\cos\theta\ket{0,0}
+
\sin\theta\ket{1,1},
\qquad
0<\theta\leq\frac{\pi}{4},
\end{equation}
and denote the EPR pair by $\ket{\Phi^{+}}=\ket{\psi_{\theta=\pi/4}}.$
The fidelity of a bipartite state $\rho$ with $\ket{\Phi^+}$ is defined as
\begin{equation}
F_{\mathrm{EPR}}(\rho)
=\bra{\Phi^+}\rho\ket{\Phi^+}.
\end{equation}
In particular, $F_{\mathrm{EPR}}(\ketbra{\psi_\theta}{\psi_\theta})=(1+\sin(2\theta))/2$ provides an equivalent ordering of entanglement within the state family $\{\ket{\psi_\theta}\}$. Thus, our target is to consume many copies of $\ket{\psi_\theta}$ to obtain a single copy of $\ket{\psi_{\theta^\prime}}$ with $F_{\mathrm{EPR}}(\ketbra{\psi_{\theta^\prime}}{\psi_{\theta^\prime}})>F_{\mathrm{EPR}}(\ketbra{\psi_\theta}{\psi_\theta})$. To preserve the canonical anticommutation relations across fermionic subsystems, we use the fermionic tensor product $\otimes_{\mathrm F}$ (see details in Appendix~\ref{sec:notation}).

\begin{theorem}[Fermionic Gaussian operations in entanglement distillation]
\label{thm:entanglement_distillation}
First, convex mixtures of local trace-preserving fermionic Gaussian channels cannot increase the EPR-pair fidelity of any finite number of copies of $\ket{\psi_\theta}$.
More precisely, for any finite $m$, probability distribution $\{p_r\}$, and arbitrary local trace-preserving fermionic Gaussian channels $\Lambda_A^{(r)}$ and $\Lambda_B^{(r)}$ mapping $m$ input modes to one output mode each,
\begin{equation}\label{eq:Gaussian_entanglemenet_distillation_upper_bound}
F_{\mathrm{EPR}}
\left[
\sum_r p_r
\bigl(\Lambda_A^{(r)}\otimes_\mathrm{F}\Lambda_B^{(r)}\bigr)
\left(
\ketbra{\psi_\theta}{\psi_\theta}_{A,B}^{\otimes_\mathrm{F} m}
\right)
\right]
\leq
F_{\mathrm{EPR}}(\ketbra{\psi_\theta}{\psi_\theta}).
\end{equation}
By contrast, there exists a local fermionic Gaussian measurement that can probabilistically convert a single copy of $\ket{\psi_\theta}$ into the maximally entangled state vector $\ket{\Phi^{+}}$ conditioned on the successful outcome, with success probability $p_{\mathrm{succ}}=2\sin^{2}\theta$.
\end{theorem}

The proof is presented in Appendix~\ref{app:proof_entanglement_distillation}.
In bosonic systems, Gaussian entanglement distillation is impossible not only under local Gaussian channels, but also when Gaussian measurements, classical communication, and processing of measurement outcomes, including feedforward and postselection, are allowed~\cite{eisert2002distilling,Giedke2002distillation,Fiurasek2002Gaussian}.
By contrast, Gaussian postselection can enable fermionic entanglement distillation, revealing a fundamental distinction between the two settings.
This difference arises despite the close similarity of their covariance matrix transformation laws (see Appendix~\ref{app:Schur_complement_formula}).
In both cases, Gaussian postselection is described by a Schur complement~\cite{eisert2002distilling}, with a characteristic sign difference reflecting the antisymmetric fermionic and symmetric bosonic covariance matrices.

\section{Summary and outlook}\label{sec:outlook}

Our results establish fermionic non-Gaussianity as a fundamental resource for quantum error correction in fermionic systems, playing a role closely analogous to that of entanglement in qubit systems.
In conventional quantum error correction, logical information is encoded nonlocally in global degrees of freedom, so that local noise cannot access or irreversibly destroy it.
Our results show that, in fermionic systems, logical information must similarly be stored in genuine many-body correlations that are invisible to low-weight Majorana observables.
This requirement is fundamentally incompatible with fermionic Gaussianity: by Wick's theorem, all higher-order correlations of a pure Gaussian state are completely determined by its covariance matrix, leaving no independent many-body degrees of freedom in which protected logical information can be encoded.

These results delineate the extent to which the native advantages of fermionic hardware can persist at the fault-tolerant level.
Native fermionic platforms can represent fermionic degrees of freedom directly, avoiding the encoding overhead associated with mappings to qubits, while their natural Gaussian dynamics, such as Majorana braiding and single-particle evolution in cold-atom systems, are experimentally accessible and efficiently classically simulable~\cite{Valiant2002QuantumCircuits,terhal2002classical}.
Quantum error correction, however, requires non-Gaussian control beyond quadratic dynamics, such as controlled atomic collisions~\cite{CollisionalQuantumGates}, adding experimental complexity while leaving the Gaussian framework in which efficient classical simulation is guaranteed.
In Majorana architectures, these operations lie outside the native topologically protected braiding gate set~\cite{PhysRevB.99.144521} and require multi-Majorana interactions or magic-state resources.
This stands in sharp contrast to qubit quantum error correction, where non-trivial codes can remain entirely within the stabilizer framework and can therefore be efficiently simulated at large scale using stabilizer-based tools~\cite{Gidney2021stimfaststabilizer}.
Thus, the resources required for fault tolerance introduce physical implementation overhead and additional challenges for classical simulation, and our bounds provide a quantitative benchmark for assessing whether the underlying hardware advantages survive once quantum error correction is incorporated.

Beyond quantum error correction, our results suggest possible connections to fermionic quantum phases of matter and complexity theory. 
The \emph{no low-energy trivial states} (NLTS) property excludes shallowly preparable states throughout the low-energy sector of a local Hamiltonian~\cite{AnshuBreuckmannNirkhe2023,Herasymenko2024FermionicNLTS}, while the proposed \emph{no low-energy trivial magic} (NLTM) property strengthens this requirement by demanding that such low-energy states exhibit long-range magic, meaning that their nonstabilizerness cannot be removed by shallow local unitaries~\cite{WeiLiu2025LongRangeMagic,Parham2026MagicHierarchy}. Our non-Gaussianity bounds suggest that fermionic codes may exhibit a related resource-sensitive obstruction at the level of codeword preparation. Together with the finite-depth local unitary viewpoint of phase classification~\cite{GuWangWen2015}, this raises the possibility that non-Gaussian preparation cost could provide a complementary perspective on interacting fermionic phases.\\

\begin{acknowledgments}
We thank Boren Gu, Zhide Lu, Chenfeng Cao, Tobias Haug, and Hong-Ye Hu for the fruitful discussions. Y.~T.~and J.~E.~are funded by the BMFTR (PasQuops, QuSol, 
Hybrid++, FermiQP), Berlin Quantum, 
the Munich Quantum Valley, the Quantum Flagship
(Millenion, PasQuans2), 
the DFG 
(SPP 2514, BoLaCo), the Clusters of Excellence (ML4Q, MATH+), the QuantERA (SPDCode),
and the European Research Council (DebuQC).
P.~F.~is funded by the European Research Council (ERC) project QPhysComplex.
Z.-W.~L.~is supported in part by NSFC under Grant No.~12475023, Dushi Program, and startup funding from YMSC.
\end{acknowledgments}

\bibliography{bib_arxiv}

\clearpage

\appendix
\onecolumngrid
\appendixtableofcontents

\section{Preliminaries}\label{app:preliminaries}
We first review the necessary notations for the fermionic systems and the concept of fermionic Gaussianity, providing the definitions for fermionic Gaussian unitaries, states, measurements, and operations.

\subsection{Notations for fermionic systems}\label{sec:notation}
For a system of $n$ fermionic modes, denote the annihilation and creation operators for mode $j\in[n]$ as $a_j$ and $a_j^\dagger$. The canonical anticommutation relations require
\begin{equation}
\left\{a_j,a_k^\dagger\right\}
=
\delta_{j,k}\bI,
\qquad
\left\{a_j,a_k\right\}
=
\left\{a_j^\dagger,a_k^\dagger\right\}
=
0.
\end{equation}
Each fermionic mode can be decomposed into a pair of \emph{Majorana operators},
\begin{equation}
\gamma_{2j-1}
=
a_j+a_j^\dagger,
\quad
\gamma_{2j}
=
-i(a_j-a_j^\dagger),\quad j\in[n].
\end{equation}
The Jordan--Wigner transformation~\cite{jordan1928BerPaulischeQuivalenzverbot} encodes 
an $n$-mode fermionic system into $n$ qubits with Hilbert space $(\mathbb C^2)^{\otimes n}$ and maps the $2n$ Majorana operators to the operators 
\begin{equation}\label{eq:Jordan-Wigner_app}
 \gamma_{2j-1}
 =
 \left(\prod_{k=1}^{j-1} Z_k\right)X_j,
 \quad
 \gamma_{2j}
 =
 \left(\prod_{k=1}^{j-1} Z_k\right)Y_j,
 \quad j\in[n],
\end{equation}
in $\mathcal{L}((\mathbb C^2)^{\otimes n})$, where $\{X_k,Y_k,Z_k\}$ are Pauli operators acting on the $k$-th qubit.
The Majorana operators satisfy $\gamma_j=\gamma_j^\dagger$, $\gamma_j^2=\mathbb I$, and $\{\gamma_j,\gamma_k\}=2\delta_{j,k}\mathbb I$.
For a subset $S=\{j_1,\cdots,j_{\abs{S}}\}\subseteq[2n]$, with
$j_1<\cdots<j_{\abs{S}}$, we define the corresponding \emph{Majorana product} $\gamma_S$ and \emph{Hermitian Majorana
product} $\widehat{\gamma}_S$,
\begin{equation}\label{eq:def_Majorana_product}
\gamma_S=\gamma_{j_1}\cdots\gamma_{j_{\abs{S}}},\qquad\widehat{\gamma}_S=\widehat{\gamma}_S^\dagger=(-i)^{\abs{S}(\abs{S}-1)/2}\gamma_S.
\end{equation}
Especially $\gamma_\emptyset=\widehat{\gamma}_\emptyset=\bI$. The Hermitian Majorana products satisfy 
\begin{equation}\label{eq:Hermitian_Majorana_product_product}
\widehat\gamma_S\widehat\gamma_T=(-1)^{\omega(S,T)}(-i)^{\abs{S}(\abs{S}-1)/2+\abs{T}(\abs{T}-1)/2-\abs{S\triangle T}(\abs{S\triangle T}-1)/2}\widehat\gamma_{S\triangle T},
\end{equation}
where $\omega(S,T)=\abs{\{(s,t)\in S\times T| s>t\}}$ and $S\triangle T=(S\setminus T)\cup(T\setminus S)$.
The \emph{Majorana weight} of $\widehat{\gamma}_S$ is the cardinality of $S$, denoted by
$\operatorname{wt}_{\mathrm F}
\left(
\widehat{\gamma}_S
\right)
=
\abs{S}$, the same definition also applies to $\gamma_S$.
The operators $\{\widehat{\gamma}_S\}_{S\subseteq[2n]}$ form an orthogonal operator basis of $\mathcal{L}((\mathbb C^2)^{\otimes n})$ w.r.t. the Hilbert--Schmidt inner product. Any nonzero operator $O\in\mathcal{L}((\mathbb C^2)^{\otimes n})$ can be decomposed as a linear combination of Hermitian products, and we define its \emph{Majorana degree} $\deg_{\mathrm{F}}(O)$ as the maximum Majorana weight of a Hermitian Majorana product with nonzero coefficient in its decomposition. In this sense, the Hermitian Majorana products are also referred to as Majorana monomials.
The total fermion parity operator is given by
\begin{equation}
\Pi=(-1)^{\sum_{j=1}^na_j^\dagger a_j}=\widehat{\gamma}_{[2n]}=Z^{\otimes n}.
\end{equation}
An operator is called even if it commutes with $\Pi$, equivalently, all its coefficients under the basis of odd-Majorana-weight Hermitian Majorana products vanish. Due to the superselection rule, all physical fermionic density matrices (denoted as $\mathcal{D}_\mathrm{F}$) are parity-preserving, i.e., even fermionic operators. 

Suppose there are two fermionic systems $\cH_1$ and $\cH_2$,
containing $n_1$ and $n_2$ fermionic modes, respectively, with
Majorana operators
$\left\{\gamma_j^{(1)}\right\}_{j\in\cI_1}$ and
$\left\{\gamma_j^{(2)}\right\}_{j\in\cI_2}$, where
$\abs{\cI_i}=2n_i$. Let
$\cH_{1,2}$
denote the Fock space of the composite $(n_1+n_2)$-mode fermionic system, with
Majorana operators
$\left\{\gamma_j^{(1,2)}\right\}_{j\in\cI_1\sqcup\cI_2}$.
As $\mathbb Z_2$-graded $*$-algebras, the operator algebra of the composite system is identified with the graded tensor product, i.e., the fermionic tensor product $\otimes_\mathrm{F}$, of the two subsystem operator algebras,
$\cL(\cH_{1,2})
\cong
\cL(\cH_1)
\otimes_{\mathrm F}
\cL(\cH_2)$.
Notice that the naive ordinary tensor product embedding $O_1\mapsto O_1\otimes\bI$ and $O_2\mapsto\bI\otimes O_2$ does not reproduce the anti-commutation relations of odd operators belonging to distinct fermionic subsystems, and a faithful ordinary tensor product representation can instead be obtained by introducing 
a $*$-algebra isomorphism
\begin{equation}
\cJ:
\cL(\cH_{1,2})
\xrightarrow{\cong}
\cL(\cH_1\otimes\cH_2),
\label{eq:CAR-representation-isomorphism}
\end{equation}
defined on the Majorana operators by
\begin{equation}
\cJ\left(\gamma_j^{(1,2)}\right)
=
\begin{cases}
\gamma_j^{(1)}\otimes\bI_{\cH_2},
& j\in\cI_1,\\[1mm]
\Pi_1\otimes\gamma_j^{(2)},
& j\in\cI_2,
\end{cases}
\label{eq:CAR-ordinary-representation}
\end{equation}
where
$\Pi_1=\widehat{\gamma}_{\cI_1}^{(1)}$
is the parity operator on $\cH_1$ and it ensures that Majorana operators associated with different
fermionic subsystems anti-commute. This is a representation equivalent to, although using a different placement of the parity operator from, the convention adopted in
Ref.~\cite{bravyi2005LagrangianRepresentationFermionica}. The Jordan–Wigner transformation in Eq.~\eqref{eq:Jordan-Wigner_app} can be viewed as an example of $\cJ$: the Pauli $Z$-string preceding an odd operator on mode $j$ is precisely the parity operator of all modes ordered before $j$. In particular, the fermionic tensor product of Majorana products
$\gamma_S^{(1)}$ and $\gamma_T^{(2)}$, with
$S\subseteq\cI_1$ and $T\subseteq\cI_2$, is transformed by
\begin{equation}
\cJ\left(
\gamma_S^{(1)}
\otimes_{\mathrm F}
\gamma_T^{(2)}
\right)
=
\gamma_S^{(1)}
\Pi_1^{\abs{T}}
\otimes
\gamma_T^{(2)}.
\label{eq:unphased-Majorana-ordinary-representation}
\end{equation}
For a nonzero operator $O$ with a definite fermionic parity (also called a parity-homogeneous operator), let
$p(O)\in\{0,1\}$ denote its fermionic parity degree, defined by
$\Pi O\Pi
=
(-1)^{p(O)}O$.
The representation $\cJ$ acts on $O_1\otimes_\mathrm{F}O_2$ for two such operators
$O_1\in\cL(\cH_1)$ and
$O_2\in\cL(\cH_2)$ as
\begin{equation}
\cJ\left(
O_1\otimes_{\mathrm F}O_2
\right)
=
O_1\Pi_1^{p(O_2)}
\otimes O_2.
\label{eq:graded-to-ordinary-representation}
\end{equation}
More generally, for parity-homogeneous operators
$O_1,O_1'\in\cL(\cH_1)$ and
$O_2,O_2'\in\cL(\cH_2)$, the fermionic tensor product obeys
\begin{equation}
\left(
O_1\otimes_{\mathrm F}O_2
\right)
\left(
O_1'\otimes_{\mathrm F}O_2'
\right)
=
(-1)^{p(O_2)p(O_1')}
\left(
O_1O_1'
\right)
\otimes_{\mathrm F}
\left(
O_2O_2'
\right).
\label{eq:graded-tensor-rule}
\end{equation}
Specially, if $O_1$ and $O_2$ are both odd operators, then
\begin{equation}
\left(
O_1\otimes_{\mathrm F}\bI
\right)
\left(
\bI\otimes_{\mathrm F}O_2
\right)
=
-
\left(
\bI\otimes_{\mathrm F}O_2
\right)
\left(
O_1\otimes_{\mathrm F}\bI
\right),
\label{eq:cross-subsystem-anti-commutation}
\end{equation}
which explicitly reproduces the canonical anti-commutation relation
between odd operators associated with distinct fermionic
subsystems.

\subsection{Fermionic Gaussian unitaries}
A parity-preserving fermionic Gaussian unitary on $n$ modes is
generated by an even quadratic Hamiltonian,
\begin{equation}
U_A
=
\exp\left[
\frac{1}{4}
\sum_{j,k=1}^{2n}
A_{j,k}\gamma_j\gamma_k
\right],
\end{equation}
where $A$ is a real antisymmetric  $2n\times2n$ matrix. Its Heisenberg action on Majorana operators is
\begin{equation}
U_A^\dagger\gamma_jU_A
=
\sum_k\left(e^A\right)_{j,k}\gamma_k,
\qquad
e^A\in\operatorname{SO}(2n).
\end{equation}
From the above definition, we denote the $n$-mode generalized fermionic Gaussian unitaries, also known as $n$-qubit generalized matchgate unitaries under the Jordan-Wigner transformation~\cite{Valiant2002QuantumCircuits}, by
$\mathrm{M}_n=\left\{U_R| R\in\mathrm{O}(2n)
\right\}$,
where the adjoint action of each $U_R$ is characterized by an orthogonal matrix $R\in\operatorname{O}(2n)$,
\begin{equation}\label{eq:Gaussian_unitary_action}
U_R^\dagger\gamma_jU_R
=
\sum_kR_{j,k}\gamma_k.
\end{equation}
The implementing unitary is unique up to a global phase.
The determinant-$(-1)$ component is parity-flipping and requires an
odd Majorana factor, or equivalently a parity-preserving Gaussian
dilation using an auxiliary fermionic mode.

\subsection{Fermionic Gaussian states and positive Gaussian operators}
A fermionic state $\sigma$ is Gaussian if it can be written as~\cite{bittel2025OptimalTraceDistanceBoundsa}
\begin{equation}\label{eq:def_Gaussian_state}
\sigma
=
U_R\bigotimes_{j=1}^n\left(\frac{\bI+\nu_jZ}2\right)U_R^\dagger,
\end{equation}
where $U_R$ is the Gaussian unitary specified by $R\in\operatorname{O}(2n)$ according to Eq.~\eqref{eq:Gaussian_unitary_action}, each parameter $\nu_j\in[-1,1]$, and $Z$ is the single-qubit Pauli operator. A Gaussian state $\sigma$ is pure if and only if all $\abs{\nu_j}=1$ in Eq.~\eqref{eq:def_Gaussian_state}, and equivalently, it can be prepared by applying a Gaussian unitary to the vacuum state vector $\ket{0^n}$. We denote the set of all pure Gaussian state vectors by 
\begin{equation}\label{eq:def_pure_Gaussian_state}
\mathcal{G}_{\textup{pure}}^{(n)}
=
\left\{
U\ket{0^n}
\middle|
U\in\mathrm{M}_n
\right\},
\end{equation}
and call a fermionic state \emph{convex Gaussian} if it is a mixture of pure Gaussian states.
Every fermionic Gaussian state, i.e., in the form of Eq.~\eqref{eq:def_Gaussian_state}, is convex Gaussian and can be prepared from a global vacuum by adjoining vacuum Gaussian auxiliary modes, applying a
parity-preserving Gaussian unitary to the system and auxiliary modes, and
tracing out the auxiliary modes.

For any fermionic state
$\rho\in\mathcal{D}_{\mathrm F}$, its covariance matrix
$\Gamma(\rho)$ is defined by
\begin{equation}\label{eq:correlation_matrix}
\Gamma_{j,k}(\rho)
=
-\frac{i}{2}
\Tr(
[\gamma_j,\gamma_k]\rho
),
\qquad
j,k\in[2n].
\end{equation}
The matrix $\Gamma(\rho)$ is real and antisymmetric, i.e.,
$\Gamma(\rho)=-\Gamma(\rho)^\mathsf{T}$ and $\norm{\Gamma(\rho)}_\infty\leq1$ for any $\rho$~\cite{bittel2025OptimalTraceDistanceBoundsa}. Fermionic
Gaussian states are characterized by Wick's theorem~\cite{wick1950EvaluationCollisionMatrix}: all odd-order
Majorana correlation functions vanish, while every even-order
correlation function is determined by the covariance matrix. In
particular, for
$S\subseteq[2n]$ with even $\abs{S}$,
\begin{equation}
\Tr(\widehat{\gamma}_S\sigma)
=
\operatorname{Pf}\left(
\Gamma_S(\sigma)
\right),
\end{equation}
where $\operatorname{Pf}$ denotes the Pfaffian and
$\Gamma_S(\sigma)$ is the principal submatrix of
$\Gamma(\sigma)$ indexed by $S$. A fermionic state $\sigma$ is
pure Gaussian if and only if
$\Gamma(\sigma)^2=-\bI$.
Consequently, a pure fermionic Gaussian state is uniquely determined
by its covariance matrix. When a fermionic state $\rho$ is evolved under a Gaussian unitary $U_R$ in Eq.~\eqref{eq:Gaussian_unitary_action}, its covariance matrix changes by the adjoint action with the corresponding $R\in\mathrm{O}(2n)$ as
\begin{equation}\label{eq:Gaussian_unitary_covariance_matrix}
\Gamma\left(U_R\rho U_R^\dagger\right)=R\Gamma(\rho)R^{\mathsf{T}}.
\end{equation}
A more general definition is the positive fermionic Gaussian operator: a nonzero positive operator $O\geq0$ is Gaussian if and only if $O/\Tr(O)$ is a Gaussian state, and the zero operator is Gaussian.

\subsection{Fermionic Gaussian completely positive maps and channels}\label{app:fermionic_Gaussian_maps}

Let
$\Phi:\cL(\cH_1)\rightarrow\cL(\cH_2)$
be a linear map, where $\cH_1$ is the Fock space of $n$ input
fermionic modes and $\cH_2$ is the Fock space of $m$ output
fermionic modes. We restrict throughout to parity-preserving maps,
namely maps satisfying
\begin{equation}
\Phi\left(
\Pi_1 O\Pi_1
\right)
=
\Pi_2\Phi(O)\Pi_2
\qquad
\forall O\in\cL(\cH_1),
\label{eq:parity_preserving_fermionic_map}
\end{equation}
where $\Pi_1$ and $\Pi_2$ are the fermion parity operators on
$\cH_1$ and $\cH_2$, respectively. Equivalently, $\Phi$ preserves
the $\mathbb Z_2$ grading of the fermionic operator algebra.
A parity-preserving linear map $\Phi$ is completely positive if,
for every finite-dimensional fermionic reference system $\cH_R$, the extended map $\Phi\otimes_{\mathrm F}\operatorname{id}_R:O_1\otimes_\mathrm{F}O_2\mapsto\Phi(O_1)\otimes_\mathrm{F}O_2$ maps positive operators to positive operators. Under the
ordinary tensor product representation of the canonical anticommutation relations (CAR) algebra on the composite fermionic system
introduced in Appendix~\ref{sec:notation}, this definition is equivalent to the usual notion
of complete positivity for matrix algebras
\cite{bravyi2005LagrangianRepresentationFermionica}. The map is
trace-non-increasing if
$\Tr[\Phi(O)]
\leq
\Tr(O)$ for all $O\geq0$,
and trace-preserving if equality holds for every positive operator
$O$.

To introduce the fermionic Choi representation, let
$\cH_{1'}\simeq\cH_1$
be an auxiliary $n$-mode reference Fock space isomorphic to the
input space. Let
$\{a_j\}_{j=1}^{n}$
be the annihilation operators of the modes of $\cH_1$, and let
$\{a_j'\}_{j=1}^{n}$
be those of the reference copy $\cH_{1'}$. Similarly, denote the
corresponding Majorana operators by
$\{\gamma_j\}_{j=1}^{2n}$ and
$\{\gamma_j'\}_{j=1}^{2n}$.
We fix the global fermionic mode ordering
\begin{equation}
a_1<\cdots<a_n<a_n'<\cdots<a_1'.
\label{eq:phase_free_ordering_H1_reference}
\end{equation}
For
$\vec\mu=(\mu_1,\cdots,\mu_n)\in\{0,1\}^n$,
define the occupation bases
\begin{equation}
\begin{aligned}
\ket{\vec\mu}_1
&=
(a_1^\dagger)^{\mu_1}
\cdots
(a_n^\dagger)^{\mu_n}
\ket{0}_1,\\
\ket{\vec\mu}_{1'}
&=
(a_n'^\dagger)^{\mu_n}
\cdots
(a_1'^\dagger)^{\mu_1}
\ket{0}_{1'}.
\end{aligned}
\label{eq:reversed_reference_occupation_basis}
\end{equation}
With this convention, the joint fermionic Fock space is represented
on the ordinary Hilbert-space tensor product
$\cH_1\otimes\cH_{1'}$. The reversed ordering of the reference
basis is chosen so that the fermionic signs generated when the
creation operators of the two subsystems are regrouped cancel
exactly. For a basis component containing $N$ occupied mode pairs,
moving all $a'^\dagger$ operators to the right of all
$a^\dagger$ operators produces the sign
$(-1)^{N(N-1)/2}$, while reversing the order of the occupied
reference creation operators produces the same sign. For example,
\begin{equation}
a_1^\dagger a_1'^\dagger a_2^\dagger a_2'^\dagger
=
-a_1^\dagger a_2^\dagger a_1'^\dagger a_2'^\dagger
=
a_1^\dagger a_2^\dagger a_2'^\dagger a_1'^\dagger.
\end{equation}
In this way, the fermionic maximally entangled state between $\cH_1$ and
$\cH_{1'}$ admits the phase-free representation
\begin{equation}
\begin{aligned}
\ket{\Psi_\mathrm{F}}
&=
2^{-n/2}
\sum_{\vec\mu\in\{0,1\}^n}
\ket{\vec\mu}_1\otimes\ket{\vec\mu}_{1'}\\
&=
2^{-n/2}
\prod_{j=1}^{n}
\left(
\bI+a_j^\dagger a_j'^\dagger
\right)
\ket{0}_{1,1'}\\
&=
\exp\left[
\frac{\pi}{4}
\sum_{j=1}^{n}
\left(
a_j^\dagger a_j'^\dagger-a_j' a_j
\right)
\right]
\ket{0}_{1,1'}\\
&=
\exp\left[
\frac{\pi}{8}
\sum_{j=1}^{n}
\left(
\gamma_{2j-1}\gamma'_{2j-1}
-
\gamma_{2j}\gamma'_{2j}
\right)
\right]
\ket{0}_{1,1'}.
\end{aligned}
\label{eq:fermionic_maximally_entangled_state}
\end{equation}
The last expression is generated from the joint vacuum by a
quadratic Majorana unitary, and therefore
$\ket{\Psi_\mathrm{F}}$ is a pure fermionic Gaussian state.
More precisely, let
$\rho_{\Psi_\mathrm{F}}$
denote the corresponding fermionic maximally entangled state in the
joint CAR algebra of $\cH_1$ and $\cH_{1'}$, whose fixed
reversed-reference ordinary tensor product representation satisfies
\begin{equation}
\cJ_{1,1'}
\left(
\rho_{\Psi_\mathrm{F}}
\right)
=
\ketbra{\Psi_{\mathrm F}}{\Psi_{\mathrm F}},
\label{eq:fermionic_Bell_state_representation}
\end{equation}
where $\cJ_{1,1'}:\cL(\cH_{1,1'})\rightarrow\cL(\cH_1\otimes\cH_{1'})$ is the isomorphism defined in Eq.~\eqref{eq:CAR-representation-isomorphism}. Similar notations apply below. We represent the fermionic extension of the parity-preserving linear map $\Phi$ as
\begin{equation}
\widetilde\Phi_{1'}
\coloneqq
\cJ_{2,1'}
\circ
(\Phi\otimes_{\mathrm F}\operatorname{id}_{1'})
\circ
\cJ_{1,1'}^{-1},
\end{equation}
acting on parity-homogeneous operators $O_1,O_1'\in\cL(\cH_1)$ by
\begin{equation}
\widetilde\Phi_{1'}(O_1\otimes O_1')=\Phi\left(O_1\Pi_1^{p(O_1')}\right)\Pi_2^{p(O_1')}\otimes O_1'.
\end{equation}
The fermionic Choi operator of $\Phi$ is defined, in the fixed
ordinary tensor product representation introduced above, by
\begin{equation}
\begin{aligned}
J_{\mathrm F}(\Phi)
&\coloneqq
\cJ_{2,1'}
\left[
\left(
\Phi\otimes_{\mathrm F}\operatorname{id}_{1'}
\right)
\left(
\rho_{\Psi_\mathrm{F}}
\right)
\right]\\
&=
\widetilde{\Phi}_{1'}
\left(
\ketbra{\Psi_{\mathrm F}}{\Psi_{\mathrm F}}
\right).
\end{aligned}
\label{eq:fermionic-Choi-operator}
\end{equation}

Using Eq.~\eqref{eq:fermionic_maximally_entangled_state} and expanding Eq.~\eqref{eq:fermionic-Choi-operator} in the occupation
basis gives
\begin{equation}\label{eq:fermionic_Choi_basis_expansion}
J_{\mathrm F}(\Phi)
=
2^{-n}
\sum_{\vec\mu,\vec\nu\in\{0,1\}^{n}}
\Phi\left(
\ketbra{\vec\mu}{\vec\nu}_1\Pi_1^{\varepsilon(\vec\mu,\vec\nu)}
\right)
\Pi_2^{\varepsilon(\vec\mu,\vec\nu)}
\otimes
\ketbra{\vec\mu}{\vec\nu}_{1'},
\end{equation}
where $\varepsilon(\vec\mu,\vec\nu)\coloneqq(\abs{\vec\mu}+\abs{\vec\nu})\mod 2$.
Consequently, the Choi correspondence is one-to-one: the action of
$\Phi$ on a parity-homogeneous $O\in\cL(\cH_1)$ can be reconstructed from $J_{\mathrm F}(\Phi)$ by
\begin{equation}
\Phi(O)
=
2^n
\Tr_{1'}
\left[
\left(
\bI_{\cH_2}\otimes \left(O\Pi_1^{p(O)}\right)^{\mathsf{T}_{1'}}
\right)
J_{\mathrm F}(\Phi)
\right]\Pi_2^{p(O)},
\label{eq:inverse_fermionic_Choi_map}
\end{equation}
where $\mathsf{T}_{1'}$ denotes the matrix transpose in the reversed
occupation basis of $\cH_{1'}$, after identifying
$\cH_{1'}\cong\cH_1$.
The map $\Phi$ is completely positive (CP) if and only if
$J_{\mathrm F}(\Phi)\geq0$. If it is
trace-preserving, it is called a fermionic quantum channel.

A nonzero fermionic CP map $\Phi$ is called
\emph{Gaussian} if its fermionic Choi operator
$J_{\mathrm F}(\Phi)$ is a positive fermionic Gaussian
operator, equivalently, in the fixed reversed-reference
representation used above,
$J_{\mathrm F}(\Phi)
/
\Tr[J_{\mathrm F}(\Phi)]$
is a fermionic Gaussian state with respect to the induced Majorana
representation on $\cH_2\otimes\cH_{1'}$. The zero map is included as a trivial Gaussian map.
Thus,
\begin{equation}
\Phi
\text{ is a fermionic Gaussian CP map}
\quad\Longleftrightarrow\quad
J_{\mathrm F}(\Phi)
\text{ is a positive Gaussian operator}.
\label{eq:Gaussian_map_Gaussian_Choi_equivalence}
\end{equation}
For parity-preserving maps, this Choi characterization is equivalent
to the quadratic Grassmann-kernel definition of fermionic Gaussian
maps introduced in Ref.~\cite{bravyi2005LagrangianRepresentationFermionica}. And Lemma~4 in Ref.~\cite{bravyi2005LagrangianRepresentationFermionica} proves that a Gaussian CP map always maps a positive Gaussian operator to a positive Gaussian operator.

A trace-preserving fermionic Gaussian CP map is called a fermionic
Gaussian channel. A Gaussian channel maps (convex) Gaussian states to (convex) Gaussian states. In particular, conjugation by a
parity-preserving Gaussian unitary $U$,
$\Phi_U(O)
=
UOU^\dagger$,
defines a fermionic Gaussian channel. More generally, the Stinespring representation of a Gaussian channel consists of appending a
Gaussian ancillary state, applying a parity-preserving Gaussian
unitary to the enlarged system, and tracing out any subset of modes. The action of a Gaussian channel on the covariance matrix of the input state can be seen from the following lemma, reformulated from Ref.~\cite{bravyi2005LagrangianRepresentationFermionica}.

\begin{lemma}[Action of a Gaussian channel on the covariance matrix,~\cite{bravyi2005LagrangianRepresentationFermionica}]\label{lem:Gaussian_channel_covariance_matrix}
For a fermionic Gaussian channel
$\Phi:\cL(\cH_1)\rightarrow\cL(\cH_2)$ mapping from an $n$-mode fermionic system to an $m$-mode fermionic system, $\Phi$ acts affinely on covariance matrices as
\begin{equation}\label{eq:Gaussian_channel_covariance_matrix}
\Gamma(\rho)\mapsto\Gamma(\Phi(\rho))=\mathsf{M}\Gamma(\rho)\mathsf{M}^{\mathsf{T}}+\mathsf{N},
\end{equation}
for some $\mathsf{M}\in\mathbb{R}^{2m\times2n}$, $\mathsf{N}\in\mathbb{R}^{2m\times2m}$ and $\mathsf{N}^\mathsf{T}=-\mathsf{N}$. Furthermore, $\bI_{2m}-\mathsf{M}\mathsf{M}^\mathsf{T}+i\mathsf{N}\geq0$ and $\norm{\mathsf{M}}_\infty\leq1$.
\end{lemma}
\begin{proof}[Proof of Lemma~\ref{lem:Gaussian_channel_covariance_matrix}]
Consider the Stinespring representation of $\Phi$. We introduce a large enough environment system $E_\mathrm{in}$ consisting of $r\geq m-n$ fermionic modes, such that there exist a pure Gaussian state $\tau_{E_\mathrm{in}}$ and a Gaussian unitary $U_R$ with $R\in\mathrm{O}(2n+2r)$ satisfying
\begin{equation}
\Phi(\rho)=\Tr_{E_\mathrm{out}}\left[U_R\left(\rho\otimes_\mathrm{F}\tau_{E_\mathrm{in}}\right)U_R^\dagger\right],
\end{equation}
where the system $E_\mathrm{out}$ to be traced out consists of $n+r-m$ fermionic modes. The covariance matrix of the input state $\rho\otimes_\mathrm{F}\tau_{E_\mathrm{in}}$ takes the block-diagonal form
\begin{equation}
\Gamma\left(\rho\otimes_\mathrm{F}\tau_{E_\mathrm{in}}\right)=
\begin{bmatrix}
\Gamma(\rho)&0\\
0&\Gamma(\tau_{E_\mathrm{in}})
\end{bmatrix},
\end{equation}
By Eq.~\eqref{eq:Gaussian_unitary_covariance_matrix}, we have $\Gamma\left(U_R\left(\rho\otimes_\mathrm{F}\tau_{E_\mathrm{in}}\right)U_R^\dagger\right)=R\Gamma\left(\rho\otimes_\mathrm{F}\tau_{E_\mathrm{in}}\right)R^{\mathsf{T}}$. After tracing out $E_{\mathrm{out}}$, the remaining $2m$ rows of $R$ corresponding to the output Majorana operators can be split as $R_{\mathrm{out}}=\begin{bmatrix}\mathsf{M}&\mathsf{L}\end{bmatrix}$ for some
$\mathsf{M}\in\mathbb{R}^{2m\times2n}$ and $\mathsf{L}\in\mathbb{R}^{2m\times2r}$. Since $R$ is orthogonal, $R_\mathrm{out}R_\mathrm{out}^\mathsf{T}=\mathsf{M}\mathsf{M}^\mathsf{T}+\mathsf{L}\mathsf{L}^\mathsf{T}=\bI_{2m}$. In this way, the covariance matrix of the output state $\Phi(\rho)$ is
\begin{equation}
\Gamma(\Phi(\rho))=
\begin{bmatrix}\mathsf{M}&\mathsf{L}\end{bmatrix}
\begin{bmatrix}
\Gamma(\rho)&0\\
0&\Gamma(\tau_{E_\mathrm{in}})
\end{bmatrix}
\begin{bmatrix}
\mathsf{M}^\mathsf{T}\\
\mathsf{L}^\mathsf{T}
\end{bmatrix}
=\mathsf{M}\Gamma(\rho)\mathsf{M}^\mathsf{T}+\mathsf{L}\Gamma(\tau_{E_\mathrm{in}})\mathsf{L}^\mathsf{T}.
\end{equation}
Let $\mathsf{N}=\mathsf{L}\Gamma(\tau_{E_\mathrm{in}})\mathsf{L}^\mathsf{T}$, then since $\Gamma(\tau_{E_\mathrm{in}})^\mathsf{T}=-\Gamma(\tau_{E_\mathrm{in}})$ as a fermionic covariance matrix, we have $\mathsf{N}^\mathsf{T}=-\mathsf{N}$. Furthermore,
\begin{equation}
\bI_{2m}-\mathsf{M}\mathsf{M}^\mathsf{T}+i\mathsf{N}=\mathsf{L}\mathsf{L}^\mathsf{T}+i\mathsf{L}\Gamma(\tau_{E_\mathrm{in}})\mathsf{L}^\mathsf{T}=\mathsf{L}(\bI_{2r}+i\Gamma(\tau_{E_{\mathrm{in}}}))\mathsf{L}^\mathsf{T}\geq0,
\end{equation}
where we use the fact that the operator norm of a covariance matrix is upper bounded by one. Also, $\mathsf{M}\mathsf{M}^\mathsf{T}+\mathsf{L}\mathsf{L}^\mathsf{T}=\bI_{2m}$ implies $\norm{\mathsf{M}}_\infty\leq1$.
\end{proof}

\subsection{Fermionic Gaussian measurements}\label{app:fermionic_Gaussian_measurements}
A fermionic measurement with a countable outcome set $\mathcal X$ is described by
a quantum instrument consisting of branches
\begin{equation}
\left\{
\cM_x:
\cL(\cH_1)
\rightarrow
\cL(\cH_2)
\right\}_{x\in\mathcal X},
\label{eq:fermionic_quantum_instrument}
\end{equation}
where every $\cM_x$ is a CP and
trace-non-increasing parity-preserving map, and the sum of all
branches is trace-preserving, i.e.,
$\sum_{x\in\mathcal X}
\cM_x$
is a fermionic quantum channel.
For an input state $\rho$, the probability of obtaining the outcome
$x$ is $\Tr[\cM_x(\rho)]$.
Conditioned on an outcome $x$ with non-zero probability, the normalized
post-measurement state is
\begin{equation}
\rho_x
=
\frac{
\cM_x(\rho)
}{
\Tr[\cM_x(\rho)]
}.
\label{eq:fermionic_conditional_state}
\end{equation}
A fermionic measurement branch $\cM_x$ is called \emph{Gaussian} if it is a
Gaussian CP map, i.e., if its fermionic Choi
operator is a positive Gaussian operator. A fermionic instrument is
called \emph{Gaussian} if every nonzero branch
$\cM_x$ is Gaussian.

The distinction between fine-grained branches and
classically coarse-grained outcomes is important. If
$\mathcal X_y\subseteq\mathcal X$ is a collection of fine-grained
outcomes that are identified with the same classical label $y$, the
corresponding coarse-grained branch is
\begin{equation}
\cM_y^{\textup{CG}}
\coloneqq
\sum_{x\in\mathcal X_y}
\cM_x.
\label{eq:coarse_grained_fermionic_branch}
\end{equation}
Even when every $\cM_x$ is a fermionic Gaussian CP map, the sum
$\cM_y^{\textup{CG}}$ need not itself be a Gaussian CP map, since positive
fermionic Gaussian operators are not generally closed under
addition. Thus throughout this work, when a protocol is said to use
fermionic Gaussian measurements or Gaussian operations only, the
Gaussianity requirement is imposed on every nonzero fine-grained
physical branch rather than on an arbitrary classical
coarse-graining of those branches.

\section{Fermionic non-Gaussianity is necessary in quantum error correction}
In this appendix, we provide proofs and further discussions of the results presented in Section~\ref{app:no_Gaussian_qec}.
\subsection{Proof of Theorem~\ref{thm:no-Gaussian-codeword}}\label{app:proof_no-Gaussian-qec}
As briefly sketched in the main text, the proof of Theorem~\ref{thm:no-Gaussian-codeword} relies on the fact that if an arbitrary pure state has the same covariance matrix as a pure fermionic Gaussian state, then the two states coincide up to a global phase (see Lemma~1 and its proof in Ref.~\cite{bittel2025OptimalTraceDistanceBoundsa}), reformulated as Lemma~\ref{lem:Gaussian-correlation-uniqueness}.
\begin{lemma}[Uniqueness of a pure Gaussian state from its covariance matrix~\cite{bittel2025OptimalTraceDistanceBoundsa}]
\label{lem:Gaussian-correlation-uniqueness}
Let $\ket{\phi}$ be a pure fermionic Gaussian state on $n$ fermionic
modes, and let $\ket{\psi}$ be an arbitrary normalized state vector on
the same system. If their covariance matrices defined in Eq.~\eqref{eq:correlation_matrix} are equal, i.e.,
$\Gamma(\psi)=\Gamma(\phi)$,
then for some $\varphi\in[0,2\pi)$,
$\ket{\psi}=e^{i\varphi}\ket{\phi}$.
\end{lemma}

Then we use this Lemma to prove Theorem~\ref{thm:no-Gaussian-codeword} in the main text.

\begin{proof}[Proof of Theorem~\ref{thm:no-Gaussian-codeword}]
For a code projector $P$ with $d_{\mathrm F}(P)\geq3$, every Majorana product of weight
strictly smaller than three is detected by the code. In particular,
for every pair $j<k$, there exists a real scalar $\lambda_{j,k}$ such
that
$P\left(-i\gamma_j\gamma_k\right)P
=
\lambda_{j,k}P$.
Consequently, the elements of
the covariance matrix of every normalized state vector
$\ket{\psi}\in\mathcal C=\operatorname{im}(P)$ satisfy
\begin{equation}
\Gamma_{j,k}(\psi)=\bra{\psi}-i\gamma_j\gamma_k\ket{\psi}=\bra{\psi}P\left(-i\gamma_j\gamma_k\right)P\ket{\psi}
=
\lambda_{j,k},
\qquad
j<k.
\label{eq:identical-bilinears}
\end{equation}
All normalized pure states in the logical space therefore have the
same covariance matrix.
Suppose, toward a contradiction, that $\mathcal C$ contains a pure
fermionic Gaussian state vector $\ket{\phi}\in\mathcal G_{\mathrm{pure}}^{(n)}$. Then every normalized pure
logical state vector $\ket{\psi}\in\mathcal C$ satisfies
$\Gamma(\psi)=\Gamma(\phi)$.
By Lemma~\ref{lem:Gaussian-correlation-uniqueness}, this is possible
only if
$\ket{\psi}=e^{i\varphi_\psi}\ket{\phi}$
for some phase $\varphi_\psi$. Hence every vector in $\mathcal C$ is
proportional to $\ket{\phi}$, and therefore
$\dim(\mathcal C)=\operatorname{rank}(P)=1$,
which contradicts the assumption that the code is non-trivial.
\end{proof}

Corollary~\ref{cor:no_convex_Gaussian_support} is a direct consequence of Theorem~\ref{thm:no-Gaussian-codeword}, stating that there is no convex Gaussian state supported on the code space of a non-trivial fermionic quantum error-correcting code. We present its proof 
below.

\begin{corollary}[No convex Gaussian density matrix in the support]\label{cor:no_convex_Gaussian_support}
Let $P$ be the projector onto a fermionic quantum code
$\mathcal{C}=\operatorname{im}(P)$ of $n$ modes with $\operatorname{rank}(P)>1$.
If
$d_{\mathrm F}(P)\geq 3$, then there is no convex Gaussian state $\rho$, such that $P\rho P=\rho$.
\end{corollary}

\begin{proof}[Proof of Corollary~\ref{cor:no_convex_Gaussian_support}]
For a code projector $P$ with $d_{\mathrm F}(P)\geq3$, suppose there is a convex Gaussian state $\rho$ supported on $\operatorname{im}(P)$, i.e., $P\rho P=\rho$ and there exists a probability distribution $\{p_\ell\}$ and an ensemble of pure fermionic Gaussian states $\{\ket{\phi_\ell}\}_\ell$, such that $\rho=\sum_{\ell}p_\ell\ket{\phi_\ell}\bra{\phi_\ell}$, 
then we have
\begin{equation}
0=\Tr(\rho)-\Tr(P\rho P)=\Tr((\bI-P)\rho)=\sum_\ell p_\ell\bra{\phi_\ell}(\bI-P)\ket{\phi_\ell}.
\end{equation}
This further implies for each $p_\ell>0$, we have
\begin{equation}
\bra{\phi_\ell}(\bI-P)\ket{\phi_\ell}=\norm{(\bI-P)\ket{\phi_\ell}}^2=0,
\end{equation}
so $(\bI-P)\ket{\phi_\ell}=0$ and $\ket{\phi_\ell}\in\operatorname{im}(P)$, which contradicts Theorem~\ref{thm:no-Gaussian-codeword}.
\end{proof}

\subsection{Tradeoff between Gaussianity and approximate correctability}\label{app:proof_approximate_qec}

In this appendix, we quantitatively provide an upper bound for the Gaussian fidelity of any codeword of a fermionic code, in terms of parameters that characterize its approximate correctability.

Consider a fermionic code with encoding isometry $V$ and encoding channel $\cE(\rho)=V\rho V^\dagger$. The orthogonal projector onto the code space $\cC$ is $P=VV^\dagger$. The Gaussian fidelity of an $n$-mode fermionic state vector $\ket{\psi}$, defined as follows, is the maximum fidelity between $\ket{\psi}$ and any pure fermionic Gaussian state~\cite{oszmaniec2014ClassicalSimulationFermionic,dias2024ClassicalSimulationNonGaussian,reardon-smith2024ImprovedSimulationQuantum},
\begin{equation}\label{eq:def_Gaussian_fidelity}
F_0^{(n)}(\ket{\psi})
\coloneqq
\sup_{\ket{\phi}\in\mathcal G_{\textup{pure}}^{(n)}}
\abs{\braket{\phi}{\psi}}^2.
\end{equation}
It directly relates to the minimum trace distance between $\ket{\psi}$ and any pure Gaussian state, and thus serves as a quantifier of the non-Gaussianity of $\ket{\psi}$. We denote the supremum Gaussian fidelity over all codewords in $\cC$, also called the Gaussian fidelity of $\cC$, by
\begin{equation}\label{eq:def_code_Gaussian_fidelity}
F_0^{(n)}(\mathcal C)
\coloneqq
\sup_{\substack{\ket{\psi}\in\mathcal C\\
\braket{\psi}{\psi}=1}}
\sup_{\ket{\phi}\in\mathcal G_{\textup{pure}}^{(n)}}
\abs{\braket{\phi}{\psi}}^2
=\sup_{\ket{\phi}\in\mathcal G_{\textup{pure}}^{(n)}}
\bra{\phi}P\ket{\phi}.
\end{equation}

We quantify the approximate correctability using the formulation established in Refs.~\cite{beny2010AQEC,faist2020ContinuousSymmetriesApproximatea}. We consider a noise channel
$\cN:\cL(\cH_A)\rightarrow\cL(\cH_B)$,
with input Hilbert space $\cH_A$ and output Hilbert space $\cH_B$.
Choose a Stinespring isometry $W_\cN:\cH_A\rightarrow\cH_B\otimes\cH_E$ for $\cN$ with environment system $\cH_E$, such that
\begin{equation}
\cN(\rho)=\Tr_E(W_\cN\rho W_\cN^\dagger).
\end{equation}
The complementary channel associated with this dilation $\widehat{\cN}:\cL(\cH_A)\rightarrow\cL(\cH_E)$ is defined as
\begin{equation}\label{eq:def_complementary_channel}
\widehat\cN(\rho)=\Tr_B(W_\cN\rho W_\cN^\dagger).
\end{equation}
It represents the information about the input transferred to the Stinespring environment. The adjoint of $\widehat{\cN}$ is defined through
\begin{equation}
\Tr(O\widehat{\cN}(\rho))=\Tr(\widehat{\cN}^\dagger(O)\rho)
\end{equation}
for any observable $O\in\cL(\cH_E)$ in the environment system.

Denote the root fidelity between two density matrices as $F(\rho_1,\rho_2)=\norm{\sqrt{\rho_1}\sqrt{\rho_2}}_1$. We define the optimal worst-case entanglement fidelity of the encoded noise $\cN\circ\cE$ as
\begin{equation}\label{eq:def_worst_case_entanglement_fidelity}
\begin{aligned}
f_{\mathrm{worst}}(\mathcal N\circ\mathcal E)
\coloneqq
\sup_{\mathcal R}
\inf_{\ket{\omega}_{L,R}}
F\left(
\left[
(\mathcal R\circ\mathcal N\circ\mathcal E)
\otimes\operatorname{id}_R
\right](\omega_{L,R}),
\omega_{L,R}
\right),
\end{aligned}
\end{equation}
where $\cR:\cL(\cH_B)\rightarrow\cL(\cH_L)$ ranges over all recovery channels, $\cH_L$ is the logical system and $\cH_R$ is an arbitrary reference system isomorphic to $\cH_L$. $\omega_{L,R}=\ketbra{\omega}{\omega}_{L,R}$ is any pure state in the joint system. Thus $f_{\mathrm{worst}}$ quantifies the best achievable worst-case entanglement fidelity after optimizing over recovery channels and taking the worst case over all normalized input states in the joint system, including entangled states. The optimal worst-case recovery error is~\cite{faist2020ContinuousSymmetriesApproximatea}
\begin{equation}\label{eq:def_worst-case_recovery_error}
\varepsilon_{\rm worst}(\cN\circ\cE)=\sqrt{1-f_{\mathrm{worst}}(\mathcal N\circ\mathcal E)^2}.
\end{equation}
It quantifies how well the encoded logical information can be recovered after noise, assuming the recovery channel is chosen optimally. The worst case is taken over all logical input states, including states entangled with an arbitrary reference system, so it captures the most adversarial loss of quantum information.

Let system $\cH_A$ be the Fock space of $n$ fermionic modes. For a real antisymmetric matrix
$\mathsf{A}=-\mathsf{A}^{\mathsf T}\in\mathbb R^{2n\times2n}$, define the corresponding Hermitian quadratic Majorana observable
\begin{equation}\label{eq:def_quadratic_majorana_observable}
\mathsf Q(\mathsf{A})
\coloneqq
i\sum_{j,k=1}^{2n}
\mathsf{A}_{j,k}\gamma_j\gamma_k.
\end{equation}
We define the
\emph{quadratic reconstruction constant} for $\cN$ as
\begin{equation}\label{eq:def_quadratic_reconstruction_constant}
\begin{aligned}
\mu_2(\cN)
\coloneqq
\sup_{\substack{
\mathsf{A}=-\mathsf{A}^{\mathsf T}\\
\norm{\mathsf{A}}_\infty=1
}}
\inf_{\substack{
O=O^\dagger,c\in\mathbb R\\
\widehat{\cN}^{\dagger}(O)
=
\mathsf Q(\mathsf{A})+c\bI
}}
\norm{O}_\infty,
\end{aligned}
\end{equation}
where the inner infimum is $\infty$ if the constraint set is empty. The value of $\mu_2(\mathcal N)$ is independent of the particular Stinespring dilation used to define the complementary channel $\widehat{\cN}$ and thus $\widehat{\cN}^\dagger$, since different dilations are related by an isometry on the environment and the operator norm is preserved under the corresponding embedding. The quadratic reconstruction constant quantifies the accessibility of the environment to the quadratic information of the system, and can be understood as follows. If an observable $O\in\cL(\cH_E)$ satisfies the constraints, then
\begin{equation}
\Tr(O\widehat{\cN}(\rho))=\Tr(\mathsf{Q}(\mathsf{A})\rho)+c,
\end{equation}
and $O$ measures the expectation value of a quadratic observable $\mathsf{Q}(\mathsf{A})$ in system $\cH_A$ from the environment system $\cH_E$, up to a constant offset. Therefore, $\mu_2$ is the worst-case operator norm amplification needed for the environment $\cH_E$ to reconstruct all normalized quadratic Majorana directions. A small value of $\mu_2$ implies a strong and well-conditioned leakage of the covariance information to the environment, while $\mu_2=\infty$ means that at least one quadratic direction is not reconstructible.

We now state the central theorem, which establishes a quantitative tradeoff between the Gaussian fidelity of a fermionic code space, the channel’s quadratic reconstruction constant, and the optimal worst-case recovery error.
\begin{theorem}[Gaussianity--approximate-correctability tradeoff]\label{thm:general_noise_Gaussianity_tradeoff}
Let $\cC$ be a fermionic code on $n$ modes with $\dim(\cC)\geq2$, and let
$\cE:\sigma_L\mapsto V\sigma_LV^\dagger$ be an isometric encoding into $\cC$.
For every noise channel $\cN$ satisfying $\mu_2(\cN)<\infty$,
the Gaussian fidelity of the code obeys
\begin{equation}\label{eq:general_noise_Gaussianity_correctability_tradeoff}
F_0^{(n)}(\cC)
\leq
1-\frac1n
+
\frac{\mu_2(\cN)}{n}
\epsilon_{\mathrm{worst}}
(\cN\circ\cE).
\end{equation}
\end{theorem}
The proof of Theorem~\ref{thm:general_noise_Gaussianity_tradeoff} is left to the end of this appendix. This tradeoff shows that, whenever quadratic information about the input is accessible with controlled operator norm amplification to the environment, good approximate correctability forces the code space to remain quantitatively separated from the Gaussian manifold. In the exact-correction limit, the channel dependence disappears entirely, yielding $F_0^{(n)}(\cC)\leq 1-1/n$.

We then discuss three examples of fermionic noise models and provide the upper bounds for their quadratic reconstruction constants $\mu_2$, so that we can apply Theorem~\ref{thm:general_noise_Gaussianity_tradeoff} to them conveniently. Their proofs are left to the end of this appendix.

The first example is motivated by the quasiparticle noise (Qp) model~\cite{knapp2018ModelingNoiseError}, which describes stochastic single-Majorana quasiparticle poisoning together with pair-wise Majorana dephasing events in Majorana-zero-mode architectures. Such stochastic Majorana models were introduced to connect microscopic error processes in topological superconductor devices with fault-tolerance analyses.
\begin{proposition}[Quadratic reconstruction constant bound of quasiparticle noise]\label{prop:mu_2_Qp}
For a quasiparticle noise on $n$ fermionic modes, let $p_{\mathrm{qp}}\geq0$, $p_{\mathrm{pair}}\geq0$, and $0<p_{\mathrm{qp}}+p_{\mathrm{pair}}\leq1$, the noise channel is
\begin{equation}\label{eq:qp_noise_channel}
\cN_{\mathrm{Qp}}(\rho)
=(1-p_{\mathrm{qp}}-p_{\mathrm{pair}})\rho
+\frac{p_{\mathrm{qp}}}{2n}
\sum_{j=1}^{2n}\gamma_j\rho\gamma_j+\frac{p_{\mathrm{pair}}}{4n^2}
\sum_{k,j=1}^{2n}
\gamma_k\gamma_j\rho\gamma_j\gamma_k.
\end{equation}
Then the quadratic reconstruction constant satisfies
\begin{equation}
\mu_2(\cN_{\mathrm{Qp}})\leq\frac{2n}{p_{\mathrm{qp}}+p_{\mathrm{pair}}}.
\end{equation}
\end{proposition}

The second example is global depolarizing noise, a standard benchmark model for incoherent noise in digital quantum simulations, and has also
been considered in fermionic classical-shadow
protocols~\cite{wu2024ErrormitigatedFermionicClassical}.

\begin{proposition}[Quadratic reconstruction constant bound of global depolarizing noise]
\label{prop:mu_2_depolarizing}
For an $n$-mode fermionic system, let $p\in(0,1]$ and consider the
global depolarizing channel
\begin{equation}\label{eq:global_depolarizing_noise_channel}
\cD_p(\rho)
=
(1-p)\rho
+
p\frac{\bI}{2^n}.
\end{equation}
Then the quadratic reconstruction constant satisfies
\begin{equation}\label{eq:mu_2_depolarizing_bound}
\mu_2(\cD_p)
\leq
\frac{2n}{p}.
\end{equation}
\end{proposition}

The third example is the class of fermionic Gaussian noise channels, which describe open-system fermionic dynamics generated by Gaussian system--environment interactions with Gaussian ancillary states. This class includes physically relevant processes such as fermionic attenuation and particle loss~\cite{greplova2018DegradabilityFermionicGaussian}, and is naturally characterized by its linear action on covariance matrices, see details and Lemma~\ref{lem:Gaussian_channel_covariance_matrix} in Appendix~\ref{app:fermionic_Gaussian_maps}.

\begin{proposition}[Quadratic reconstruction constant bound of fermionic Gaussian noise]\label{prop:mu_2_Gaussian}
Let
$\cN_\mathrm{G}:\cL(\cH_A)\rightarrow\cL(\cH_B)$ be a fermionic Gaussian channel mapping from an $n$-mode fermionic system to an $m$-mode fermionic system, whose action on covariance matrices is (see Lemma~\ref{lem:Gaussian_channel_covariance_matrix})
\begin{equation}\label{eq:Gaussian_channel_covariance_matrix_mu_2}
\Gamma(\rho)\mapsto\Gamma(\cN_\mathrm{G}(\rho))=\mathsf{M}\Gamma(\rho)\mathsf{M}^{\mathsf{T}}+\mathsf{N},
\end{equation}
for some $\mathsf{M}\in\mathbb{R}^{2m\times2n}$, $\mathsf{N}\in\mathbb{R}^{2m\times2m}$ and $\mathsf{N}^\mathsf{T}=-\mathsf{N}$.
If
$\norm{\mathsf{M}}_\infty<1$,
then the quadratic reconstruction constant satisfies
\begin{equation}\label{eq:mu_2_Gaussian_noise_bound}
\mu_2(\cN_{\mathrm{G}})
\leq
\Tr[
\left(
\bI_{2n}-\mathsf{M}^{\mathsf T}\mathsf{M}
\right)^{-1}
]
\leq
\frac{2n}{1-\norm{\mathsf{M}}_\infty^2}.
\end{equation}
\end{proposition}

Before we prove Theorem~\ref{thm:general_noise_Gaussianity_tradeoff}, we review two useful lemmas.

\begin{lemma}[Parent Hamiltonian for pure Gaussian states]\label{lem:parent_Hamiltonian}
Let $\ket{\phi}\in\cG_{\textup{pure}}^{(n)}$ be an $n$-mode fermionic Gaussian state vector and denote $\phi=\ketbra{\phi}{\phi}$, define the quadratic Hamiltonian (cf. Eq.~\eqref{eq:def_quadratic_majorana_observable})
\begin{equation}
H_\phi\coloneqq\frac n2\bI+\frac i4\sum_{j,k=1}^{2n}\Gamma_{j,k}(\phi)\gamma_j\gamma_k
=\frac n2\bI+\mathsf{Q}\left(\frac{\Gamma(\phi)}4\right),
\end{equation}
where $\Gamma(\phi)$ is the covariance matrix defined in Eq.~\eqref{eq:correlation_matrix}. Then
\begin{equation}
\ker(H_\phi)=\operatorname{span}(\ket{\phi}),\qquad
\operatorname{spec}(H_\phi)=\{0,1,\cdots,n\},
\end{equation}
and
\begin{equation}\label{eq:parent_Hamiltonian_bounds}
\bI-\phi\leq H_\phi\leq n(\bI-\phi).
\end{equation}
\end{lemma}

The construction above coincides with the positive parent Hamiltonian associated with the fermionic Gaussian fidelity witness of Ref.~\cite{gluza2018FidelityWitnessesFermionic}: for the witness $\mathcal W$ constructed for the target state $\ket{\phi}$ in their
Eq.~(7), we have $H_\phi=\bI-\mathcal W$.

\begin{proof}[Proof of Lemma~\ref{lem:parent_Hamiltonian}]
For the pure Gaussian state $\ket{\phi}\in\cG_{\textup{pure}}^{(n)}$, there exists a Gaussian unitary $U_R$ with some $R\in\mathrm{O}(2n)$ such that $\ket{\phi}=U_R\ket{0^n}$. Recall that the covariance matrix of $\ketbra{0^n}{0^n}$ is
\begin{equation}
J^{(n)}=\Gamma\left(\ketbra{0^n}{0^n}\right)=\bigoplus_{j=1}^n
\begin{bmatrix}
0&1\\
-1&0
\end{bmatrix},
\end{equation}
and $\Gamma(\phi)=RJ^{(n)}R^\mathsf{T}$ by Eq.~\eqref{eq:Gaussian_unitary_covariance_matrix}. Recall that the total particle number operator is
\begin{equation}
N=\sum_{j=1}^n\ket{1}_j\bra{1}_j=\frac n2\bI+\frac i2\sum_{j=1}^n\gamma_{2j-1}\gamma_{2j}=\frac n2\bI+\frac i4\sum_{\mu,\nu=1}^{2n}J^{(n)}_{\mu,\nu}\gamma_\mu\gamma_\nu.
\end{equation}
The canonical occupation number basis is an eigenbasis of $N$, and
\begin{equation}
\ker(N)=\operatorname{span}(\ket{0^n}),\qquad
\operatorname{spec}(N)=\{0,1,\cdots,n\},
\end{equation}
we have
\begin{equation}
\begin{aligned}
H_\phi&=\frac n2\bI+\frac i4\sum_{j,k=1}^{2n}\Gamma_{j,k}(\phi)\gamma_j\gamma_k
\\&=\frac n2\bI+\frac i4\sum_{j,k=1}^{2n}\left(RJ^{(n)}R^\mathsf{T}\right)_{j,k}\gamma_j\gamma_k
\\&=\frac n2\bI+\frac i4\sum_{\mu,\nu=1}^{2n}J^{(n)}_{\mu,\nu}\sum_{j,k=1}^{2n}R_{j,\mu}R_{k,\nu}\gamma_j\gamma_k
\\&=\frac n2\bI+\frac i4\sum_{\mu,\nu=1}^{2n}J^{(n)}_{\mu,\nu}U_R\gamma_\mu\gamma_\nu U_R^\dagger
\\&=U_RNU_R^\dagger,
\end{aligned}
\end{equation}
where we use Eq.~\eqref{eq:Gaussian_unitary_action}. So
\begin{equation}
\ker(H_\phi)=\operatorname{span}(U_R\ket{0^n})=\operatorname{span}(\ket{\phi}),\qquad
\operatorname{spec}(H_\phi)=\{0,1,\cdots,n\}.
\end{equation}
Finally, $H_\phi$ has eigenvalue zero on $\ket{\phi}$ and eigenvalues in $\{1,\cdots,n\}$ on its orthogonal complement.
Since
$\bI-\ket{\phi}\bra{\phi}$ is the orthogonal projector onto
that complement, it follows that
\begin{equation}
    \bI-\ket{\phi}\bra{\phi}
    \leq
    H_\phi
    \leq
    n\left(\bI-\ket{\phi}\bra{\phi}\right).
\end{equation}
This completes the proof.
\end{proof}

The second lemma shows that approximate correctability limits how well the complementary environment can distinguish different logical states, see Lemma~28 and its proof in Ref.~\cite{faist2020ContinuousSymmetriesApproximatea}.
\begin{lemma}[Complementary environment indistinguishability,~\cite{faist2020ContinuousSymmetriesApproximatea}]\label{lem:complementary_environment_indistinguishability}
For any encoding channel $\cE$ and noise channel $\cN$, any complementary channel $\widehat{\cN\circ\cE}$ of $\cN\circ\cE$, and any two logical states $\sigma_L,\sigma_L'$, we have
\begin{equation}
\delta\left(\widehat{\cN\circ\cE}(\sigma_L),\widehat{\cN\circ\cE}(\sigma_L')\right)\leq2\varepsilon_{\textup{worst}}\left(\cN\circ\cE\right),
\end{equation}
where $\delta(\rho_1,\rho_2)=\norm{\rho_1-\rho_2}_1/2$ is the trace distance between two density matrices and $\varepsilon_{\textup{worst}}$ is defined in Eq.~\eqref{eq:def_worst-case_recovery_error}.
\end{lemma}

Now we prove Theorem~\ref{thm:general_noise_Gaussianity_tradeoff}.
\begin{proof}[Proof of Theorem~\ref{thm:general_noise_Gaussianity_tradeoff}]
Denote $F=F_0^{(n)}(\cC)$ and since the pure fermionic Gaussian manifold is compact in finite dimension, there exists $\ket{\phi}\in\cG_{\textup{pure}}^{(n)}$ such that $\bra{\phi}P\ket{\phi}=F$, where $P=VV^\dagger$ is the orthogonal projector onto $\cC$ (cf. Eq.~\eqref{eq:def_code_Gaussian_fidelity}). Notice that $F\geq\operatorname{rank}(P)/2^n>0$ and define the codeword
\begin{equation}
\ket{\psi_1}=\frac{P\ket{\phi}}{\sqrt{F}},
\end{equation}
so $\abs{\bra{\phi}\ket{\psi_1}}^2=F$.
Since $\dim(\cC)>1$, choose a codeword $\ket{\psi_2}$ orthogonal to $\ket{\psi_1}$, so that
\begin{equation}
\bra{\phi}\ket{\psi_2}=\bra{\phi}P\ket{\psi_2}=\sqrt{F}\bra{\psi_1}\ket{\psi_2}=0.
\end{equation}
Apply Lemma~\ref{lem:parent_Hamiltonian} to $\ket{\phi}$ and construct the corresponding Hamiltonian $H_\phi$. Write $\mathsf{A}_\phi=\Gamma(\phi)/4$, so $\norm{\mathsf{A}_\phi}_\infty=1/4$ and
\begin{equation}
H_\phi=\frac n2\bI+\mathsf{Q}(\mathsf{A}_\phi).
\end{equation}
Denote $\psi_{1,2}=\ketbra{\psi_{1,2}}{\psi_{1,2}}$, Eq.~\eqref{eq:parent_Hamiltonian_bounds} implies
\begin{equation}\label{eq:H_phi_lower_bound}
\Tr(H_\phi\psi_2)\geq1,\qquad
\Tr(H_\phi\psi_1)\leq n(1-F).
\end{equation}

For the noise channel $\cN$, denote its complementary channel associated with a minimal Stinespring dilation with isometry $W_\cN$ as $\widehat{\cN}$, see Eq.~\eqref{eq:def_complementary_channel}. 
Fix an arbitrary $\eta>0$, since $\mu_2(\cN)<\infty$, there exist an observable $\widetilde{O}\in\cL(\cH_E)$ and $\widetilde{c}\in\mathbb{R}$, such that (cf. Eq.~\eqref{eq:def_quadratic_reconstruction_constant})
\begin{equation}
\widehat{\cN}^\dagger(\widetilde{O})=\mathsf{Q}(\Gamma(\phi))+\widetilde{c}\bI,\qquad
\norm{\widetilde{O}}_\infty\leq\mu_2(\cN)+\eta.
\end{equation}
Let $O_\phi=\widetilde{O}/4$ and $c_\phi=\widetilde{c}/4$, then
\begin{equation}
\widehat{\cN}^\dagger(O_\phi)=\mathsf{Q}(\mathsf{A}_\phi)+c_\phi\bI,\qquad
\norm{O_\phi}_\infty\leq\frac{\mu_2(\cN)+\eta}4.
\end{equation}
Using H\"{o}lder's inequality, we have
\begin{equation}
\begin{aligned}
\Tr(H_\phi\psi_2)-\Tr(H_\phi\psi_1)&=\Tr(\mathsf{Q}(\mathsf{A}_\phi)(\psi_2-\psi_1))
\\&=\Tr(O_\phi\left(\widehat{\cN}(\psi_2)-\widehat{\cN}(\psi_1)\right))
\\&\leq\norm{O_\phi}_\infty\norm{\widehat{\cN}(\psi_2)-\widehat{\cN}(\psi_1)}_1
\\&\leq\frac{\mu_2(\cN)+\eta}2\delta\left(\widehat{\cN}(\psi_2),\widehat{\cN}(\psi_1)\right).
\end{aligned}
\end{equation}
Using the encoding isometry $V$ fixed above, define $\sigma_L^{(j)}=V^\dagger\psi_jV$. Since $P\psi_jP=\psi_j$, these are normalized logical states and $\cE\left(\sigma_L^{(j)}\right)=\psi_j$. So $W_\cN V$ is a Stinespring isometry for $\cN\circ\cE$, and $\widehat{\cN}\circ\cE$ is a complementary channel of $\cN\circ\cE$. Applying Lemma~\ref{lem:complementary_environment_indistinguishability} and inserting Eq.~\eqref{eq:H_phi_lower_bound}, we obtain
\begin{equation}
\begin{aligned}
1-n(1-F)&\leq\Tr(H_\phi\psi_2)-\Tr(H_\phi\psi_1)
\\&\leq\frac{\mu_2(\cN)+\eta}2\delta\left(\widehat{\cN}(\psi_2),\widehat{\cN}(\psi_1)\right)
\\&\leq\left(\mu_2(\cN)+\eta\right)\varepsilon_{\textup{worst}}\left(\cN\circ\cE\right).
\end{aligned}
\end{equation}
Since $\eta>0$ is arbitrary, we obtain
\begin{equation}
F\leq
1-\frac1n
+
\frac{\mu_2(\cN)}{n}
\varepsilon_{\mathrm{worst}}
(\cN\circ\cE),
\end{equation}
which proves Eq.~\eqref{eq:general_noise_Gaussianity_correctability_tradeoff}.
\end{proof}

Here are two lemmas used in the proofs of Propositions~\ref{prop:mu_2_Qp},~\ref{prop:mu_2_depolarizing}, and~\ref{prop:mu_2_Gaussian}. The first lemma shows the operator norm of the quadratic Majorana observable defined in Eq.~\eqref{eq:def_quadratic_majorana_observable} equals the trace norm of the corresponding antisymmetric matrix.
\begin{lemma}[Operator norm of quadratic Majorana observables]\label{lem:quadratic_Majorana_norm}
For any real antisymmetric matrix
$\mathsf A\in\mathbb R^{2r\times 2r}$ and the quadratic Majorana observable
\begin{equation}\label{eq:def_quadratic_majorana_observable_norm}
\mathsf Q(\mathsf{A})
\coloneqq
i\sum_{j,k=1}^{2r}
\mathsf{A}_{j,k}\gamma_j\gamma_k,
\end{equation}
we have
\begin{equation}
\norm{\mathsf Q(\mathsf A)}_\infty
=
\norm{\mathsf A}_1.
\end{equation}
\end{lemma}

\begin{proof}[Proof of Lemma~\ref{lem:quadratic_Majorana_norm}]
By the real canonical form of antisymmetric matrices, there exists
$R\in\mathrm O(2r)$ such that
\begin{equation}
R^{\mathsf T}\mathsf A R
=
\bigoplus_{j=1}^{r}
\begin{bmatrix}
0&a_j\\
-a_j&0
\end{bmatrix},
\qquad
a_j\geq0.
\end{equation}
Defining rotated Majorana operators
$\widetilde{\gamma}_a=\sum_b R_{b,a}\gamma_b$, we obtain
\begin{equation}
\mathsf Q(\mathsf A)
=
2\sum_{j=1}^{r}
a_ji\widetilde{\gamma}_{2j-1}\widetilde{\gamma}_{2j}.
\end{equation}
The operators
$i\widetilde{\gamma}_{2j-1}\widetilde{\gamma}_{2j}$
are mutually commuting Hermitian involutions and can be regarded as Pauli-$Z_j$ operators of the rotated canonical fermionic modes up to a sign convention. So the eigenvalues of
$\mathsf Q(\mathsf A)$ are
$2\sum_j s_j a_j$ with $s_j\in\{\pm1\}$. Hence
\begin{equation}
\norm{\mathsf Q(\mathsf A)}_\infty
=
2\sum_{j=1}^{r}a_j.
\end{equation}
Notice that the singular values of $\mathsf A$ are
$a_1,a_1,\ldots,a_r,a_r$, so we prove
\begin{equation}
\norm{\mathsf Q(\mathsf A)}_\infty
=
\norm{\mathsf A}_1.
\end{equation}
\end{proof}

The second lemma gives an upper bound for the quadratic reconstruction constant $\mu_2$ of a noise channel in its Kraus representation.
\begin{lemma}[Kraus-space reconstruction of quadratic observables]\label{lem:Kraus_space_reconstruction}
Let $\cN:\cL(\cH_A)\longrightarrow\cL(\cH_B)$ be a finite-dimensional quantum channel, where $\cH_A$ is the Fock space of $n$ fermionic modes, and let
\begin{equation}
\cN(\rho)=\sum_{\alpha=1}^{r}K_\alpha\rho K_\alpha^\textnormal{\textdagger}
\end{equation}
be an arbitrary Kraus representation of $\cN$, with $\sum_{\alpha=1}^{r}K_\alpha^\textnormal{\textdagger} K_\alpha=\bI_A$. Introduce an $r$-dimensional environment $\cH_E=\operatorname{span}\{\ket{\alpha}\}_{\alpha=1}^{r}$. Suppose that there exists a constant $C<\infty$ such that, for every real antisymmetric matrix $\mathsf A$ with $\norm{\mathsf A}_\infty=1$, there exist an observable $O\in\cL(\cH_E)$ and $c\in\mathbb{R}$ satisfying
\begin{equation}\label{eq:Kraus_quadratic_reconstruction}
\sum_{\alpha,\beta=1}^{r}\bra{\beta}O\ket{\alpha}K_\beta^\textnormal{\textdagger} K_\alpha=\mathsf Q(\mathsf A)+c\bI_A,
\end{equation}
with $\norm{O}_\infty\leq C$ and $\mathsf{Q}(\mathsf{A})$ defined in Eq.~\eqref{eq:def_quadratic_majorana_observable}.
Then the quadratic reconstruction constant obeys
\begin{equation}
\mu_2(\cN)\leq C.
\end{equation}
Conversely, if there exists a real antisymmetric matrix $\mathsf A$ with $\norm{\mathsf A}_\infty=1$ such that
\begin{equation}
\mathsf Q(\mathsf A)
\notin
\operatorname{span}
\left\{
K_\beta^\textnormal{\textdagger} K_\alpha\middle\mid
\alpha,\beta\in[r]
\right\}
+\mathbb R\bI_A,
\end{equation}
then
\begin{equation}
\mu_2(\cN)=\infty.
\end{equation}
\end{lemma}
\begin{proof}[Proof of Lemma~\ref{lem:Kraus_space_reconstruction}]
We first derive the complementary adjoint associated with the given
Kraus representation. The Stinespring isometry is
$W
=
\sum_{\alpha=1}^{r}
K_\alpha\otimes\ket{\alpha}$, and
tracing out the system output gives the complementary channel
\begin{equation}\label{eq:Kraus_complementary_channel}
\widehat{\cN}(\rho)=\Tr_B(W\rho W^\dagger)
=
\sum_{\alpha,\beta=1}^{r}
\Tr(
K_\alpha\rho K_\beta^\dagger
)
\ketbra{\alpha}{\beta}.
\end{equation}
For an arbitrary observable $O\in\cL(\cH_E)$, denote $\bra{\beta}O\ket{\alpha}=O_{\beta,\alpha}$ and we obtain
\begin{equation}
\begin{aligned}
\Tr(
O\widehat{\cN}(\rho)
)
&=
\sum_{\alpha,\beta=1}^{r}
O_{\beta,\alpha}
\Tr(
K_\alpha\rho K_\beta^\dagger
)\\
&=
\Tr[
\left(
\sum_{\alpha,\beta=1}^{r}
O_{\beta,\alpha}
K_\beta^\dagger K_\alpha
\right)
\rho
].
\end{aligned}
\end{equation}
By the definition of the adjoint channel, 
\begin{equation}\label{eq:Kraus_complementary_adjoint}
\widehat{\cN}^\dagger(O)=\sum_{\alpha,\beta=1}^{r}
O_{\beta,\alpha}
K_\beta^\dagger K_\alpha.
\end{equation}
Suppose that there exists a constant $C<\infty$ such that, for every real antisymmetric matrix $\mathsf A$ with $\norm{\mathsf A}_\infty=1$, there exist an  observable $O\in\cL(\cH_E)$ with $\norm{O}_\infty\leq C$ and $c\in\mathbb{R}$ satisfying Eq.~\eqref{eq:Kraus_quadratic_reconstruction},
then
\begin{equation}
\widehat{\cN}^\dagger(O)
=
\mathsf Q(\mathsf A)+c\bI_A.
\end{equation}
Taking the supremum over all
$\mathsf A=-\mathsf A^{\mathsf T}$ with
$\norm{\mathsf A}_\infty=1$ gives
\begin{equation}
\mu_2(\cN)\leq C.
\end{equation}

Finally, Eq.~\eqref{eq:Kraus_complementary_adjoint} shows that the range of $\widehat{\cN}^\dagger$ is contained in $\operatorname{span}
\left\{
K_\beta^\dagger K_\alpha\middle\mid
\alpha,\beta\in[r]
\right\}$.
Hence, if for some normalized antisymmetric $\mathsf A$ there is no
real scalar $c$ such that
\begin{equation}
\mathsf Q(\mathsf A)+c\bI_A
\in
\operatorname{span}
\left\{
K_\beta^\dagger K_\alpha\middle\mid
\alpha,\beta\in[r]
\right\},
\end{equation}
then no environment observable can reconstruct this quadratic
direction, even up to an additive scalar. The corresponding inner
infimum in the definition of $\mu_2(\cN)$ is therefore infinite,
which implies
\begin{equation}
\mu_2(\cN)=\infty.
\end{equation}
\end{proof}

We then provide the proofs of Propositions~\ref{prop:mu_2_Qp},~\ref{prop:mu_2_depolarizing},and~\ref{prop:mu_2_Gaussian}.
\begin{proof}[Proof of Proposition~\ref{prop:mu_2_Qp}]
We apply Lemma~\ref{lem:Kraus_space_reconstruction} and follow the setting therein. Consider the Kraus representation of $\cN_{\mathrm{Qp}}$ with Kraus operators
\begin{equation}\label{eq:Qp_Kraus_operators}
\begin{aligned}
K_0
&=
\sqrt{1-p_{\mathrm{qp}}-p_{\mathrm{pair}}}\bI,\\
K_j
&=
\sqrt{\frac{p_{\mathrm{qp}}}{2n}}\gamma_j,
\qquad j\in[2n],\\
K_{k,j}
&=
\frac{\sqrt{p_{\mathrm{pair}}}}{2n}
\gamma_k\gamma_j,
\qquad k,j\in[2n].
\end{aligned}
\end{equation}
Introduce the corresponding environment system
\begin{equation}
\cH_E=\operatorname{span}\{\ket{0}\}\oplus
\operatorname{span}\left\{\ket{j}_{\mathrm{qp}}\middle\mid j\in[2n]\right\}\oplus\bigoplus_{k=1}^{2n}
\operatorname{span}\left\{\ket{k,j}_{\mathrm{pair}}\middle\mid j\in[2n]\right\},
\end{equation}
and set
$C=2n/(p_{\mathrm{qp}}+p_{\mathrm{pair}})$.
For any real antisymmetric matrix $\mathsf{A}$ with $\norm{\mathsf{A}}_\infty=1$, define the Hermitian operator $O\in\cL(\cH_E)$ as
\begin{equation}
\begin{aligned}
O&=iC\sum_{j,l=1}^{2n}A_{j,l}\ketbra{j}{l}_{\mathrm{qp}}+iC\sum_{k,j,l=1}^{2n}A_{j,l}\ketbra{k,j}{k,l}_{\mathrm{pair}}
\\&=0\oplus iC\mathsf{A}\oplus\bigoplus_{k=1}^{2n}iC\mathsf{A},
\end{aligned}
\end{equation}
then $\norm{O}_\infty=\norm{iC\mathsf{A}}_\infty=C$. It remains to verify Eq.~\eqref{eq:Kraus_quadratic_reconstruction} (cf. Eq.~\eqref{eq:def_quadratic_majorana_observable}):
\begin{equation}
\begin{aligned}
\sum_{j,l=1}^{2n}\bra{j}O\ket{l}_{\mathrm{qp}}K_j^\dagger K_l+
\sum_{k,j,l=1}^{2n}\bra{k,j}O\ket{k,l}_{\mathrm{pair}}K_{k,j}^\dagger K_{k,l}
&=iC\frac{p_{\mathrm{qp}}}{2n}\sum_{j,l=1}^{2n}\mathsf{A}_{j,l}\gamma_j\gamma_l
+iC\frac{p_{\mathrm{pair}}}{4n^2}\sum_{k,j,l=1}^{2n}\mathsf{A}_{j,l}(\gamma_k\gamma_j)^\dagger\gamma_k\gamma_l
\\&=C\left(\frac{p_{\mathrm{qp}}}{2n}+\frac{p_{\mathrm{pair}}}{2n}\right)\mathsf{Q}(\mathsf{A})
\\&=\mathsf{Q}(\mathsf{A}).
\end{aligned}
\end{equation}
So we have
\begin{equation}
\mu_2(\cN_{\mathrm{Qp}})\leq C=\frac{2n}{p_{\mathrm{qp}}+p_{\mathrm{pair}}}.
\end{equation}
\end{proof}

\begin{proof}[Proof of Proposition~\ref{prop:mu_2_depolarizing}]
We apply Lemma~\ref{lem:Kraus_space_reconstruction} and follow the
setting therein. Let $d=2^n$ and fix an orthonormal basis
$\{\ket{u}\}_{u=1}^{d}$ of the fermionic Hilbert space.
Consider the Kraus representation of $\cD_p$ with Kraus operators
\begin{equation}\label{eq:depolarizing_Kraus_operators}
K_0=\sqrt{1-p}\bI,
\qquad
K_{u,v}
=
\sqrt{\frac{p}{d}}\ketbra{u}{v},
\qquad
u,v\in[d].
\end{equation}
Introduce the corresponding environment system
\begin{equation}
\cH_E
=
\operatorname{span}\{\ket{0}\}
\oplus
\operatorname{span}
\left\{
\ket{u,v}\middle\mid u,v\in[d]
\right\},
\end{equation}
and set $C=2n/p$.
For any real antisymmetric matrix $\mathsf A$ with
$\norm{\mathsf A}_\infty=1$, define the Hermitian operator
$O\in\cL(\cH_E)$ by (cf. Eq.~\eqref{eq:def_quadratic_majorana_observable})
\begin{equation}\label{eq:depolarizing_environment_observable}
\begin{aligned}
O
&=
\frac{1}{p}
\sum_{u=1}^{d}
\sum_{v,w=1}^{d}
\bra{v}\mathsf Q(\mathsf A)\ket{w}
\ketbra{u,v}{u,w}
\\
&=
0\oplus
\frac{1}{p}
\bigoplus_{u=1}^{d}
\mathsf Q(\mathsf A).
\end{aligned}
\end{equation}
Using Lemma~\ref{lem:quadratic_Majorana_norm},
we have
\begin{equation}
\norm{O}_\infty
=
\frac{1}{p}
\norm{\mathsf Q(\mathsf A)}_\infty
=
\frac{1}{p}\norm{\mathsf A}_1
\leq
\frac{2n}{p}\norm{\mathsf A}_\infty
=
C.
\end{equation}

It remains to verify
Eq.~\eqref{eq:Kraus_quadratic_reconstruction}:
\begin{equation}
\sum_{u,v,w=1}^{d}
\bra{u,v}O\ket{u,w}
K_{u,v}^{\dagger}K_{u,w}
=
\frac{1}{p}\frac{p}{d}
\sum_{u,v,w=1}^{d}
\bra{v}\mathsf Q(\mathsf A)\ket{w}
\ketbra{v}{w}
=
\mathsf Q(\mathsf A).
\end{equation}
Thus, for every real antisymmetric $\mathsf A$ with
$\norm{\mathsf A}_\infty=1$, there exists an environment observable
$O$ satisfying the quadratic reconstruction condition with
$\norm{O}_\infty\leq C$. Therefore,
\begin{equation}
\mu_2(\cD_p)
\leq
C
=
\frac{2n}{p}.
\end{equation}
\end{proof}

\begin{proof}[Proof of Proposition~\ref{prop:mu_2_Gaussian}]
We use the same notations and Gaussian dilation as in the proof of
Lemma~\ref{lem:Gaussian_channel_covariance_matrix}.  Write the
orthogonal matrix associated with the Gaussian unitary $U_R$ as
\begin{equation}
R=
\begin{bmatrix}
\mathsf M & \mathsf L\\
\mathsf C & \mathsf D
\end{bmatrix},
\end{equation}
where the first $2m$ rows,
$\begin{bmatrix}\mathsf M&\mathsf L\end{bmatrix}$,
correspond to the output system $\cH_B$, while the remaining rows,
$\begin{bmatrix}\mathsf C&\mathsf D\end{bmatrix}$,
correspond to $E_{\mathrm{out}}$.
Since $R$ is orthogonal, its first $2n$ columns satisfy
\begin{equation}\label{eq:Gaussian_complement_MC_relation}
\mathsf M^{\mathsf T}\mathsf M
+
\mathsf C^{\mathsf T}\mathsf C
=
\bI_{2n}.
\end{equation}
The covariance
matrix of the complementary output on $E_{\mathrm{out}}$ is
\begin{equation}\label{eq:Gaussian_complement_covariance_action}
\Gamma\left(
\widehat{\cN}_{\mathrm G}(\rho)
\right)
=
\mathsf C\Gamma(\rho)\mathsf C^{\mathsf T}
+
\mathsf N_E,
\end{equation}
where
$\mathsf N_E
=
\mathsf D
\Gamma(\tau_{E_{\mathrm{in}}})
\mathsf D^{\mathsf T}$ using the notation in the proof of Lemma~\ref{lem:Gaussian_channel_covariance_matrix}.

Since $\norm{\mathsf M}_\infty<1$,
Eq.~\eqref{eq:Gaussian_complement_MC_relation} implies
$\mathsf C^{\mathsf T}\mathsf C
=
\bI_{2n}-\mathsf M^{\mathsf T}\mathsf M
>0$.
Therefore, $\mathsf C$ has full column rank and admits the left inverse
$\mathsf C^+
=
(\mathsf C^{\mathsf T}\mathsf C)^{-1}
\mathsf C^{\mathsf T}$.
For any real antisymmetric matrix $\mathsf A$ with
$\norm{\mathsf A}_\infty=1$, define
$\mathsf A_E
=
(\mathsf C^+)^{\mathsf T}
\mathsf A
\mathsf C^+$.
Then $\mathsf A_E^{\mathsf T}=-\mathsf A_E$ and
$\mathsf C^{\mathsf T}\mathsf A_E\mathsf C
=
(\mathsf C^+\mathsf C)^{\mathsf T}
\mathsf A
(\mathsf C^+\mathsf C)
=
\mathsf A$.
Define the Hermitian operator
$O=\mathsf{Q}(\mathsf{A}_E)\in\cL(\cH_E)$ as in Eq.~\eqref{eq:def_quadratic_majorana_observable} with Majorana operators in the $E_{\textup{out}}$.
Using Eq.~\eqref{eq:Gaussian_complement_covariance_action} and
$\Tr[\mathsf Q(\mathsf B)\rho]
=-\Tr[\mathsf B^{\mathsf T}\Gamma(\rho)]$ for any real
antisymmetric matrix $\mathsf B$, we have
\begin{equation}
\begin{aligned}
\Tr[
O\widehat{\cN}_{\mathrm G}(\rho)
]
&=
-\Tr[
\mathsf A_E^{\mathsf T}
\left(
\mathsf C\Gamma(\rho)\mathsf C^{\mathsf T}
+\mathsf N_E
\right)
]
\\
&=
-\Tr[
\left(
\mathsf C^{\mathsf T}
\mathsf A_E
\mathsf C
\right)^{\mathsf T}
\Gamma(\rho)
]
-
\Tr(
\mathsf A_E^{\mathsf T}\mathsf N_E
)
\\
&=
-\Tr[
\mathsf A^{\mathsf T}\Gamma(\rho)
]
+c
\\
&=
\Tr[
\mathsf Q(\mathsf A)\rho
]
+c,
\end{aligned}
\end{equation}
where
$c
=-\Tr(
\mathsf A_E^{\mathsf T}\mathsf N_E
)
\in\mathbb R$ is a constant.
Since this equality holds for every input state $\rho$, we obtain
\begin{equation}
\widehat{\cN}_{\mathrm G}^{\dagger}(O)
=
\mathsf Q(\mathsf A)+c\bI.
\end{equation}

It remains to bound $\norm{O}_\infty$. 
Using Lemma~\ref{lem:quadratic_Majorana_norm} and H\"older's inequality,
\begin{equation}
\begin{aligned}
\norm{O}_\infty
&=
\norm{
(\mathsf C^+)^{\mathsf T}
\mathsf A
\mathsf C^+
}_1
\\
&\leq
\norm{\mathsf C^+}_2^2
\norm{\mathsf A}_\infty
\\
&=
\Tr[
\mathsf C^+
(\mathsf C^+)^{\mathsf T}
]
\\
&=
\Tr[
(\mathsf C^{\mathsf T}\mathsf C)^{-1}
]
\\
&=
\Tr[
\left(
\bI_{2n}-\mathsf M^{\mathsf T}\mathsf M
\right)^{-1}
].
\end{aligned}
\end{equation}
Since the construction works for every real antisymmetric
$\mathsf A$ with $\norm{\mathsf A}_\infty=1$, the definition in Eq.~\eqref{eq:def_quadratic_reconstruction_constant} gives
\begin{equation}
\mu_2(\cN_{\mathrm G})
\leq
\Tr[
\left(
\bI_{2n}-\mathsf M^{\mathsf T}\mathsf M
\right)^{-1}
].
\end{equation}

Finally, every eigenvalue of
$\mathsf M^{\mathsf T}\mathsf M$ is upper bounded by
$\norm{\mathsf M}_\infty^2$. Hence
\begin{equation}
\left(
\bI_{2n}-\mathsf M^{\mathsf T}\mathsf M
\right)^{-1}
\leq
\frac{1}{1-\norm{\mathsf M}_\infty^2}\bI_{2n},
\end{equation}
which yields
\begin{equation}
\Tr[
\left(
\bI_{2n}-\mathsf M^{\mathsf T}\mathsf M
\right)^{-1}
]
\leq
\frac{2n}{1-\norm{\mathsf M}_\infty^2}.
\end{equation}
This proves Eq.~\eqref{eq:mu_2_Gaussian_noise_bound}.
\end{proof}

\subsection{Gaussian codewords with Majorana distance two}\label{app:proof_distance_two_Gaussian_code}
The $4$-Majorana tetron code constructed in Ref.~\cite{litinski2018QuantumComputingMajorana} is an example of an $n=2$ mode fermionic code with Majorana distance $2$, whose codewords are all pure Gaussian states. Here we generalize this construction to a family of $[\![n,1,2]\!]_\mathrm{F}$ 
codes for which every codeword
is Gaussian for any $n\geq2$.
We define the two-dimensional code space
\begin{equation}
\mathcal C_n
\coloneqq
\operatorname{span}
\left\{
\ket{0_L},
\ket{1_L}
\right\},
\end{equation}
where
$\ket{0_L}
\coloneqq
\ket{0^n}$,
$\ket{1_L}
\coloneqq
\ket{1,1,0^{n-2}}$.
Its orthogonal projector is
$P_n
=
\ket{0_L}\bra{0_L}
+
\ket{1_L}\bra{1_L}$.
Both basis vectors have even Hamming weight and hence even fermion
parity. Therefore,
$\Pi P_n
=
P_n\Pi
=
P_n$.
On the one hand, $P_n$ can detect every Majorana operator with weight one. Since $\forall j\in[2n]$,
\begin{equation}
P_n\gamma_jP_n=P_n\Pi\gamma_j\Pi P_n=-P_n\gamma_jP_n,
\end{equation}
so $P_n\gamma_jP_n=0$ and $d_{\mathrm F}(P_n)\geq 2$.
On the other hand, $P_n$ cannot detect the weight-two error
\begin{equation}
\widehat{\gamma}_{\{1,2\}}
=
-i\gamma_1\gamma_2
=
Z_1.
\end{equation}
It acts non-trivially within the code space,
\begin{equation}
Z_1\ket{0_L}
=
\ket{0_L},
\qquad
Z_1\ket{1_L}
=
-\ket{1_L},
\end{equation}
so that
\begin{equation}
P_nZ_1P_n
=
\ket{0_L}\bra{0_L}
-
\ket{1_L}\bra{1_L}
\notin
\mathbb C P_n,
\end{equation}
which implies $d_{\mathrm F}(P_n)\leq 2$.
So we conclude
$d_{\mathrm F}(P_n)=2$.
It remains to prove that every normalized vector in $\mathcal C_n$
is a pure fermionic Gaussian state. Up to a global phase, an arbitrary
normalized codeword has the form
\begin{equation}
\ket{\psi_{\theta,\phi}}
=
\left(
\cos\theta\ket{0,0}
+
e^{i\phi}\sin\theta\ket{1,1}
\right)_{1,2}
\otimes
\ket{0^{n-2}}_{3,\cdots,n},
\label{eq:distance_two_general_codeword}
\end{equation}
where
$0\leq\theta\leq\frac{\pi}{2}$ and
$0\leq\phi<2\pi$. We evaluate the covariance matrix of $\psi_{\theta,\phi}=\ketbra{\psi_{\theta,\phi}}{\psi_{\theta,\phi}}$. Notice that the covariance matrix of $\ketbra{0}{0}$ is 
\begin{equation}
J=\begin{bmatrix}
0&1\\
-1&0
\end{bmatrix}.
\end{equation}
The tensor product structure in Eq.~\eqref{eq:distance_two_general_codeword} gives
\begin{equation}\label{eq:Gamma_psi_theta_phi}
\Gamma(\psi_{\theta,\phi})
=
\Gamma\left(\ketbra{\psi_{\theta,\phi}^{(1,2)}}{\psi_{\theta,\phi}^{(1,2)}}\right)
\oplus
J^{\oplus(n-2)},
\end{equation}
where
\begin{equation}
\ketbra{\psi_{\theta,\phi}^{(1,2)}}{\psi_{\theta,\phi}^{(1,2)}}=\Tr_{3,\cdots,n}\left(\ketbra{\psi_{\theta,\phi}}{\psi_{\theta,\phi}}\right),
\end{equation}
and
\begin{equation}
\begin{aligned}
\Gamma\left(\ketbra{\psi_{\theta,\phi}^{(1,2)}}{\psi_{\theta,\phi}^{(1,2)}}\right)_{k,l}=\left(
\cos\theta\bra{0,0}
+
e^{-i\phi}\sin\theta\bra{1,1}
\right)(-i\gamma_k\gamma_l)\left(
\cos\theta\ket{0,0}
+
e^{i\phi}\sin\theta\ket{1,1}
\right),&\qquad&
1\leq k<l\leq4,\\
\Gamma\left(\ketbra{\psi_{\theta,\phi}^{(1,2)}}{\psi_{\theta,\phi}^{(1,2)}}\right)_{k,k}=0,&&
1\leq k\leq4.
\end{aligned}
\end{equation}
Notice that the Majorana operators and products under Jordan--Wigner transformation are
\begin{equation}
\gamma_1=X_1,
\qquad
\gamma_2=Y_1,
\qquad
\gamma_3=Z_1X_2,
\qquad
\gamma_4=Z_1Y_2,
\end{equation}
and
\begin{equation}
\begin{aligned}
-i\gamma_1\gamma_2&=Z_1,
&\qquad
-i\gamma_2\gamma_3&=X_1X_2,
\\
-i\gamma_1\gamma_3&=-Y_1X_2,
&-i\gamma_2\gamma_4&=X_1Y_2
,
\\
-i\gamma_1\gamma_4&=-Y_1Y_2,
&
-i\gamma_3\gamma_4&=Z_2.
\label{eq:first_two_mode_Majorana_bilinears}
\end{aligned}
\end{equation}
A direct evaluation gives
\begin{equation}
\Gamma\left(\ketbra{\psi_{\theta,\phi}^{(1,2)}}{\psi_{\theta,\phi}^{(1,2)}}\right)=
\begin{bmatrix}
0 & \cos(2\theta) & -\sin(2\theta)\sin\phi & \sin(2\theta)\cos\phi\\
-\cos(2\theta) & 0 & \sin(2\theta)\cos\phi & \sin(2\theta)\sin\phi\\
\sin(2\theta)\sin\phi & -\sin(2\theta)\cos\phi & 0 & \cos(2\theta)\\
-\sin(2\theta)\cos\phi & -\sin(2\theta)\sin\phi & -\cos(2\theta) & 0
\end{bmatrix},
\end{equation}
and
\begin{equation}
\Gamma\left(\ketbra{\psi_{\theta,\phi}^{(1,2)}}{\psi_{\theta,\phi}^{(1,2)}}\right)^2=-\bI_4.
\end{equation}
Since $J^2=-\bI_2$, from Eq.~\eqref{eq:Gamma_psi_theta_phi}, we have
\begin{equation}
\Gamma(\psi_{\theta,\phi})^2=(-\bI_4)\oplus(-\bI_{2n-4})=-\bI_{2n},
\end{equation}
which implies $\psi_{\theta,\phi}$ is a pure Gaussian state~\cite{bittel2025OptimalTraceDistanceBoundsa}.

\subsection{Proof of Proposition~\ref{prop:Gaussian_encoding_measurement_informal}}\label{app:proof_Gaussian_encoding_measurement_informal}

In this appendix, we prove Proposition~\ref{prop:Gaussian_encoding_measurement_informal}. We first present the detailed settings and several useful lemmas, leaving their proofs to the end of this appendix.
In what follows,
we consider an $[\![n,k,\dF]\!]_\mathrm{F}$ fermionic code with code projector $P$. Denote the logical Hilbert space, i.e., the Fock space of $k$ fermionic modes, as $\cH_k$, and the physical Hilbert space, i.e., the Fock space of $n$ fermionic modes, as $\cH_n$. An encoding isometry of $P$ is a linear map
$V:\cH_k\rightarrow \cH_n,$ such that $V^\dagger V=\bI_{\cH_k}$ and $VV^\dagger=P$. The code space $\cC=\operatorname{im}(P)=V\cH_k\subseteq\cH_n$. The exact universal encoding channel is defined as follows:
\begin{definition}[Exact universal encoding map]\label{def:exact_universal_encoding_map}
A completely positive trace-preserving map $\cE$ is called an exact and universal deterministic encoding channel associated with the encoding isometry $V$ if
\begin{equation}
\cE(O)=VOV^\dagger,\qquad\text{ for every }O\in\mathcal{L}(\cH_k),
\end{equation}
where $\mathcal{L}(\cH_k)$ is the operator space of $\cH_k$. Furthermore, if the encoder is probabilistic, we identify a set of accepted outcome labels $\mathcal{X}_{\mathrm{acc}}$, which is a subset of all outcome labels $\mathcal{X}$. Consider a family of completely positive trace-non-increasing maps $\{\cE_x\}_{x\in\mathcal{X}}$ from $\cL(\cH_k)$ to $\cL(\cH_n)$ such that $\sum_{x\in\mathcal{X}}\cE_x$ is trace-preserving, then the accepted fine-grained branches $\{\cE_x\}_{x\in\mathcal{X}_{\mathrm{acc}}}$ are called exact universal probabilistic encoders if there exists a state-independent probability $p>0$, such that
\begin{equation}\label{eq:def_exact_universal_encoding}
\sum_{x\in\mathcal{X}_{\mathrm{acc}}}\cE_x(O)=pVOV^\dagger,\qquad\text{ for every }O\in\mathcal{L}(\cH_k).
\end{equation}
\end{definition}
From the above definition, universal means that the same channel $\cE$ must encode every logical state without knowledge of the input. It therefore excludes state-by-state preparation protocols whose control parameters depend on the target logical state. We then specify the definition of an exact non-destructive measurement of a code space projector.
\begin{definition}[Exact non-destructive measurement of a code space projector]\label{def:exact_non-destructive_measurement}
We identify a set of accepted outcome labels $\mathcal{X}_{\mathrm{acc}}$, which is a subset of all outcome labels $\mathcal{X}$. Consider a family of completely positive trace-non-increasing maps $\{\cM_x\}_{x\in\mathcal{X}}$ from $\cL(\cH_n)$ to $\cL(\cH_n)$ such that $\sum_{x\in\mathcal{X}}\cM_x$ is trace-preserving, then the accepted fine-grained branches $\{\cM_x\}_{x\in\mathcal{X}_{\mathrm{acc}}}$ are called an exact non-destructive measurement (on the accepted code space) of a code space projector $P$ if for every 
$O\in P\cL(\cH_n) P$,
\begin{equation}\label{eq:def_exact_non-destructive_measurement}
\sum_{x\in\mathcal{X}_{\mathrm{acc}}}\cM_x^\textnormal{\textdagger}(\bI_{\cH_n})=P\text{ and }\sum_{x\in\mathcal{X}_{\mathrm{acc}}}\cM_x(O)=O.
\end{equation}
\end{definition}
The non-destructive requirement is stronger than merely producing a
correct classical membership flag. A fine-grained measurement may
distinguish states inside and outside $\operatorname{im}(P)$ while at the same
time resolving additional observables within the code space. Such a
measurement can destroy off-diagonal logical coherences and therefore
need not satisfy the second equation of
Eq.~\eqref{eq:def_exact_non-destructive_measurement}. It constitutes
destructive code space detection, not a non-destructive measurement of
the code space projector.

In the following lemma, we show that Gaussianity of the maximally mixed code state forces a non-trivial fermionic code to have raw Majorana distance one.
\begin{lemma}[Gaussian normalized code projector implies distance one]\label{lem:Gaussian_projector_distance_one}
Let $P$ be the projector onto a fermionic quantum code
$\mathcal{C}=\operatorname{im}(P)$ of $n$ modes with $\operatorname{rank}(P)>1$. If $P$ is a positive Gaussian operator, equivalently, the maximally mixed state on $\mathcal{C}$, i.e., $P/\operatorname{rank}(P)$, is Gaussian, then $d_{\mathrm{F}}(P)=1$.
\end{lemma}
We also show that a completely positive map with a rank-one Choi operator admits no non-trivial completely positive refinement, so that every fine-grained branch preserves exactly the same logical transformation and can differ only in its state-independent occurrence probability, which applies to the exact universal probabilistic encoders defined above.
\begin{lemma}[Rank-one Choi refinement]\label{lem:rank-one_Choi_refinement}
Let $\sum_x\Lambda_x$ be a countable sum of completely positive maps on a finite-dimensional input space. Suppose for any input operator $O$,
\begin{equation}
\sum_x\Lambda_x(O)=pVOV^\textnormal\textdagger
\end{equation}
for some $p>0$ and isometry $V$, then there exist coefficients $p_x\geq0$ with $\sum_xp_x=p$, such that
\begin{equation}
\Lambda_x(O)=p_xVOV^\textnormal\textdagger
\end{equation}
for every $x$ and every $O$.
\end{lemma}
Another useful lemma states that a Gaussian measurement branch can only have a positive Gaussian operator as its associated effect, i.e., the positive operator that determines the probability of obtaining a particular measurement outcome, which constrains the support geometry of any exactly accepted code space.
\begin{lemma}[Gaussian branch effects]\label{lem:Gaussian_branch_effects}
Let $\Phi:\cL(\cH_1)\rightarrow\cL(\cH_2)$ be a fermionic Gaussian completely positive and trace-non-increasing map, where $\cH_1$ is the Fock space of $n_1$ input fermionic modes and
$\cH_2$ is the Fock space of $n_2$ output fermionic modes. The Hilbert-Schmidt adjoint $\Phi^\textnormal\textdagger:\cL(\cH_2)\rightarrow\cL(\cH_1)$ is defined by
\begin{equation}
\Tr(A^\textnormal\textdagger\Phi(B))=\Tr(\Phi^\textnormal\textdagger(A)^\textnormal\textdagger B),
\end{equation}
for any $A\in\cL(\cH_2)$ and $B\in\cL(\cH_1)$. Then the effect $E_\Phi=\Phi^\textnormal\textdagger(\bI_{\cH_2})$ of $\Phi$ is a positive fermionic Gaussian operator.
\end{lemma}

We now present the proof of Proposition~\ref{prop:Gaussian_encoding_measurement_informal}.
\begin{proof}[Proof of Proposition~\ref{prop:Gaussian_encoding_measurement_informal}]
For any non-trivial $[\![n,k,\dF]\!]_\mathrm{F}$ fermionic code with projector $P$, $k
\geq 1$, and $\dF\geq2$, denote the logical space as $\cH_k$ and the physical space as $\cH_n$. We prove the statements for encoding and non-destructive measurement separately.
\begin{itemize}
\item \textbf{Encoding:} Assume there exists an exact universal encoding map in Definition~\ref{def:exact_universal_encoding_map} consisting of Gaussian branches $\left\{\cE_x\right\}_{x\in\mathcal{X}_{\mathrm{acc}}}$ and $p>0$, such that
\begin{equation}
\sum_{x\in\mathcal{X}_{\mathrm{acc}}}\cE_x(O)=pVOV^\dagger,\qquad\text{ for every }O\in\mathcal{L}(\cH_k).
\end{equation}
Then by Lemma~\ref{lem:rank-one_Choi_refinement}, there exist $p_x\geq0$, such that $\sum_{x\in\mathcal{X}_{\mathrm{acc}}}p_x=p$ and
\begin{equation}
\cE_x(O)=p_xVOV^\dagger,\qquad\text{ for every }O\in\mathcal{L}(\cH_k).
\end{equation}
Select index $x'\in\mathcal{X}_{\mathrm{acc}}$ such that $p_{x'}>0$ and consider the action of $\cE_{x'}$ on the maximally mixed state on $\cH_k$, which gives
\begin{equation}
\cE_{x'}\left(\frac{\bI_{\cH_k}}{2^k}\right)=\frac{p_{x'}}{2^k}V\bI_{\cH_k}V^\dagger=p_{x'}\frac{P}{2^k}.
\end{equation}
Since $\cE_{x'}$ is Gaussian, the maximally mixed state on the code space $P/2^k$ is a Gaussian state, and by Lemma~\ref{lem:Gaussian_projector_distance_one} we obtain $\dF=1$, contradicting
$\dF\geq2$.
\item \textbf{Non-destructive measurement:} Assume there exists an exact non-destructive measurement of $P$ in Definition~\ref{def:exact_non-destructive_measurement} consisting of Gaussian completely positive and trace-non-increasing maps $\left\{\cM_x\right\}_{x\in\mathcal{X}_{\mathrm{acc}}}$, such that for every $O\in P\cL(\cH_n)P$,
\begin{equation}
\sum_{x\in\mathcal{X}_{\mathrm{acc}}}\cM_x^\textnormal{\textdagger}(\bI_{\cH_n})=P\text{ and }\sum_{x\in\mathcal{X}_{\mathrm{acc}}}\cM_x(O)=O.
\end{equation}
The second equation implies the summation is the identity channel restricted to $\cC=\operatorname{im}(P)$, i.e., 
\begin{equation}
\left.\left(\sum_{x\in\mathcal{X}_{\mathrm{acc}}}\cM_x\right)\right|_{\cL(\cC)}=\bI_{\cL(\cC)}.
\end{equation}
Then by Lemma~\ref{lem:rank-one_Choi_refinement}, there exist $\{p_x\}\geq0$, such that $\sum_{x\in\mathcal{X}_{\mathrm{acc}}}p_x=1$ and
\begin{equation}
\cM_x(O)=p_xO,\qquad\text{ for every }O\in\cL(\cC).
\end{equation}
For each $x\in\mathcal{X}_{\mathrm{acc}}$, let $E_x=\cM_x^\dagger(\bI_{\cH_n})$,
then it is a positive Gaussian operator by Lemma~\ref{lem:Gaussian_branch_effects}. Also, for every $O\in P\cL(\cH_n)P$, we have $\Tr(E_xO)=\Tr(\cM_x(O))=p_x\Tr(O)$. Consider $O=PE_xP-p_xP$, then $PO^\dagger P=O^\dagger$, and
\begin{equation}
\Tr(OO^\dagger)=\Tr(PE_xPO^\dagger-p_xPO^\dagger)=\Tr(E_xO^\dagger)-p_x\Tr(O^\dagger)=0,
\end{equation}
which implies $PE_xP=p_xP$. Notice that $\sum_{x\in\mathcal{X}_{\mathrm{acc}}}E_x=\sum_{x\in\mathcal{X}_{\mathrm{acc}}}\cM_x^\textnormal{\textdagger}(\bI_{\cH_n})=P$, so for each $x$, $0\leq E_x\leq P$. Then for any $\ket{\psi}\in\ker(P)$, $0\leq\bra{\psi}E_x\ket{\psi}\leq\bra{\psi}P\ket{\psi}=0$, so the support of $E_x$ is contained in $\cC$ and thus $E_x=PE_xP=p_xP$. Finally, at least one $p_x>0$ implies $P$ is a positive Gaussian operator and thus $\dF=1$ by Lemma~\ref{lem:Gaussian_projector_distance_one}, contradicting $\dF\geq2$.
\end{itemize}
\end{proof}

The following are the proofs of Lemmas~\ref{lem:Gaussian_projector_distance_one},~\ref{lem:rank-one_Choi_refinement}, and~\ref{lem:Gaussian_branch_effects}.

\begin{proof}[Proof of Lemma~\ref{lem:Gaussian_projector_distance_one}]
For a code projector $P$ of $n$ fermionic modes with $K=\operatorname{rank}(P)>1$, if the maximally mixed state on $\operatorname{im}(P)$ is Gaussian, we write it into the form in Eq.~\eqref{eq:def_Gaussian_state} as
\begin{equation}\label{eq:projector_Gaussian_form}
\frac{P}{K}
=
U_R\bigotimes_{j=1}^n\left(\frac{\bI+\nu_jZ_j}2\right)U_R^\dagger,
\end{equation}
for some Gaussian unitary $U_R$ with $R\in\operatorname{O}(2n)$ and each $\nu_j\in[-1,1]$.
On one hand, $P$ is an orthogonal projector, so $P/K$ has eigenvalues $1/K$ with multiplicity $K$ and $0$ of multiplicity $2^n-K$. On the other hand, Eq.~\eqref{eq:projector_Gaussian_form} implies the eigenvalues are
\begin{equation}
p_{\vec{s}}=\prod_{j=1}^n\frac{1+s_j\nu_j}2,\qquad s_j\in\{\pm1\},\qquad \vec{s}=(s_1,\cdots,s_n).
\end{equation}
Suppose that $0<\abs{\nu_j}<1$ for some $j$. Fix all $s_k$ with $k\neq j$ such that $s_k\nu_k\neq-1$, the two possible occupations of mode $j$
then give two nonzero eigenvalues proportional to
$(1+\nu_j)/2$ and $(1-\nu_j)/2$, which are unequal. This
contradicts the fact that all nonzero eigenvalues of $P/K$ should be $1/K$.
Therefore
\begin{equation}
\nu_j\in\{-1,0,+1\}
\qquad
\forall j\in[n].
\end{equation}
Let $r$ be the number of indices for which $\nu_j=0$. A factor
with $\abs{\nu_j}=1$ has rank one, whereas a factor with $\nu_j=0$
equals $\bI/2$ and has rank two. Hence 
$K=2^r$ with $r\geq1$ since $K>1$.
After relabeling the canonical
modes, we may assume that
\begin{equation}
\nu_j
\begin{cases}
\in\{\pm1\},\qquad &j=1,\cdots,n-r,\\
=0,&j=n-r+1,\cdots,n.
\end{cases}
\end{equation}
In this way, $P/K$ and $P$ take the form
\begin{equation}
\frac PK=U_R
\left(
\ketbra{t}{t}
\otimes \frac{\bI_{2^r}}{2^r}
\right)
U_R^{\dagger},\qquad
P=U_RP_0U_R^\dagger,
\qquad
P_0=
\ketbra{t}{t}
\otimes \bI_{2^r},
\end{equation}
where $\ket{t}$ fixes the occupations of the
$n-r$ nonzero canonical modes.
Choose the first unconstrained canonical mode
$\ell=n-r+1$ and the Majorana operator
$\gamma_{2\ell-1}
=
\left(\prod_{j<\ell}Z_j\right)X_\ell$.
Since all modes preceding $\ell$ are fixed by $\ket{t}$, there is a
sign $\sigma_t\in\{\pm1\}$ such that
$\left(\prod_{j<\ell}Z_j\right)\ket{t}
=
\sigma_t\ket{t}$.
Hence, for any normalized state $\ket{\chi}$ on the remaining
$r-1$ unconstrained modes,
\begin{equation}
\begin{aligned}
\gamma_{2\ell-1}
\left(
\ket{t}\otimes\ket{0_\ell}\otimes\ket{\chi}
\right)
&=
\sigma_t
\ket{t}\otimes\ket{1_\ell}\otimes\ket{\chi},
\\
\gamma_{2\ell-1}
\left(
\ket{t}\otimes\ket{1_\ell}\otimes\ket{\chi}
\right)
&=
\sigma_t
\ket{t}\otimes\ket{0_\ell}\otimes\ket{\chi}.
\end{aligned}
\end{equation}
Consequently,
\begin{equation}
P_0\gamma_{2\ell-1}P_0
\notin
\mathbb C P_0,
\end{equation}
since it exchanges the above two orthogonal vectors in
$\operatorname{im}(P_0)$ and therefore cannot act as a scalar on that
subspace.
From Eq.~\eqref{eq:Gaussian_unitary_action}, define 
\begin{equation}
\widetilde{\gamma}_{2\ell-1}=U_R\gamma_{2\ell-1}U_R^\dagger=\sum_{k=1}^{2n}R_{k,2\ell-1}\gamma_k,
\end{equation}
we have
\begin{equation}
P\widetilde{\gamma}_{2\ell-1}P=\left(U_RP_0U_R^\dagger\right)\left(U_R\gamma_{2\ell-1}U_R^\dagger\right)\left(U_RP_0U_R^\dagger\right)=U_RP_0\gamma_{2\ell-1}P_0U_R^\dagger\qquad\notin \qquad U_R(\mathbb{C}P_0)U_R^\dagger=\mathbb{C}P,
\end{equation}
which implies there exists $k\in[2n]$, such that $P\gamma_kP\notin\mathbb{C}P$ and thus $d_{\mathrm{F}}(P)=1$.
\end{proof}

\begin{proof}[Proof of Lemma~\ref{lem:rank-one_Choi_refinement}]
Notice that the ordinary, unnormalized Choi operator of $\Lambda=\sum_x\Lambda_x$ is a positive rank-one operator $J(\Lambda)=p\dketbra{V}{V}$, which equals $\sum_xJ(\Lambda_x)$. Here $\dket{V}=\sum_aV\ket{a}\otimes\ket{a}$. Since each $J(\Lambda_x)$ is positive, and its support is contained in the one-dimensional support of $J(\Lambda)$, there exists $p_x\geq0$ such that 
$J(\Lambda_x)=p_x\dketbra{V}{V}$ with $\sum_xp_x=p$. Then by Choi duality, we have for each $x$ and any input $O$, $\Lambda_x(O)=p_xVOV^\dagger$.
\end{proof}

\begin{proof}[Proof of Lemma~\ref{lem:Gaussian_branch_effects}]
We first show $E_\Phi$ is a positive operator, for any $\ket{\psi_1}\in\cH_1$, since $\Phi$ is completely positive, we have
\begin{equation}
\Tr(E_\Phi\ketbra{\psi_1}{\psi_1})=\Tr(\Phi^\textnormal\textdagger(\bI_{\cH_2})\ketbra{\psi_1}{\psi_1})=\Tr(\Phi(\ketbra{\psi_1}{\psi_1}))\geq0.
\end{equation}
Noticing that $\Phi$ is parity-preserving, so for any $O\in\cL(\cH_1)$, 
\begin{equation}
\begin{aligned}
\Tr(\Pi_1E_\Phi\Pi_1O)&=\Tr(E_\Phi\Pi_1O\Pi_1)
\\&=\Tr(\Phi(\Pi_1O\Pi_1))
\\&=\Tr(\Pi_2\Phi(O)\Pi_2)
\\&=\Tr(\Phi(O))
\\&=\Tr(E_\Phi O),
\end{aligned}
\end{equation}
which implies $\Pi_1E_\Phi\Pi_1=E_\Phi$ and $E_\Phi$ is an even operator.
Consider the fermionic Choi operator of $\Phi$ defined in Eq.~\eqref{eq:fermionic_Choi_basis_expansion},
\begin{equation}
J_\mathrm{F}(\Phi)
=
2^{-{n_1}}
\sum_{\vec\mu,\vec\nu\in\{0,1\}^{n_1}}
\Phi\left(
\ketbra{\vec\mu}{\vec\nu}_1\Pi_1^{\varepsilon(\vec\mu,\vec\nu)}
\right)
\Pi_2^{\varepsilon(\vec\mu,\vec\nu)}
\otimes
\ketbra{\vec\mu}{\vec\nu}_{1'},
\end{equation}
where $\varepsilon(\vec\mu,\vec\nu)\coloneqq(\abs{\vec\mu}+\abs{\vec\nu})\mod 2$. $J_\mathrm{F}(\Phi)$ is either zero or a positive Gaussian operator since $\Phi$ is a fermionic Gaussian completely positive map.
Taking the partial trace over system $\cH_2$ leads to
\begin{equation}
\begin{aligned}
\Tr_2(J_\mathrm{F}(\Phi))=&2^{-n_1}\sum_{\vec\mu,\vec\nu\in\{0,1\}^{n_1}}\Tr(\Phi\left(
\ketbra{\vec\mu}{\vec\nu}_1\Pi_1^{\varepsilon(\vec\mu,\vec\nu)}
\right)
\Pi_2^{\varepsilon(\vec\mu,\vec\nu)})\ketbra{\vec\mu}{\vec\nu}_{1'}
\\=&2^{-n_1}\sum_{\substack{\vec\mu,\vec\nu\in\{0,1\}^{n_1}\\\varepsilon(\vec\mu,\vec\nu)=1}}\Tr(\Phi\left(
\ketbra{\vec\mu}{\vec\nu}_1\Pi_1
\right)
\Pi_2)\ketbra{\vec\mu}{\vec\nu}_{1'}
\\&+2^{-n_1}\sum_{\substack{\vec\mu,\vec\nu\in\{0,1\}^{n_1}\\\varepsilon(\vec\mu,\vec\nu)=0}}\Tr(\Phi\left(
\ketbra{\vec\mu}{\vec\nu}_1
\right)
)\ketbra{\vec\mu}{\vec\nu}_{1'}
\\=&2^{-n_1}\sum_{\vec\mu,\vec\nu\in\{0,1\}^{n_1}}{}_1\bra{\vec\nu}E_\Phi\ket{\vec\mu}_1\ketbra{\vec\mu}{\vec\nu}_{1'}
\\=&2^{-n_1}\left(E_\Phi^{\mathsf{T}}\right)_{1'}.
\end{aligned}
\end{equation}
In the second equality, we split the summation into two parts according to the value of $\varepsilon(\vec\mu,\vec\nu)$. The odd term vanishes since $p(O_1O_2)\equiv p(O_1)+p(O_2)\pmod 2$ for parity-homogeneous $O_1$ and $O_2$, $\Phi$ is parity-preserving, and odd operators are traceless. The summation indices of the even term can be relaxed since $E_\Phi$ is an even operator and thus $\bra{\vec\nu}E_\Phi\ket{\vec\mu}=0$ for $\varepsilon(\vec\mu,\vec\nu)=1$.
By Wick's theorem, taking a partial trace of a positive Gaussian operator yields a positive Gaussian operator, so $\left(E_\Phi^{\mathsf{T}}\right)_{1'}$ is Gaussian. Since we use the same fermionic CAR algebra for both $\cH_1$ and $\cH_{1'}$, $E_\Phi^{\mathsf{T}}$ is a positive Gaussian operator. Noticing that the transposition under the occupation basis acting on Majorana operators gives
\begin{equation}
\gamma_{2j-1}^{\mathsf{T}}=\gamma_{2j-1},\qquad\gamma_{2j}^{\mathsf{T}}=-\gamma_{2j},
\end{equation}
so it preserves Gaussianity and $E_\Phi$ is Gaussian.
\end{proof}

\section{Non-Gaussianity cost of fermionic encoding}\label{app:proof_universal_non-Gaussian_cost}
In this appendix, we introduce notations and lemmas in Appendix~\ref{app:notations_lemmas} and then prove Theorem~\ref{thm:universal_non_Gaussian_cost} in the main text in Appendix~\ref{app:g-doped_Gaussian_fidelity}.

\subsection{Notations and lemmas}\label{app:notations_lemmas}

We generalize the reference family in Eq.~\eqref{eq:def_Gaussian_fidelity} from pure Gaussian states to the family of non-Gaussian states.
We consider a family of $w$-local gates, each of which is generated by a Hermitian Majorana product with an even Majorana weight $w>2$:
\begin{equation}
V_K(\theta)=\exp(i\theta\widehat{\gamma}_K), \qquad K\subseteq[2n]\text{ with } \abs{K}=w.
\end{equation}
In particular, quartic gates, i.e., the $w=4$ family, can be regarded as typical non-Gaussian elementary gates since $V_K(\theta)$ is a parity-preserving non-Gaussian unitary whenever $\sin(2\theta)\neq0$. Quartic Majorana products are the lowest-degree even generators beyond the quadratic Gaussian sector. For example, the nearest-neighbor SWAP gate is a ``magic'' gate in fermionic Gaussian circuits~\cite{jozsa2008MatchgatesClassicalSimulation,brod2011ExtendingMatchgatesUniversal,mele2025EfficientLearningQuantum} and $\operatorname{SWAP}_{j,j+1}$ can be composed by Gaussian unitaries and the quartic gate $V_{\{2j-1,2j,2j+1,2j+2\}}(\pi/4)$. 
Recall that an $n$-mode
\emph{$(g,q)$-doped Gaussian circuit} for $n,g\in\mathbb{N}$ and $q\geq4$ takes the form
\begin{equation}\label{eq:g_doped_circuit_app}
U
=
G_gV_gG_{g-1}\cdots V_2G_1V_1G_0,
\end{equation}
where each $G_j$ is a fermionic Gaussian unitary, and each $V_j$ is a non-Gaussian unitary generated by a Hermitian Majorana product of even weight $w_j\leq q$.
We call an $n$-mode state vector $\ket{\phi}$ a \emph{$(g,q)$-doped fermionic Gaussian state} (the set is denoted as $\mathcal G_{g,q}^{(n)}$), if it can be generated by a $(g,q)$-doped Gaussian circuit from state $\ket{0^n}$. Theorem~1 in Ref.~\cite{mele2025EfficientLearningQuantum} shows that the non-Gaussian content of a $(g,q)$-doped Gaussian state can be compressed to a local subsystem of size $gq$, reformulated as the following lemma.
\begin{lemma}[Fermionic non-Gaussianity compression,~\cite{mele2025EfficientLearningQuantum}]\label{lem:fermionic_non-Gaussianity_compression}
For any $n$-mode $(g,q)$-doped fermionic Gaussian state $\ket{\psi}$, there exist a Gaussian unitary $G$ and a state $\ket{\chi}$ on the first $c=\min\{n,gq\}$ modes, such that
\begin{equation}
\ket{\psi}=G\left(\ket{\chi}_C\otimes\ket{0^{n-c}}_T\right),
\end{equation}
where we call $C$ the core system of $c$ modes and $T$ the tail system of the remaining $n-c$ modes.
\end{lemma}

For a state vector $\ket{\psi}$, define its $(g,q)$-doped Gaussian fidelity as the maximum fidelity between $\ket{\psi}$ and any $(g,q)$-doped Gaussian state, i.e.,
\begin{equation}\label{eq:def_g-doped_Gaussian_fidelity}
F_{g,q}^{(n)}(\ket{\psi})
\coloneqq
\sup_{\ket{\phi}\in\mathcal G_{g,q}^{(n)}}
\abs{\braket{\phi}{\psi}}^2.
\end{equation}
When $g=0$, it equals the fermionic
Gaussian fidelity in Eq.~\eqref{eq:def_Gaussian_fidelity}. 

We then provide some useful lemmas and their proofs.
The following Lemma~\ref{lem:distance_degradation} provides a lower bound on the Majorana distance remaining after conjugation by a doped circuit, whose proof uses Lemma~\ref{lem:local-even-degree-growth}.

\begin{lemma}[Degree growth under a local unitary]
\label{lem:local-even-degree-growth}
For $r\in\mathbb{N}$, denote the subspace that is spanned by Majorana products with weight at most $r$, i.e., every operator in this subspace has Majorana degree at most $r$, as
\begin{equation}
\mathcal A_{\leq r}
\coloneqq
\operatorname{span}
\left\{
\widehat{\gamma}_S\middle|
\abs{S}\leq r
\right\}.
\label{eq:Majorana-degree-filtration}
\end{equation}
Let $V=\exp(i\theta\widehat{\gamma}_K)$ be a parity-preserving unitary generated by a Hermitian Majorana product with even $\abs{K}\geq2$ and $\theta\in[0,2\pi)$. Then
\begin{equation}
V\mathcal A_{\leq r}V^{\dagger}
\subseteq
\mathcal A_{\leq r+\abs{K}-2}.
\label{eq:local-even-degree-growth}
\end{equation}
\end{lemma}

\begin{proof}[Proof of Lemma~\ref{lem:local-even-degree-growth}]
It suffices to consider a Hermitian Majorana product
$\widehat{\gamma}_S$ of degree $\abs{S}\leq r$. 
Up to an irrelevant phase,
$\widehat{\gamma}_S$ is the product of a monomial supported outside
$K$ and a monomial supported inside $K$, i.e., 
\begin{equation}
\widehat\gamma_S
\propto
\widehat\gamma_{S\setminus K}
\widehat\gamma_{S\cap K},
\end{equation}
so
\begin{equation}\label{eq:V_gamma_S_Vdagger}
V\widehat\gamma_S V^\dagger\propto\widehat\gamma_{S\setminus K}V\widehat\gamma_{S\cap K}V^\dagger
\end{equation}

If $\abs{S\cap K}=0$, then $V$ commutes with every operator on
the disjoint support $S$ since $\abs{K}$ is even, so $V\widehat\gamma_SV^\dagger=\widehat\gamma_S$ is unchanged. Suppose,
therefore, that $\abs{S\cap K}>0$. Using Eq.~\eqref{eq:Hermitian_Majorana_product_product} and $\abs{K}$ is even, 
$V\widehat{\gamma}_{S\cap K}V^{\dagger}$
is a linear combination of monomials
$\widehat{\gamma}_T$ with $T\subseteq K$ and
\begin{equation}
\abs{T}\equiv \abs{S\cap K}\pmod 2.
\end{equation}
\begin{itemize}
\item If $\abs{S\cap K}$ is odd, then $\abs{S\cap K}\geq1$ and
$\abs{T}\leq \abs{K}-1$,
the total degree in Eq.~\eqref{eq:V_gamma_S_Vdagger} is at most
\begin{equation}
\abs{S}-\abs{S\cap K}+\abs{T}
\leq\abs{S}-1+\abs{K}-1=\abs{S}+\abs{K}-2.
\end{equation}

\item If $\abs{S\cap K}$ is even, then $\abs{S\cap K}\geq2$ and
$\abs{T}\leq \abs{K}$,
the total degree in Eq.~\eqref{eq:V_gamma_S_Vdagger} is at most
\begin{equation}
\abs{S}-\abs{S\cap K}+\abs{T}
\leq\abs{S}-2+\abs{K}.
\end{equation}
\end{itemize}
Consequently, every monomial appearing after conjugation has degree
at most
\begin{equation}
\abs{S}+\abs{K}-2
\leq
r+\abs{K}-2.
\end{equation}
Using linearity, we prove
Eq.~\eqref{eq:local-even-degree-growth}.
\end{proof}

\begin{lemma}[Distance degradation under local doping]\label{lem:distance_degradation}
Let $P$ be the orthogonal code projector of an $[\![n,k,\dF]\!]_{\mathrm{F}}$ fermionic code and $U$ be a $(g,q)$-doped circuit in the form
\begin{equation}\label{eq:g_doped_circuit_app_2}
U
=
G_gV_gG_{g-1}\cdots V_2G_1V_1G_0,
\end{equation}
where each $G_j$ is a fermionic Gaussian unitary, and each $V_j$ is a non-Gaussian unitary generated by a Hermitian Majorana product of even weight $w_j\leq q$. Then $P'=U^\dagger PU$ is the orthogonal code projector of an $[\![n,k,\dF']\!]_\mathrm{F}$ fermionic code with
\begin{equation}
\dF'\geq\max\left\{1,\dF-\sum_{j=1}^g(w_j-2)\right\}\geq\max\left\{1,\dF-g(q-2)\right\}.
\end{equation}
In particular, when $g=0$, conjugation by a Gaussian unitary preserves Majorana distance since $\dF'=\dF$.
\end{lemma}

\begin{proof}[Proof of Lemma~\ref{lem:distance_degradation}]
For $P'=U^\dagger PU$, it is an orthogonal projector since $P'^\dagger=P'$ and $P'^2=U^\dagger P^2U=P'$. Its rank is $\operatorname{rank}(P')=\Tr(P')=\Tr(P)=2^k$. So $P'$ is a valid fermionic code projector encoding $k$ logical modes to $n$ physical modes. We have $\dF'\geq1$.
For any $S\subseteq[2n]$ with $\abs{S}<\dF-\sum_{j=1}^g(w_j-2)$, recursively using Lemma~\ref{lem:local-even-degree-growth} for $g$ times, we obtain
$U\gamma_SU^\dagger\in\cA_{\leq\dF-1}$
so that $\deg_\mathrm{F}\left(U\gamma_SU^\dagger\right)<\dF$. Expanding in the basis of Majorana products,
\begin{equation}
U\gamma_SU^\dagger=\sum_{\substack{T\subseteq[2n]\\\abs{T}\leq\dF-1}}c_{S,T}\gamma_T,
\end{equation}
and since the Majorana distance of $P$ is $\dF$, there exist coefficients $\{\lambda_T\}$ such that $P\gamma_TP=\lambda_TP$. Then
\begin{equation}
\begin{aligned}
P'\gamma_SP'&=U^\dagger PU\gamma_SU^\dagger PU
\\&=\sum_{\substack{T\subseteq[2n]\\\abs{T}\leq\dF-1}}c_{S,T}U^\dagger P\gamma_TPU
\\&=\left(\sum_{\substack{T\subseteq[2n]\\\abs{T}\leq\dF-1}}c_{S,T}\lambda_T\right)P',
\end{aligned}
\end{equation}
which implies $\dF'\geq\dF-\sum_{j=1}^g(w_j-2)\geq\dF-g(q-2)$ since each $w_j\leq q$.

If $g=0$ and $U$ is a Gaussian unitary, then $\dF'\geq\dF$ from above. Applying the same argument for $U^\dagger$, we obtain $\dF\geq\dF'$, so $\dF'=\dF$.
\end{proof}

For any $n\in\mathbb{N}_+$ and $a\in\mathbb{N}$, we denote the volume of Hamming ball in $\{0,1\}^n$ centered at $0^n$ with radius $a$ by
\begin{equation}
 V_n(a)=\sum_{w=0}^a\binom nw,
 \qquad
\text{and } V_n(-1)=0.
\label{eq:Hamming-volume}
\end{equation}
For $1\leq k\leq n$, denote the radius of the largest Hamming ball with volume less than $2^k$ by 
\begin{equation}
a_{n,k}
=\max\left\{a\in\{-1,0,\cdots,n-1\}\middle|V_n(a)<2^k\right\},
\label{eq:def_a_nk}
\end{equation}
and denote
\begin{equation}\label{eq:def_q_nk}
q_{n,k}=\frac{a_{n,k}+1}n.
\end{equation}
Recall the binary entropy function defined for $x\in[0,1]$,
\begin{equation}
H_2(x)=-x\log_2x-(1-x)\log_2(1-x),
\end{equation}
and we set $0\log 0=0$. We denote the inverse of the restriction of $H_2$ to $[0,1/2]$ as $H_2^{-1}$ and then prove the following inequality.

\begin{lemma}[Lower bound of $q_{n,k}$]\label{lem:lower_bound_q_nk}
For any $1\leq k\leq n$ and $q_{n,k}$ defined in Eq.~\eqref{eq:def_q_nk},
\begin{equation}\label{eq:lower_bound_q_nk}
q_{n,k}>H_2^{-1}\left(\frac kn\right).
\end{equation}
\end{lemma}

\begin{proof}[Proof of Lemma~\ref{lem:lower_bound_q_nk}]
We denote the right-hand side in Eq.~\eqref{eq:lower_bound_q_nk} as $x\in[0,1/2]$ and let $a=\lfloor nx\rfloor$. We need to show $a_{n,k}\geq a$ (cf. Eq.~\eqref{eq:def_a_nk}), which is equivalent to showing that $V_n(a)<2^k$ (cf.\ Eq.~\eqref{eq:Hamming-volume}).
\begin{itemize}
\item If $x=1/2$, then $k=n$ and $V_n(\lfloor n/2\rfloor)<2^n=2^k$. Below we assume $0<x<1/2$.
\item If $a=0$, then $V_n(a)=1<2^k$.
\item If $a\geq1$, let $p=a/n\in(0,1/2)$, notice that the function $p^w(1-p)^{n-w}$ is decreasing for $0\leq w\leq a$, we have
\begin{equation}\label{eq:proof_binomial_ball}
\begin{aligned}
1
&>
\sum_{w=0}^{a}
\binom nw
p^w(1-p)^{n-w}
\\
&\geq
V_n(a)
p^a(1-p)^{n-a}
\\
&=
V_n(a)
2^{-nH_2(a/n)}.
\end{aligned}
\end{equation}
Since $H_2(x)$ is increasing for $x\in[a/n,1/2]$,
\begin{equation}
V_n(a)<2^{nH_2(a/n)}\leq2^{nH_2(x)}=2^k.
\end{equation}
\end{itemize}
So in all cases we have $a_{n,k}\geq\lfloor nx\rfloor$, so
\begin{equation}
q_{n,k}\geq\frac{\lfloor nx\rfloor+1}n>x,
\end{equation}
which is Eq.~\eqref{eq:lower_bound_q_nk}.
\end{proof}

For any fermionic error-correcting code, assume we fix a partition as in Lemma~\ref{lem:fermionic_non-Gaussianity_compression}, we show that there exists a mode in the tail system, such that the probability of getting vacuum on that mode is the same for any codeword, and this quantity is upper bounded.

\begin{lemma}[Vacuum occupation bound]\label{lem:vacuum_occupation_bound}
For any $[\![n,k,\dF]\!]_\mathrm{F}$ fermionic code with orthogonal projector $Q$ with $\dF\geq3$, fix a core system $C$ consisting of $c<k$ modes and denote the remaining system of $n-c$ modes as tail system $T$. There exists a mode $\ell\in T$ and a scalar $p_\ell$ independent of the codeword $\ket{\zeta}\in\operatorname{im}(Q)$, such that
\begin{equation}\label{eq:vacuum_occupation_bound}
p_\ell=\bra{\zeta}\Pi_\ell\ket{\zeta}\leq1-q_{n-c,k-c},
\end{equation}
where $\Pi_\ell=\ket{0}_\ell\bra{0}_\ell$ with $\ket{0}_\ell$ denoting the vacuum state on mode $\ell$, and $q_{n-c,k-c}$ is defined in Eq.~\eqref{eq:def_q_nk}.
\end{lemma}

\begin{proof}[Proof of Lemma~\ref{lem:vacuum_occupation_bound}]
For each mode $j\in T$, denote the projector on the vacuum state on mode $j$ as $\Pi_j=\ket{0}_j\bra{0}_j$, 
and the number operator on mode $j$ as $N_j=\bI-\Pi_j$, equivalently,
\begin{equation}\label{eq:def_Pi_j_N_j}
\Pi_j
=
\frac{
\mathbb I-
i
\gamma_{2j-1}
\gamma_{2j}
}{2},
\qquad
N_j
=
\frac{
\mathbb I+
i
\gamma_{2j-1}
\gamma_{2j}
}{2}.
\end{equation}
Since $\deg_\mathrm{F}(N_j)=2$ and $\dF\geq3$, there exists $\lambda_j\in\mathbb{C}$ such that $QN_jQ=\lambda_jQ$. So for any normalized $\ket{\zeta}\in\operatorname{im}(Q)$,
\begin{equation}\label{eq:xi_N_j_xi}
\bra{\zeta}N_j\ket{\zeta}=\bra{\zeta}QN_jQ\ket{\zeta}=\lambda_j,
\end{equation}
and
\begin{equation}
\bra{\zeta}\Pi_j\ket{\zeta}=1-\lambda_j=p_j
\end{equation}
is independent of $\ket{\zeta}$.

Denote the number of modes in the tail system $T$ as $r=n-c$ and consider the subspace spanned by the states with Hamming weight upper bounded by $a_{r,k-c}$ in $T$
\begin{equation}
\cB\coloneqq\cH_C\otimes\operatorname{span}\left\{\ket{x}\middle|x\in\{0,1\}^r,\abs{x}\leq a_{r,k-c}\right\},
\end{equation}
we immediately have 
$\dim(\cB)=2^c\times V_r(a_{r,k-c})<2^k=\dim(\operatorname{im}(Q))$. Thus there exists a normalized $\ket{\xi}\in\operatorname{im}(Q)\cap\cB^\perp$, such that for all $\abs{x}\leq a_{r,k-c}$, $(\bI_C\otimes\bra{x}_T)\ket{\xi}=0$.
Consider the total number operator in the tail system $N_T=\sum_{j\in T}N_j$, then we have
\begin{equation}
\bra{\xi}N_T\ket{\xi}\geq a_{r,k-c}+1,
\end{equation}
since $\ket{\xi}$ has zero overlap with states in $\cB$. Combining with Eq.~\eqref{eq:xi_N_j_xi}, we have
\begin{equation}\label{eq:collective_occupation_bound}
\bra{\xi}N_T\ket{\xi}=\sum_{j\in T}\lambda_j\geq a_{r,k-c}+1.
\end{equation}
In this way, there exists $\ell\in T$, such that
\begin{equation}
1-p_\ell=\lambda_\ell\geq\frac{a_{r,k-c}+1}r=q_{r,k-c},
\end{equation}
which proves Eq.~\eqref{eq:vacuum_occupation_bound}.
\end{proof}

After identifying the mode $\ell$ chosen in Lemma~\ref{lem:vacuum_occupation_bound}, we then consider the outcomes when projecting the projector of a $[\![n,k,\dF]\!]_\mathrm{F}$ fermionic error-correcting code onto $\Pi_{\ell}$.

\begin{lemma}[Mode postselection]\label{lem:mode_postselection}
For any $[\![n,k,\dF]\!]_\mathrm{F}$ fermionic code with orthogonal projector $Q$ with $\dF\geq3$ and any mode $j\in[n]$, consider the projector $\Pi_j=\ket{0}_j\bra{0}_j$ onto the vacuum state on mode $j$, then the projection of $Q$ onto $\Pi_j$, i.e., $\Pi_jQ$, either:
\begin{enumerate}
\item annihilates the code projector $Q$, or
\item induces an isometry $W_j\propto\Pi_jQ$ on $\operatorname{im}(Q)$. We consider the isometry $J_j$ that maps from modes $[n]\setminus\{j\}$ to modes $[n]$ and inserts the vacuum state on mode $j$. Equivalently, $J_j$ and its adjoint act on the occupation number basis as
\begin{equation}\label{eq:def_J_j}
\begin{aligned}
J_j
\ket{x_1,\cdots,x_{j-1},x_{j+1},\cdots,x_n}
&=
\ket{x_1,\cdots,x_{j-1},0,x_{j+1},\cdots,x_n},\\
J_j^{\textnormal{\textdagger}}
\ket{x_1,\cdots,x_n}
&=\delta_{x_j,0}
\ket{x_1,\cdots,x_{j-1},x_{j+1},\cdots,x_n}.
\end{aligned}
\end{equation}
Together with $W_j$, they form a code projector $J_j^\textnormal{\textdagger} W_jW_j^\textnormal{\textdagger} J_j$ onto a $[\![n-1,k,\dF']\!]_\mathrm{F}$ fermionic code with $\dF'\geq\dF-2$, defined on the remaining $n-1$ modes after post-selecting $\ket{0}_j\bra{0}_j$ on mode $j$.
\end{enumerate}
\end{lemma}

\begin{proof}[Proof of Lemma~\ref{lem:mode_postselection}]
For any $j\in[n]$, Eq.~\eqref{eq:def_Pi_j_N_j} implies $\deg_\mathrm{F}(\Pi_j)=2$. Then $\dF\geq3$ implies
\begin{equation}
Q\Pi_jQ=p_jQ
\end{equation}
for some $p_j\in[0,1]$ since $\Pi_j$ is a Hermitian projector.
\begin{enumerate}
\item If $p_j=0$, then $\Pi_jQ=0$.
\item If $p_j>0$, define
\begin{equation}\label{eq:def_Q'}
W_j=\frac{\Pi_jQ}{\sqrt{p_j}},\qquad Q'=W_jW_j^\dagger=\frac{\Pi_jQ\Pi_j}{p_j}.
\end{equation}
Then $W_j$ is an isometry on $\operatorname{im}(Q)$ since $W_j^\dagger W_j=Q$, and $Q'$ is a rank-$2^k$ projector since $Q'^2=Q'$ and 
\begin{equation}
\Tr(Q')=\frac{\Tr(\Pi_jQ)}{p_j}=\frac{\Tr(Q\Pi_jQ)}{p_j}=\Tr(Q)=2^k.
\end{equation}
Notice that $\Pi_jQ'=Q'\Pi_j=Q'$, so $\operatorname{im}(Q')\subseteq\operatorname{im}(\Pi_j)$. Denote the remaining $n-1$ modes except mode $j$ as $\overline{j}=[n]\setminus\{j\}$ and consider the isometry $J_j$ that maps from modes $\overline{j}$ to modes $[n]$, acting as
\begin{equation}
\begin{aligned}
J_j
\ket{x_1,\cdots,x_{j-1},x_{j+1},\cdots,x_n}
&=
\ket{x_1,\cdots,x_{j-1},0,x_{j+1},\cdots,x_n},\\
J_j^\dagger
\ket{x_1,\cdots,x_n}
&=
\delta_{x_j,0}\ket{x_1,\cdots,x_{j-1},x_{j+1},\cdots,x_n}.
\end{aligned}
\end{equation}
It satisfies $J_j^\dagger J_j=\bI_{\overline{j}}$ and $J_jJ_j^\dagger=\Pi_j$. Define $Q''=J_j^\dagger Q'J_j=J_j^\dagger W_jW_j^\dagger J_j$, we then show that $Q''$ is the projector of a $[\![n-1,k,\dF']\!]_\mathrm{F}$ fermionic code with $\dF'\geq\dF-2$. It is indeed a projector defined on an $(n-1)$-mode fermionic system with rank $2^k$ since
\begin{equation}
Q''^2=Q'',\qquad
\Tr(Q'')=\Tr(Q'\Pi_j)=\Tr(Q')=2^k.
\end{equation}
Its Majorana distance is
\begin{equation}
\dF'=\dF(Q'')=\min
\left\{
\abs{S}\middle|\emptyset\neq S\subseteq[2n]\setminus\{2j-1,2j\},
Q''\overline\gamma_S Q''
\notin
\mathbb{C}Q''
\right\},
\end{equation}
where $\overline{\gamma}_S=J_j^\dagger\gamma_SJ_j$ is defined on modes $\overline{j}$. It is straightforward to verify $\{\overline\gamma_l\}_{l\in[2n]\setminus\{2j-1,2j\}}$ satisfy the canonical anticommutation relation and can be regarded as Majorana operators on modes $\overline{j}$.
For any $T\subseteq[2n]\setminus\{2j-1,2j\}$ with $\abs{T}<\dF-2$, we need to show $Q''\overline\gamma_{T}Q''\in\mathbb{C}Q''$. Notice that $\Pi_j$ defined in Eq.~\eqref{eq:def_Pi_j_N_j} is an even operator and its support has no overlap with that of $T$, we have $[\Pi_j,\gamma_T]=0$, so 
that
\begin{equation}
\deg_\mathrm{F}(\Pi_j\gamma_T\Pi_j)=\deg_\mathrm{F}(\gamma_T\Pi_j)=\deg_\mathrm{F}\left(\frac{\gamma_T}2-\frac{i\gamma_T\gamma_{2j-1}\gamma_{2j}}2\right)\leq\abs{T}+2<\dF.
\end{equation}
So there exists $\lambda_T\in\mathbb{C}$, such that
$Q\Pi_j\gamma_T\Pi_jQ=\lambda_TQ$,
thus
$p_j^2Q'\gamma_TQ'=\Pi_jQ\Pi_j\gamma_T\Pi_jQ\Pi_j=\lambda_Tp_jQ'$,
and
\begin{equation}
Q'\gamma_TQ'=\frac{\lambda_T}{p_j}Q'.
\end{equation}
Finally,
\begin{equation}
\begin{aligned}
Q''\overline\gamma_{T}Q''&=J_j^\dagger Q'J_jJ_j^\dagger\gamma_TJ_jJ_j^\dagger Q'J_j
\\&=J_j^\dagger Q'\Pi_j\gamma_T\Pi_j Q'J_j
\\&=J_j^\dagger Q'\gamma_TQ'J_j
\\&=\frac{\lambda_T}{p_j}J_j^\dagger Q'J_j
\\&=\frac{\lambda_T}{p_j}Q'',
\end{aligned}
\end{equation}
which finishes the proof.
\end{enumerate}
\end{proof}

\subsection{Proof of Theorem~\ref{thm:universal_non_Gaussian_cost}}\label{app:g-doped_Gaussian_fidelity}
We now present the proof of Theorem~\ref{thm:universal_non_Gaussian_cost} in the main text.
\begin{proof}[Proof of Theorem~\ref{thm:universal_non_Gaussian_cost}]
We denote
\begin{equation}
\begin{gathered}
r_{\dF,q}\coloneqq\left\lfloor\frac{d_{\mathrm F}-3}{q-2}\right\rfloor,\quad
s_{g,q}\coloneqq q\max\{g-r_{\dF,q},0\},\quad
c_{g,q}\coloneqq\min\left\{k,s_{g,q}\right\},\\
d_{g,q}\coloneqq\dF-(q-2)\min\left\{g,r_{\dF,q}\right\},\qquad
t_{g,q}\coloneqq\left\lfloor\frac{d_{g,q}-1}2\right\rfloor.
\end{gathered}
\end{equation}
By definition we have $d_{g,q}\geq3$ and $t_{g,q}\geq1$.
We first prove Eq.~\eqref{eq:unified_gq_fidelity_bound} in the main text,
\begin{equation}\label{eq:unified_gq_fidelity_bound_app}
F_{g,q}^{(n)}(\mathcal C)
\leq
\left[
1-H_2^{-1}\left(
\frac{ k-c_{g,q} }{n-c_{g,q}}
\right)
\right]^{t_{g,q}}.
\end{equation}
If $s_{g,q}\geq k$, then $c_{g,q}=k$ and Eq.~\eqref{eq:unified_gq_fidelity_bound_app} holds trivially. So we focus on $s_{g,q}<k$ and thus $c_{g,q}=s_{g,q}$. 

Fix an arbitrary codeword $\ket{\psi}\in\operatorname{im}(P)$. We calculate its fidelity with a $(g,q)$-doped state, i.e., $\abs{\bra{0^n}U^\dagger\ket{\psi}}^2$, for a generic $(g,q)$-doped circuit $U$ in the form of Eq.~\eqref{eq:g_doped_circuit_app}. We decompose $U$ into two parts $U=U_1U_2$, where $U_1$ contains $\min\{g,r_{\dF,q}\}$ non-Gaussian gates together with their adjacent Gaussian gates, and $U_2$ contains the remaining $\max\{0,g-r_{\dF,q}\}$ non-Gaussian gates and adjacent Gaussian gates, i.e.,
\begin{equation}
\begin{aligned}
U_1&=G_g\prod_{j=g}^{\max\{0,g-r_{\dF,q}\}+1}\left(V_jG_{j-1}\right),\\
U_2&=\prod_{j=\max\{0,g-r_{\dF,q}\}}^1\left(V_jG_{j-1}\right),
\end{aligned}
\end{equation}
where the index $j$ in the products takes a decreasing order, and the empty product is the identity. Define $P_1=U_1^\dagger PU_1$ and use Lemma~\ref{lem:distance_degradation}, then $P_1$ is the code projector of an $[\![n,k,d']\!]_\mathrm{F}$ fermionic code with
\begin{equation}
d'\geq\dF-\min\left\{g,r_{\dF,q}\right\}(q-2)=d_{g,q}.
\end{equation}
We then consider $U_2\ket{0^n}$, it is a $(\max\{0,g-r_{\dF,q}\},q)$-doped state. Applying Lemma~\ref{lem:fermionic_non-Gaussianity_compression}, there exist a Gaussian unitary $G$ and a state $\ket{\chi}$ on the first $\max\{0,g-r_{\dF,q}\}q=s_{g,q}=c_{g,q}$ modes, such that
\begin{equation}
U_2\ket{0^n}=G\left(\ket{\chi}_C\otimes\ket{0^{n-c_{g,q}}}_T\right),
\end{equation}
where $C$, as a collection of mode indices, is the core system on the first $c_{g,q}$ modes and $T$ is the tail system on the remaining $n-c_{g,q}$ modes. 

Define 
\begin{equation}
Q=G^\dagger P_1G=G^\dagger U_1^\dagger PU_1G,\qquad
\ket{\zeta}=G^\dagger U_1^\dagger\ket{\psi}.
\end{equation}
Using Lemma~\ref{lem:distance_degradation} with $g=0$, we see that $Q$ is the orthogonal projector of an $[\![n,k,d']\!]_{\textup{F}}$ code and $\ket{\zeta}$ is a codeword in $\operatorname{im}(Q)$.
In this way, the fidelity between $\ket{\psi}$ and the $(g,q)$-doped state $U\ket{0^n}$ can be written as
\begin{equation}\label{eq:fidelity_rewrite}
\abs{\bra{0^n}U^\dagger\ket{\psi}}^2=\abs{\bra{0^n}U_2^\dagger U_1^\dagger\ket{\psi}}^2=\abs{\left(\bra{\chi}_C\otimes\bra{0^{n-c_{g,q}}}_T\right)\ket{\zeta}}^2,
\end{equation}

Let
\begin{equation}
R_0\coloneqq[n],
\qquad
r_0\coloneqq\abs{R_0}=n,
\qquad
T_0\coloneqq T,
\qquad
d_0\coloneqq d',\qquad
Q_0\coloneqq Q,\qquad
\ket{\psi_0}\coloneqq\ket{\zeta}.
\end{equation}
We recursively construct distinct mode labels
$\left\{\ell_j\right\}$ in the tail system $T$,
code projectors
$\left\{Q_j\right\}$,
and codewords
$\left\{\ket{\psi_j}\right\}$ from the following steps. Suppose for $1\leq j\leq t_{g,q}$ and before step $j$, we have a code projector $Q_{j-1}$ on the remaining mode set $R_{j-1}$ with $r_{j-1}=\abs{R_{j-1}}=n-j+1$, such that $\operatorname{rank}(Q_{j-1})=2^k$ and $d_{j-1}=\dF(Q_{j-1})\geq d'-2j+2$, and a codeword $\ket{\psi_{j-1}}\in\operatorname{im}(Q_{j-1})$. We denote $T_j=R_j\cap T$. The initial construction for $j=1$ is provided above. We show that the above properties also hold after step $j$.

For $1\leq j\leq t_{g,q}$, since $d_{j-1}\geq d'-2j+2\geq d'-2t_{g,q}+2\geq d_{g,q}-2t_{g,q}+2\geq3$, we apply Lemma~\ref{lem:vacuum_occupation_bound} to $Q_{j-1}$, core system $C$ and the remaining tail system $T_{j-1}$. We determine a mode $\ell_j\in T_{j-1}$ and a scalar $p_j$, independent of the codeword in $\operatorname{im}(Q_{j-1})$, such that
\begin{equation}\label{eq:Gaussian_fidelity_step_probability}
p_j
=
\bra{\psi_{j-1}}
\Pi_{\ell_j}
\ket{\psi_{j-1}}
\leq
1-q_{r_{j-1}-c_{g,q},k-c_{g,q}},
\end{equation}
where $\Pi_{\ell_j}=\ket{0}_{\ell_j}\bra{0}_{\ell_j}$ and $q_{r_{j-1}-c_{g,q},k-c_{g,q}}$ is defined in Eq.~\eqref{eq:def_q_nk}.
Using Lemma~\ref{lem:lower_bound_q_nk} and noticing that the function $H_2^{-1}$ is increasing on its domain, we have

\begin{equation}\label{eq:Gaussian_fidelity_uniform_step_bound}
p_j
<
1-
H_2^{-1}
\left(
\frac{k-c_{g,q}}{r_{j-1}-c_{g,q}}
\right)
\leq
1-
H_2^{-1}
\left(
\frac{k-c_{g,q}}{n-c_{g,q}}
\right).
\end{equation}

We then apply Lemma~\ref{lem:mode_postselection} to $Q_{j-1}$ and mode $\ell_j\in T_{j-1}$. Then there are two cases.
\begin{enumerate}
\item If $\Pi_{\ell_j}Q_{j-1}=0$, then
\begin{equation}
\left(\bra{\chi}_C\otimes\bra{0^{r_{j-1}-c_{g,q}}}_{T_{j-1}}\right)\ket{\psi_{j-1}}=\left(\bra{\chi}_C\otimes\bra{0^{r_{j-1}-c_{g,q}}}_{T_{j-1}}\right)\Pi_{\ell_j}\ket{\psi_{j-1}}=0.
\end{equation}
We terminate the series of construction at step $j$.
\item If $\Pi_{\ell_j}Q_{j-1}\neq0$, we obtain the isometries $W_{\ell_j}$, $J_{\ell_j}$, and the code projector $Q_j=J_{\ell_j}^\dagger W_{\ell_j}W_{\ell_j}^\dagger J_{\ell_j}$ on the remaining mode set $R_j=R_{j-1}\setminus\{\ell_j\}$ with $r_j=\abs{R_j}=n-j$, such that $\operatorname{rank}(Q_j)=2^k$ and $d_j=\dF(Q_j)\geq d_{j-1}-2\geq d'-2j$. We denote $T_j=R_j\cap T$. Recall the isometry $J_{\ell_j}$ in Eq.~\eqref{eq:def_J_j}, then $J_{\ell_j}\left(\ket{\chi}_C\otimes\ket{0^{r_j-c_{g,q}}}_{T_j}\right)=\left(\ket{\chi}_C\otimes\ket{0^{r_{j-1}-c_{g,q}}}_{T_{j-1}}\right)_{R_{j-1}}$. Let $\ket{\psi_j}\in\operatorname{im}(Q_j)$ be the normalized post-selected codeword from $\ket{\psi_{j-1}}$, defined by
\begin{equation}\label{eq:Gaussian_fidelity_post-selected_state}
\ket{\psi_j}=
J_{\ell_j}^\dagger W_{\ell_j}\ket{\psi_{j-1}}
=
J_{\ell_j}^\dagger\frac{
\Pi_{\ell_j}\ket{\psi_{j-1}}
}{
\sqrt{p_j}
}.
\end{equation}
Then
\begin{equation}\label{eq:Gaussian_fidelity_overlap_recursion}
\begin{aligned}
\left(\bra{\chi}_C\otimes\bra{0^{r_{j-1}-c_{g,q}}}_{T_{j-1}}\right)
\ket{\psi_{j-1}}
&=
\left(\bra{\chi}_C\otimes\bra{0^{r_{j-1}-c_{g,q}}}_{T_{j-1}}\right)
\Pi_{\ell_j}
\ket{\psi_{j-1}}
\\
&=
\sqrt{p_j}
\left(\bra{\chi}_C\otimes\bra{0^{r_{j-1}-c_{g,q}}}_{T_{j-1}}\right)
J_{\ell_j}
\ket{\psi_j}
\\
&=
\sqrt{p_j}
\left(\bra{\chi}_C\otimes\bra{0^{r_j-c_{g,q}}}_{T_j}\right)
\ket{\psi_j}.
\end{aligned}
\end{equation}
This completes the recursive construction at step $j$.
\end{enumerate}
If the recursive construction runs into the first case above at some step, then we immediately get $\left(\bra{\chi}_C\otimes\bra{0^{n-c_{g,q}}}_T\right)\ket{\psi_0}=0$ from Eq.~\eqref{eq:Gaussian_fidelity_overlap_recursion}. Otherwise, we can construct $t_{g,q}=\lfloor(d_{g,q}-1)/2\rfloor$ steps and obtain
\begin{equation}
\left(\bra{\chi}_C\otimes\bra{0^{n-c_{g,q}}}_T\right)\ket{\psi_0}=\left(
\prod_{j=1}^{t_{g,q}}
\sqrt{p_j}
\right)
\left(\bra{\chi}_C\otimes\bra{0^{r_{t_{g,q}}-c_{g,q}}}_{T_{t_{g,q}}}\right)
\ket{\psi_{t_{g,q}}}.
\end{equation}
Using Eq.~\eqref{eq:fidelity_rewrite} and Eq.~\eqref{eq:Gaussian_fidelity_uniform_step_bound}, we have
\begin{equation}
\abs{\bra{0^n}U^\dagger\ket{\psi}}^2=\abs{\left(\bra{\chi}_C\otimes\bra{0^{n-c_{g,q}}}_T\right)\ket{\psi_0}}^2\leq
\prod_{j=1}^{t_{g,q}}
p_j<
\left[
1-
H_2^{-1}\left(\frac{k-c_{g,q}}{n-c_{g,q}}\right)
\right]^{t_{g,q}}.
\end{equation}
Since $\ket{\psi}\in\operatorname{im}(P)$ and $(g,q)$-doped circuit $U$ are arbitrary, we prove Eq.~\eqref{eq:unified_gq_fidelity_bound_app}.

Assume a $(g,q)$-doped circuit prepares a codeword exactly, then
\begin{equation}
1=F_{g,q}^{(n)}(\mathcal C)
\leq
\left[
1-H_2^{-1}\left(
\frac{ k-c_{g,q} }{n-c_{g,q}}
\right)
\right]^{t_{g,q}}.
\end{equation}
Notice that $t_{g,q}\geq1$, we have $0<k=c_{g,q}\leq s_{g,q}=q(g-r_{\dF,q})$, $g$ is an integer implies
\begin{equation}
g\geq r_{\dF,q}+\left\lceil\frac kq\right\rceil=\left\lfloor\frac{\dF-3}{q-2}\right\rfloor+\left\lceil\frac kq\right\rceil,
\end{equation}
which is Eq.~\eqref{eq:exact_non_Gaussian_cost_from_fidelity} in the main text.
\end{proof}

\section{Fermionic entanglement distillation with Gaussian channels}
\subsection{Proof of Theorem~\ref{thm:entanglement_distillation}}
\label{app:proof_entanglement_distillation}
In this section, we provide the proof of Theorem~\ref{thm:entanglement_distillation}.

\begin{proof}[Proof of Theorem~\ref{thm:entanglement_distillation}]
We first prove Eq.~\eqref{eq:Gaussian_entanglemenet_distillation_upper_bound}. Assume that the input consists of $m$ copies of $\ketbra{\psi_\theta}{\psi_\theta}$, namely
\begin{equation}
\rho_{\mathrm{in}}
=
\left(
\ketbra{\psi_\theta}{\psi_\theta}
\right)^{\otimes_{\mathrm F}m}.
\end{equation}
For each copy $\ell\in[m]$, let
$\alpha_{2\ell-1},\alpha_{2\ell}$ denote the two Majorana operators of $A$'s mode and
$\beta_{2\ell-1},\beta_{2\ell}$ those of $B$'s mode. Choosing the global Majorana ordering in which all of $A$'s modes precede all of $B$'s modes, the 
covariance matrix of $\rho_{\mathrm{in}}$ takes the form
\begin{equation}
\Gamma(\rho_{\mathrm{in}})=
\begin{bmatrix}
\Gamma_A&\mathsf{C}\\
-\mathsf{C}^\mathsf{T}&\Gamma_B
\end{bmatrix},
\end{equation}
where for $j,k\in[2m]$,
\begin{equation}\label{eq:Gamma_A_Gamma_B_C}
\begin{aligned}
(\Gamma_A)_{j,k}&=-\frac{i}2\Tr([\alpha_j,\alpha_k]\rho_{\mathrm{in}}),\\
(\Gamma_B)_{j,k}&=-\frac{i}2\Tr([\beta_j,\beta_k]\rho_{\mathrm{in}}),\\
\mathsf{C}_{j,k}&=-\frac{i}2\Tr([\alpha_j,\beta_k]\rho_{\mathrm{in}})=-i\Tr(\alpha_j\beta_k\rho_{\mathrm{in}}).
\end{aligned}
\end{equation}
Since each $\ketbra{\psi_\theta}{\psi_\theta}$ is an even operator, all of its odd Majorana expectations vanish, so the two-point correlation between Majorana operators from distinct copies is zero. So the matrix $\mathsf{C}$ factorizes as the direct sum of $m$ copies of the matrix $\mathsf{C}_\theta$ for the single-copy case, i.e.,
\begin{equation}
\mathsf{C}=\bigoplus_{\ell=1}^m\mathsf{C}_\theta,
\qquad\text{with }
\mathsf{C}_\theta=
\begin{bmatrix}
0&\sin(2\theta)\\
\sin(2\theta)&0
\end{bmatrix},
\end{equation}
where we use $\bra{\psi_\theta}X_AX_B\ket{\psi_\theta}=-\bra{\psi_\theta}Y_AY_B\ket{\psi_\theta}=\sin(2\theta)$.
So $\norm{\mathsf{C}}_\infty=\norm{\mathsf{C}_\theta}_\infty=\sin(2\theta)$. We then consider the action of the local Gaussian channel $\Lambda_A\otimes_\mathrm{F}\Lambda_B$ on the covariance matrix. According to Lemma~\ref{lem:Gaussian_channel_covariance_matrix}, there exist matrices $\mathsf{M}_A,\mathsf{M}_B\in\mathbb{R}^{2\times2m}$ and antisymmetric matrices $\mathsf{L}_A,\mathsf{L}_B\in\mathbb{R}^{2\times2}$, such that for the output state $\rho_{\mathrm{out}}=(\Lambda_A\otimes_\mathrm{F}\Lambda_B)(\rho_\mathrm{in})$,
\begin{equation}
\begin{aligned}
\Gamma(\rho_{\mathrm{out}})&=\left(\mathsf{M}_A\oplus\mathsf{M}_B\right)\Gamma(\rho_{\mathrm{in}})\left(\mathsf{M}_A^\mathsf{T}\oplus\mathsf{M}_B^\mathsf{T}\right)+\left(\mathsf{L}_A\oplus\mathsf{L}_B\right)
\\&=
\begin{bmatrix}
\mathsf{M}_A\Gamma_A\mathsf{M}_A^\mathsf{T}+\mathsf{L}_A&\mathsf{M}_A\mathsf{C}\mathsf{M}_B^\mathsf{T}\\
-\mathsf{M}_B\mathsf{C}^\mathsf{T}\mathsf{M}_A^\mathsf{T}&\mathsf{M}_B\Gamma_B\mathsf{M}_B^\mathsf{T}+\mathsf{L}_B
\end{bmatrix}.
\end{aligned}
\end{equation}
We denote the cross-correlation block as $\mathsf{C}_\mathrm{out}=\mathsf{M}_A\mathsf{C}\mathsf{M}_B^\mathsf{T}$, by Lemma~\ref{lem:Gaussian_channel_covariance_matrix}, $\norm{\mathsf{M}_A}_\infty,\norm{\mathsf{M}_B}_\infty\leq1$, so that
\begin{equation}
\norm{\mathsf{C}_\mathrm{out}}_\infty\leq
\norm{\mathsf{M}_A}_\infty\norm{\mathsf{C}}_\infty\norm{\mathsf{M}_B}_\infty
\leq\sin(2\theta).
\end{equation}
Notice that $\rho_{\mathrm{out}}$ is a bipartite two-mode state, and its fidelity with the EPR state
\begin{equation}
\begin{aligned}
\ketbra{\Phi^+}{\Phi^+}&=\frac14(\bI+X_AX_B-Y_AY_B+Z_AZ_B)
\end{aligned}
\end{equation}
can be upper bounded by
\begin{equation}
\begin{aligned}
F_{\mathrm{EPR}}(\rho_{\mathrm{out}})&=\frac14\left(1+(\mathsf{C}_\mathrm{out})_{2,1}+(\mathsf{C}_\mathrm{out})_{1,2}+\Tr(Z_AZ_B\rho_{\mathrm{out}})\right)\\
&\leq
\frac{1+\norm{\mathsf{C}_\mathrm{out}}_\infty}2
\\&\leq\frac{1+\sin(2\theta)}2
\\&=F_{\mathrm{EPR}}(\ketbra{\psi_\theta}{\psi_\theta}),
\end{aligned}
\end{equation}
where we use $\Tr(Z_AZ_B\rho_{\mathrm{out}})\leq1$ and relate the expectation values to the elements of $\mathsf{C}_\mathrm{out}$ defined in Eq.~\eqref{eq:Gamma_A_Gamma_B_C}. The statement can be extended to the case where classical randomness is used to implement a channel
\begin{equation}
\cN
=
\sum_r
p_r
\left(
\Lambda_A^{(r)}
\otimes_{\mathrm F}
\Lambda_B^{(r)}
\right),
\qquad
p_r\geq0,
\qquad
\sum_rp_r=1.
\end{equation}
Indeed, the EPR-pair fidelity is linear in the output state, and
each product channel in the mixture satisfies the bound proved above. Thus we prove Eq.~\eqref{eq:Gaussian_entanglemenet_distillation_upper_bound} in the main text.

We then prove the statement about local Pauli measurement. Consider a local fermionic measurement instrument by subsystem $A$: for $0<t\leq1$, let
\begin{equation}
\begin{aligned}
\cM_t^{(\mathrm{succ})}(\rho)&=K_t\rho K_t^\dagger,\qquad&&\text{with }
K_t=t\ketbra{0}{0}+\ketbra{1}{1},\\
\cM_t^{(\mathrm{fail})}(\rho)&=L_t\rho L_t^\dagger,&&\text{with }
L_t=\sqrt{1-t^2}\ketbra{1}{0}.
\end{aligned}
\end{equation}
Since $K_t^\dagger K_t+L_t^\dagger L_t=\bI$, $[K_t,\Pi]=0$, and $\{L_t,\Pi\}=0$, it is straightforward to verify $\cM_t^{(\mathrm{succ})}$ and $\cM_t^{(\mathrm{fail})}$ are both CP and trace-non-increasing parity-preserving maps, and they form a fermionic measurement. To show they are Gaussian, we write down their Choi operators defined in Eq.~\eqref{eq:fermionic-Choi-operator}
\begin{equation}
\begin{aligned}
J_\mathrm{F}\left(\cM_t^{(\mathrm{succ})}\right)&=\frac12\left(t\ket{0,0}+\ket{1,1}\right)\left(t\bra{0,0}+\bra{1,1}\right),\\
J_\mathrm{F}\left(\cM_t^{(\mathrm{fail})}\right)&=\frac{1-t^2}2\ketbra{1,0}{1,0}.
\end{aligned}
\end{equation}
The normalized operator $\ketbra{\phi_t}{\phi_t}=J_\mathrm{F}\left(\cM_t^{(\mathrm{succ})}\right)/\Tr(J_\mathrm{F}\left(\cM_t^{(\mathrm{succ})}\right))$ is indeed a pure Gaussian state since its covariance matrix
\begin{equation}
\Gamma\left(\ketbra{\phi_t}{\phi_t}\right)=
\begin{bmatrix}
0&c_t&0&s_t\\
-c_t&0&s_t&0\\
0&-s_t&0&c_t\\
-s_t&0&-c_t&0
\end{bmatrix},
\end{equation}
with $c_t=(t^2-1)/(t^2+1)$ and $s_t=2t/(t^2+1)$, satisfies $\Gamma(\ketbra{\phi_t}{\phi_t})^2=-\bI$. Also, $\ketbra{1,0}{1,0}$ is Gaussian. 
Then for any $\theta\in(0,\pi/4]$, let $t=\tan\theta$, so
\begin{equation}
\left(\cM_t^{(\mathrm{succ})}\otimes_\mathrm{F}\operatorname{id}\right)\left(\ketbra{\psi_\theta}{\psi_\theta}\right)=\sin^2\theta\left(\ket{0,0}+\ket{1,1}\right)\left(\bra{0,0}+\bra{1,1}\right)=2\sin^2\theta\ketbra{\Phi^+}{\Phi^+}.
\end{equation}
That is, after performing a local Gaussian measurement $\left\{\cM_t^{(\mathrm{succ})},\cM_t^{(\mathrm{fail})}\right\}$, with success probability $p_\mathrm{succ}=2\sin^2\theta$, subsystem $A$ can obtain the successful outcome. After subsystem $ B$ receives the one-bit success information via classical communication, they retain the EPR state $\ketbra{\Phi^+}{\Phi^+}$.
\end{proof}

\subsection{Schur complement formulas for conditioned Gaussian projections}\label{app:Schur_complement_formula}

To put the discussion about fermionic measurements in Appendix~\ref{app:fermionic_Gaussian_measurements} on the same footing as common bosonic arguments \cite{eisert2002distilling,Giedke2002distillation,Fiurasek2002Gaussian}, we discuss selective CP maps in a language of covariance matrices.
In this light, 
it is instructive to compare the conditional update rules for
fermionic and bosonic Gaussian states. Consider first a fermionic
Gaussian state with covariance matrix
\begin{equation}
 \Gamma =
 \begin{bmatrix}
  \Gamma_A & \mathsf{C} \\
  -\mathsf{C}^{\mathsf T} & \Gamma_B
 \end{bmatrix},
\end{equation}
where the second subsystem is projected onto a pure fermionic Gaussian
state with covariance matrix $\Gamma'$. Conditioned on this outcome,
the retained subsystem has covariance matrix
\begin{equation}
 \Gamma_{\mathrm{cond}}
 =
 \Gamma_A + \mathsf{C}(\Gamma_B+\Gamma')^{-1}\mathsf{C}^{\mathsf T}
 ,\label{eq:fermionic-schur-complement}
\end{equation}
that takes the form of a Schur complement.
For fermion-number measurements on a single mode, one has
$\Gamma'=\pm J$, where
\begin{equation}
 J =
 \begin{bmatrix}
  0 & 1\\
  -1 & 0
 \end{bmatrix},
\end{equation}
with the sign specifying the occupation outcome
(which is implied but not spelled out in Ref.\
\cite{bravyi2005LagrangianRepresentationFermionica}).

By comparison, let a bosonic Gaussian state have covariance matrix
\begin{equation}
 V =
 \begin{bmatrix}
  V_A & \mathsf{C}\\
  \mathsf{C}^{\mathsf T} & V_B
 \end{bmatrix}.
\end{equation}
Upon projection of the second subsystem onto the vacuum, whose
covariance matrix is the identity $\bI$ in our convention, the conditional covariance
matrix is
\begin{equation}
 V_{\mathrm{cond}}
 =
 V_A-\mathsf{C}(V_B+\bI)^{-1}\mathsf{C}^{\mathsf T}.
 \label{eq:bosonic-schur-complement}
\end{equation}
Thus, in both cases Gaussian postselection is described by a Schur
complement \cite{eisert2002distilling}, 
but with the characteristic sign difference originating from the antisymmetric fermionic and symmetric bosonic covariance
matrices. That said, while such selective operations do not allow for distilling entanglement of bosonic Gaussian states, the same cannot be straightforwardly said for fermionic states.

\end{document}